\documentclass[a4wide,10pt]{article}   	
\usepackage{epigraph}

\usepackage{stmaryrd}
\usepackage{thmtools,thm-restate}

\usepackage{amsmath}
\usepackage{amsthm}
\usepackage{amssymb}
\usepackage{amsfonts}
\usepackage{bm,bbold}
\usepackage{cancel}
\usepackage[dvipsnames]{xcolor}

\usepackage{chronology}

\usepackage[scr=boondoxo,scrscaled=1.05]{mathalfa}

\usepackage{bussproofs}
\EnableBpAbbreviations

\usepackage{graphicx}
\usepackage{adjustbox}

\usepackage{listings}

\newtheorem{notation}{Notation}[section]
\newtheorem{remark}{Remark}[section]
\newtheorem{theorem}{Theorem}[section]

\newtheorem{lemma}[theorem]{Lemma}

\newtheorem{cor}{Corollary}[section]
\newtheorem{prop}[theorem]{Proposition}

 \theoremstyle{definition}

\newtheorem{ex}{Example}[section]
\newtheorem{defn}{Definition}[section]

\newcommand{\Cyls}{\mathscr{C}}

\newcommand{\fone}{F}
\newcommand{\ftwo}{G}
\newcommand{\fthree}{H}

\newcommand{\midd}{\; \; \mbox{\Large{$\mid$}}\;\;}

\newcommand{\Nat}{\mathbb N }

\newcommand{\Bool}{\mathbb B}
\newcommand{\Ss}{\mathbb S}
\newcommand{\Os}{\mathbb O}

\newcommand{\cc}[1]{\mathbf{#1}}

\newcommand{\lc}[1]{\textbf{\textsf{#1}}}

\newcommand{\MQPA}{\lc{MQPA}}

\newcommand{\PA}{\lc{PA}}

\newcommand{\NP}{\cc{NP}}

\newcommand{\RFP}{\cc{RFP}}

\newcommand{\one}{\mathbb{1}}
\newcommand{\zero}{\mathbb{0}}

\newcommand{\model}[1]{\llbracket #1\rrbracket}

\makeatletter
\providecommand*{\Dashv}{%
  \mathrel{%
    \mathpalette\@Dashv\vDash
  }%
}
\newcommand*{\@Dashv}[2]{%
  \reflectbox{$\m@th#1#2$}%
}
\makeatother

\newcommand{\POR}{\mathcal{POR}}
\newcommand{\SFP}{\cc{SFP}}
\newcommand{\PTF}{\cc{PTF}}

\newcommand{\rl}{\mathcal{RL}}

\newcommand{\Flip}{\mathtt{Flip}}

\newcommand{\eepsilon}{\bm{\epsilon}}
\newcommand{\ssubseteq}{\bm{\subseteq}}
\newcommand{\cconc}{\bm{\frown}}
\newcommand{\ttimes}{\bm{\times}}
\newcommand{\bool}{\mathbb{b}}

\newcommand{\Ef}{E}
\newcommand{\Pf}{P}
\newcommand{\Sf}{S}
\newcommand{\Cf}{C}
\newcommand{\query}{Q}

\newcommand{\SIFPra}{\text{SIFP}_{\text{RA}}}
\newcommand{\SIFPla}{\text{SIFP}_{\text{LA}}}
\newcommand{\LSra}{\mathcal{L}(\prog{Stm}_{\text{RA}})}
\newcommand{\LSla}{\mathcal{L}(\prog{Stm}_{\text{LA}})}
\newcommand{\SIFP}{\text{SIFP}}

\newcommand{\toLA}{\leadsto_{\text{LA}}}
\newcommand{\toRA}{\leadsto_{\text{RA}}}

\newcommand{\prog}[1]{\textsf{#1}}
\newcommand{\progZ}{0}

\newcommand{\progE}{\epsilon}

\newcommand{\zzero}{\mathtt{0}}
\newcommand{\oone}{\mathtt{1}}
\newcommand{\conc}{\frown}
\newcommand{\bbool}{\mathtt{b}}

\newcommand{\PORl}{\POR^\lambda}
\newcommand{\IPOR}{\mathcal{I}\PORl}

\newcommand{\arrowT}{\Rightarrow}

\newcommand{\Eps}{\mathsf{Eps}}

\newcommand{\ooverline}[1]{\overline{\overline{#1}}}

\newcommand{\FV}{FV}

\newcommand{\realize}{\ \circledR\ }

\newcommand{\PIND}{\text{PIND}}

\newcommand{\termO}{\mathbf{t}}
\newcommand{\termT}{\mathbf{u}}
\newcommand{\varO}{\mathbf{x}}
\newcommand{\varT}{\mathbf{y}}

\newcommand{\df}{:=}

\newcommand{\termOne}{\term{t}}
\newcommand{\termTwo}{\term{u}}
\newcommand{\termThree}{\term{v}}

\newcommand{\typeOne}{\sigma}

\newcommand{\term}[1]{\mathsf{#1}}

\newcommand{\Chole}{\textsf{C}}

\newcommand{\EM}{\mathbf{EM}}

\newcommand{\PTM}{\mathscr{M_P}}
\newcommand{\STM}{\mathscr{M_S}}

\newcommand{\Qstates}{\mathbf{Q}}
\newcommand{\blank}{\circledast}

\newcommand{\strTT}{\gamma}

\newcommand{\longv}[1]{}

\newcommand{\cyl}{\textsf{C}}

\newcommand{\dist}{\mathbb{D}}

\newcommand{\PR}{\mathcal{PR}}

\newcommand{\RS}{\textsf{\textbf{RS}}^1_2}
\newcommand{\BussS}{\textsf{\textbf{S}}^1_2}

\newcommand{\PTCA}{\mathcal{PTCA}}

\newcommand{\tailF}{Tail}
\newcommand{\eqF}{Eq}
\newcommand{\prF}{Pr}

\newcommand{\str}{\sigma}
\newcommand{\strT}{\tau}

\newcommand{\PV}{\textsf{PV}^\omega}
\newcommand{\IPV}{\textsf{IPV}^\omega}
\newcommand{\IRS}{\textsf{\textbf{I}}\RS}

\newcommand{\TSet}{\emph{T}}

\newcommand{\Markov}{\mathbf{(Markov)}}

\newcommand{\TTop}{\underline{\mathbf{1}}}

\usepackage{hyperref}

\title{\textbf{An Arithmetic Theory for the Poly-Time Random Functions}}
\author{M. Antonelli, U. Dal Lago, D. Davoli, I. Oitavem, P. Pistone}
\date{\today}							

\begin{document}
\maketitle


\newcommand{\queryS}{\query}
\newcommand{\emptyS}{\Ef}
\newcommand{\projS}{\Pf}
\newcommand{\successorS}{\Sf}
\newcommand{\condS}{\Cf}

We introduce a new bounded theory $\RS$,
and show that the functions which
are $\Sigma^b_1$-representable in it are precisely
random functions which can be computed
in polynomial time.
Concretely, we pass through the class of oracle functions
over strings $\POR$, 
which is introduced in Section~\ref{sec:introRS},
together with $\RS$.
Then, we show that functions computed by poly-time PTMs
are \emph{arithmetically} characterized by a class of 
probabilistic bounded formulas:
in Section~\ref{sec:RStoPOR}, 
we prove that the class of poly-time oracle functions 
is equivalent to that of functions which are $\Sigma^b_1$-representable
in $\RS$, and 
in Section~\ref{sec:PORandPTM} we establish the converse.

\section{Overview}

%
Usual characterizations of poly-time (deterministic)
functions in bounded arithmetic are obtained by two ``macro''
results~\cite{Buss86,Ferreira90}.
Some Cobham-style algebra for poly-time functions is 
introduced and shown equivalent to (1) that
of functions computed by TMs running in polynomial time,
and (2) that of functions which are $\Sigma^b_1$-representable
in the proper bounded theory.
The global structure of our proof follows a similar path,
with an algebra of oracle recursive function,
called $\POR$, playing the role of our Cobham-style
function algebra.
In our case, functions are poly-time computable by PTMs 
and the theory is randomized $\RS$.
After introducing these classes, we show
that the random functions which are 
$\Sigma^b_1$-representable in $\RS$ are precisely those in
$\POR$, and 
that $\POR$ is equivalent
(in a very specific sense)
to the class of functions computed by PTMs running
in polynomial time.
While the first part is established due to standard
arguments~\cite{Ferreira90,CookUrquhart},
the presence of randomness introduced a 
delicate ingredient to be dealt with in the second part.
Indeed, functions in $\POR$ access randomness
in a rather different way
with respect to PTMs, and relating these  
models requires some effort,
that involves long chains of 
intermediate simulations.

 \small
 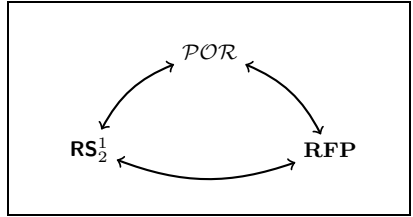
\begin{figure}[h!]\label{fig:schema2}
\begin{center}
\framebox{
\parbox[t][2.6cm]{5cm}{
\footnotesize{

 \begin{center}
 \begin{tikzpicture}[node distance=2cm]
\node at (0,2) (a) {$\POR$}; 
\node at (-1.6,0.7) (b) {$\RS$}; 
\node at (1.6,0.7) (c) {$\RFP$};

    \draw[<->,thick] (a) to [bend right=20] (b);
    \draw[<->,thick] (b) to [bend right=20] (c);
    \draw[<->,thick] (c) to [bend right=20] (a);

   %
\end{tikzpicture} 

 \end{center}
  }}}
\caption{Proof Schema}
\end{center}
\end{figure}  
\normalsize

Concretely, 
we start by defining the class of oracle functions over strings,
the new theory $\RS$, strongly inspired
by~\cite{FerreiraOitavem}, but over a ``probabilistic word language'',
and considering a slightly modified notion of 
$\Sigma^b_i$-representability,
fitting the domain of our peculiar oracle functions.
Then, we prove that the class of
random functions computable in polynomial time, called $\RFP$, 
is precisely the class of
functions which are $\Sigma^b_1$-representable in $\RS$
in three steps:
\begin{enumerate}
\itemsep0em

\item We prove that functions in $\POR$ are $\Sigma^b_1$-representable
in $\RS$ by induction on the structure of oracle functions 
(and relying on the encoding machinery presented
in~\cite{Buss86,Ferreira90}).

\item We show that all functions which are $\Sigma^b_1$-representable
in $\RS$ are in $\POR$ by realizability techniques
similar to Cook and Urquhart's one~\cite{CookUrquhart}.

\item We generalize Cobham's result to  probabilistic
models, 
showing that functions in $\POR$
are precisely those in $\RFP$.
\end{enumerate}

 \small
 
 \begin{figure}[h!]\label{fig:schema2}
\begin{center}
\framebox{
\parbox[t][2.8cm]{10cm}{
\footnotesize{

 \begin{center}
 \begin{tikzpicture}[node distance=2cm]
\node at (-4,0) (a) {$\RS$}; 
\node at (0,0) (b) {$\POR$}; 
\node at (4,0) (c) {$\RFP$};
\node at (-2,1) {\textcolor{gray}{realizability~\cite{CookUrquhart}}}; 
\node at (-2,-1) {\textcolor{gray}{induction~\cite{Ferreira88}}}; 
\node at (2,0.3) {\textcolor{gray}{series of simulations}}; 

    \draw[->,thick,dotted] (a) to [bend right=30] (b);
    \draw[->,thick,dotted] (b) to [bend right=30] (a);
    \draw[<->,thick,dotted] (c) to (b);

   %
\end{tikzpicture} 

 \end{center}
  }}}
\caption{Our Proof in a Nutshell}
\end{center}
\end{figure}
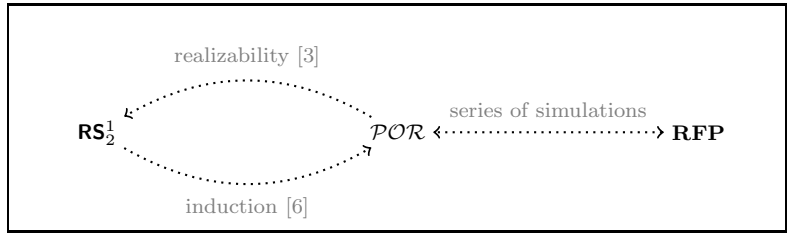  
\normalsize

\section{Introducing $\POR$ and $\RS$}\label{sec:introRS}

In this section, we introduce
a Cobham-style function algebra for poly-time oracle
recursive functions $\POR$,
and a randomized bounded arithmetic $\RS$.
As Ferreira's ones~\cite{Ferreira88,Ferreira90},
these classes are both associated with binary strings
rather than natural positive integer:
$\POR$ is a class of oracle functions over sequences of bits,
while $\RS$ is defined in a probabilistic word language $\rl$.
 Strings support a natural notion of term-size
 and make it easier to deal with bounds and time-complexity.
Observe that working with strings is not crucial
and all results below could be spelled out in terms of natural
 numbers.
Indeed, theories have been introduced
 in both formulations – Ferreira's $\Sigma^b_1$-NIA
 and Buss' $\BussS$ – and proved equivalent~\cite{FerreiraOitavem}.

 \subsection{The Function Algebra $\POR$} 
 We introduce a function algebra for
 poly-time \emph{oracle} recursive functions inspired by Ferreira's 
 $\PTCA$~\cite{Ferreira88,Ferreira90},
 and defined over strings.

 \begin{notation}
Let $\Bool=\{\zero,\one\}$,
 $\Ss=\Bool^*$ be the set of binary strings
 of finite length, and $\Os=\Bool^\Ss$ be the set of binary
 strings of infinite length.
 Metavariables $\eta',\eta'',\dots$ are used to
 denote the elements of $\Os$.
 \end{notation}
 \noindent
 Let $|\cdot |$ denote the length-map string, so that for any 
 string $x$, $|x|$ indicates the length of $x$.
 Given two binary strings $x,y$ we use $x\ssubseteq y$
 to express that $x$ is an \emph{initial}
 or \emph{prefix substring} of $y$,
 $x\cconc y$ (abbreviated as $xy$) for concatenation,
 and $x\ttimes y$ obtained by self-concatenating
 $x$ for $|y|$-times.
 Given an infinite string of bits $\eta$,
 and a finite string $x$, $\eta(x)$ denotes
 \emph{one} specific bit of $\eta$,
 the so-called $x$-th bit of $x$.

A fundamental difference between oracle functions
and those of $\PTCA$ is that the latter ones are 
of the form $f:\Ss^k \times \Os \to \Ss$,
carrying an additional argument 
to be interpreted as the underlying source of random bits.
Furthermore, $\PR$ includes the basic function \emph{query},
$\queryS(x,\eta)=\eta(x)$, which can be used to
observe any bit in $\eta$.
%
%

\begin{defn}[The Class $\POR$]
The \emph{class $\POR$} is the smallest class
of functions $f:\Ss^k \times \Os\to \Ss$,
containing:
\begin{itemize}
\itemsep0em

\item The \emph{empty (string) function} $\emptyS(x,\eta)=\eepsilon$

\item The \emph{projection (string) functions} $\projS^{n}_i(x_1,\dots, x_n,\eta) =x_i$,
for $n\in\Nat$ and $1\leq i\leq n$

\item The \emph{word-successor} $\successorS_{b}(x,\eta)=x\bool$,
where $b=0$ if $\bool=\one$ and $b=1$ if $\bool=\one$

\item The \emph{conditional (string) function}
\begin{align*}
\condS(\eepsilon, y, z_0, z_1, \eta) &= y \\
\condS(x\bool, y, z_0, z_1, \eta) &= z_{b}, 
\end{align*}
where $b=0$ if $\bool = \zero$ and $b=1$ if $\bool=\one$

\item The \emph{query (string) function} $\queryS(x,\eta) = \eta(x)$

\end{itemize}
and closed under the following schemas:
\begin{itemize}
\itemsep0em

\item \emph{Composition}, where $f$ is defined from
$g,h_1,\dots, h_k$ as
$$
f(\vec{x},\eta) = g(h_1(\vec{x},\eta), \dots, h_k(\vec{x},\eta),
\eta)
$$

\item \emph{Bounded recursion on notation}, where $f$
is defined from $g,h_0,$ and $h_1$ as
\begin{align*} 
f(\vec{x},\eepsilon,\eta) &= g(\vec{x},\eta) \\
f(\vec{x}, y\zero, \eta) &= h_0(\vec{x}, y, f(\vec{x}, y, \eta),
\eta) |_{t(\vec{x},y)} \\
f(\vec{x}, y\one, \eta) &= 
h_1(\vec{x}, y, f(\vec{x},y,\eta), \eta)|_{t(\vec{x},y)}
\end{align*}
and $t$ is obtained from $\eepsilon, \one, \zero, \cconc,$ 
and $\ttimes$ by explicit definition,
that is $t$ can be obtained applying $\cconc$ and $\ttimes$
on the constants $\eepsilon, \zero,\one$, and
the variables $\vec{x}$ and $y$.\footnote{Notice that there is a
clear correspondence with the grammar for terms in $\rl$,~Definition~\ref{def:RLterms}.}
\end{itemize}
\end{defn}
\noindent
Actually, the conditional function $\condS$
could be defined by bounded recursion.

\longv{
As follows
\begin{align*}
\condS(\eepsilon, y, z_0, z_1,\eta) = y \\
\condS(x\zero, y, z_0,z_1,\eta) = \projS_3(x\zero,y,z_0,z_1,
f(x\zero,y,z_0,z_1,\eta),\eta)|_{\zero} \\
\condS(x\one, y, z_0,z_1,\eta) = \projS_4(x\one,y,z_0,z_1,\eta)|_{\zero}.
\end{align*}
We take it as a primitive function of $\POR$
to help making the realizability of Section~\ref{sec:realizability}
straightforward.
}

\begin{remark}
Neither the query function or the conditional function
appear in Ferreira's characterization~\cite{Ferreira90},
which instead contains the ``substring-conditional''
function:
$$
\successorS(x,y,\eta) = \begin{cases}
\one \ \ &\text{if } x \ssubseteq y \\
\zero \ \ &\text{otherwise,}
\end{cases}
$$
which can be defined in $\POR$ by bounded recursion.
\longv{
First, let $\tailF (x,\eta)$ be defined as follows:
\begin{align*}
\tailF (\eepsilon,\eta) &= \eepsilon \\
\tailF (x\bool,\eta) &= x.
\end{align*}
\textcolor{red}{Then, let $\prF(x,\eta)$ be:}
$$
\textcolor{red}{\cdots}
$$
and $\eqF(x,y,\eta)$ be:
\begin{align*}
\eqF(x,\eepsilon,\eta) &= 
\condS(x,\one,\eepsilon,\eepsilon,\eta) \\
\eqF(x,y\zero,\eta) &=
\condS(x,\eepsilon,\eqF(\prF(x), y,\eta), \eepsilon, \eta) \\
\eqF(x,y\one,\eta) &=
\condS(x,\eepsilon,\eepsilon, \eqF(\prF(x),y,\eta),\eta).
\end{align*}
Finally, we can define $\successorS(x,y,\eta)$ as:
\begin{align*}
\successorS(x,\eepsilon,\eta) &= \condS(x,\one,\zero,\zero,\eta) \\
\successorS(x,y\zero,\eta) &= 
\condS(x,\one, \condS(\eqF(x,y\zero,\eta),
\successorS(x,y\zero,\eta), \one,\one,\eta),
\successorS(x,y,\eta),\eta) \\
\successorS(x,y\one,\eta) &=
\condS(x,\one,\successorS(x,y,\eta), \condS(\eqF(x,y\one,\eta),
\successorS(x,y,\eta), \one,\one,\eta)), \eta).
\end{align*}
}
\end{remark}

\subsection{Randomized Bounded Arithmetics}

First, we introduce a probabilistic word language for our bounded arithmetic, 
together with its quantitative interpretation.

\paragraph{The Language $\rl$.}
Following~\cite{FerreiraOitavem},
we consider a first-order signature for natural numbers
in binary notation endowed with a special predicate symbol
$\Flip(\cdot)$. 
Consequently, formulas are interpreted over $\Ss$ rather
than $\Nat$.

\begin{defn}[Terms and Formulas of $\rl$]\label{df:rl}
\emph{Terms} and \emph{formulas of $\rl$} are defined by the grammar below:
\begin{align*}
t &\df x \midd \epsilon \midd \zzero \midd \oone \midd t\conc t \midd 
t \times t \\
F &\df \Flip(t) \midd t=s \midd \neg F \midd
F \wedge F \midd F \vee F \midd (\exists x)F \midd
(\forall x)F.
\end{align*}
\end{defn}
\noindent

\begin{notation}[Truncation]
For readability, we adopt the following abbreviations:
$ts$ for $t\conc s$,
$\oone^t$ for $\oone \times t$,
and $t\preceq s$ for $\oone^t \subseteq \oone^s$,
expressing that the length of $t$ is smaller than that of $s$.
Given three terms $t,r,$ and $s$, the abbreviation
$t|_r = s$ denotes the following formula,
$$
(\oone^r \subseteq \oone^t \wedge s \subseteq t \wedge
\oone^r = \oone^s) \vee
(\oone^t \subseteq \oone^r \wedge s=t),
$$
saying that $s$ is the \emph{truncation} of $t$
at the length of $r$.
\end{notation}
\noindent
Every string $\str\in \Ss$ can be seen as a term of $\rl$
$\overline{\str}$, such that
$\overline{\eepsilon}=\epsilon$,
$\overline{\str\bool} = \overline{\str}\bbool$,
where $\bool\in\Bool$ and $\bbool\in\{\zzero,\oone\}$,
e.g.~$\overline{\zero\zero\one} = \zzero\zzero\oone$.

A central feature of bounded arithmetic is the presence of
bounded quantification.
%

\begin{notation}[Bounded Quantifiers]
In $\rl$, \emph{bounded quantified expressions}
are expressions of either the form 
$(\forall x)(\oone^x \subseteq \oone^t \rightarrow \fone)$ or
$(\exists x)(\oone^x \subseteq \oone^t \wedge \fone)$,
usually abbreviated as $(\forall x\preceq t)\fone$ and
$(\exists x \preceq t)\fone$ respectively.
\end{notation}
%
%
%

\begin{notation}
We call \emph{subword quantifications}, 
quantifications of the form 
$(\forall x\subseteq^* t)\fone$ and 
$(\exists x\subseteq^* t)\fone$,
abbreviating 
$(\forall x)\big((\exists w \subseteq t)(wx\subseteq t)
\rightarrow \fone\big)$
and 
$(\exists x)(\exists w\subseteq t)
(wx \subseteq t \wedge \fone)$.
Furthermore, 
we abbreviate so-called \emph{initial subword
quantification} 
$(\forall x)(x\subseteq t \rightarrow \fone)$
as $(\forall x\subseteq t)\fone$
and 
$(\exists x)(x\subseteq t \wedge \fone)$
as $(\exists x\subseteq t)\fone$.
\end{notation}
\noindent
The distinction between bounded and subword
quantification is important for complexity reasons.
If $\str \in \Ss$ is a string of length $k$,
the witness of a subword existentially quantified
formula
$(\exists x \subseteq^* \overline{\str})\fone$
is to be looked for among all possible sub-strings of $\str$, 
that is~within a space of size $\mathcal{O}(k)$.
On the contrary, the witness of a bounded formula 
$(\exists x\preceq \overline{\str})\fone$
is to be looked for among all possible strings of
length $k$, namely~within a space of size $\mathcal{O}(2^k)$.

\begin{remark}
In order to avoid misunderstanding let us briefly sum up 
the different notions and symbols used for subword relations.
We uses $\ssubseteq$ to express a relation between strings,
that is $x\ssubseteq y$ expresses that $x$ is an initial or prefix substring of $y$.
We use $\subseteq$ as a relation symbol in the language $\rl$.
We use $\preceq$ as an auxiliary symbol in the language $\rl$;
in particular, as seen, $t\preceq s$ is syntactic sugar for $\oone^t\subseteq \oone^s$.
We use $\subseteq^*$ as an auxiliary symbol in the language $\rl$
to denote subword quantification.
We also use $(\exists w\subseteq t)\fone$ as an abbreviation of 
$(\exists w)(w\subseteq t\wedge \fone)$
and similarly for $(\forall w\subseteq t)$.
\end{remark}

\begin{defn}[$\Sigma^b_1$-Formulas]\label{df:Sigmab1}
A $\Sigma^b_0$-formula is a subword quantified
formula,
i.e.~a formula belonging to the smallest class
of $\rl$ containing atomic
formulas and closed under Boolean operations 
and subword quantification.
A formula is said to be a 
\emph{$\Sigma^b_1$-formula}, if it is of the form
$(\exists x_1\preceq t)\dots(\exists x_n\preceq t_n)
\fone$,
where the only quantifications in $\fone$ are 
subword ones.
We call $\Sigma^b_1$ the class containing
all and only the $\Sigma^b_1$-formulas.
\end{defn}
\noindent
An \emph{extended $\Sigma^b_1$-formula}
is any formula of $\rl$ that can be constructed
in a finite number of steps, starting with
subword quantifications and bounded existential
quantifications and bounded existential quantifications.

\paragraph{The Bounded Theory $\RS$.}
We introduce the bounded theory $\RS$, 
which can be seen as a
probabilistic version to Ferreira's 
$\Sigma^b_1$-NIA~\cite{Ferreira88}.
It is expressed in the language $\rl$.

\begin{defn}[Theory $\RS$]
The theory $\RS$ is defined by axioms
belonging to two classes:
\begin{itemize}
\itemsep0em
\item \emph{Basic axioms}:

\begin{enumerate}
\itemsep0em

\item $x\epsilon=x$

\item $x(y\bbool) = (xy)\bbool$

\item $x\times \epsilon=\epsilon$

\item $x\times x\bbool = (x\times y)x$

\item $x\subseteq \epsilon \leftrightarrow x=\epsilon$

\item $x\subseteq y\bbool \leftrightarrow
x\subseteq y \vee x=y\bbool$

\item $x\bbool = y\bbool \rightarrow x=y$

\item $x\zzero \neq y\oone$,

\item $x\bbool\neq \epsilon$
\end{enumerate}
 with $\bbool\in \{\zzero,\oone\}$

\item \emph{Axiom schema for induction on notation},
$$
B(\epsilon) \wedge (\forall x)\big(B(x) \rightarrow
B(x\zzero) \wedge B(x\oone)\big)
\rightarrow (\forall x)B(x),
$$
where $B$ is a $\Sigma^b_1$-formula in $\rl$.
\end{itemize}
\end{defn}
\noindent
Induction on notation
adapts the usual induction schema of $\PA$
to the binary representation.
Of course, as in Buss' and Ferreira's approach, 
the restriction of this schema to $\Sigma^b_1$-formulas
is essential
to characterize algorithms computed \emph{with bounded resources}.
Indeed, more general instances of the schema
would extend representability to random functions which
are not poly-time (probabilistic) computable.

\begin{prop}[\cite{Ferreira88}]
In $\RS$ any extended $\Sigma^b_1$-formula
is logically equivalent to a 
$\Sigma^b_1$-formula.\footnote{Actually, Ferreira
proved this result for the theory $\Sigma^b_1$-NIA~\cite[pp. 148-149]{Ferreira88}, 
but it clearly holds for $\rl$ as well.}
\end{prop}

\paragraph{Semantics for Formulas in $\rl$.}
We introduce a \emph{quantitative} semantics for
formulas of $\rl$,
which is strongly inspired by that for $\MQPA$.
In particular, function symbols of $\rl$ as well as the predicate
symbols ``='' and ``$\subseteq$'' have a standard interpretation
as relations over $\Ss$
in the canonical model $\mathscr{W}=(\Ss, \cconc, \ttimes)$,
while, as we shall see,
$\Flip(t)$ can be interpreted either in a standard way 
or following Definition~\ref{df:MQPAsem}.

\begin{defn}[Semantics for Terms in $\rl$]\label{def:RLterms}
Given a set of term variables $\mathcal{G}$,
an environment $\xi: \mathcal{G} \to \Ss$
is a mapping that assigns to each variable a string.
Given a term $t$ in $\rl$ and an environment $\xi$,
the \emph{interpretation of $t$ in $\xi$} is the string
$\model{t}_\xi \in \Ss$ inductively defined as follows:

\begin{minipage}{\linewidth}
\begin{minipage}[t]{0.35\linewidth}
\begin{align*}
\model{\epsilon}_\xi &\df \eepsilon \\
\model{\zzero}_\xi &\df \zero \\
\model{\oone}_\xi &\df \one \\
\end{align*}
\end{minipage}
\hfill
\begin{minipage}[t]{0.55\linewidth}
\begin{align*}
\model{x}_\xi &\df \xi(x) \in \Ss \\
\model{t\conc s}_\xi &\df \model{t}_\xi \model{s}_\xi \\
\model{t \times s} _\xi &\df \model{t}_\xi \times \model{s}_\xi.
\end{align*}
\end{minipage}
\end{minipage}
\end{defn}

%
%
As in Chapter~\ref{ch:MQPA}, we extend the canonincal model $\mathscr{W}$ with the probability probability space
$\mathscr{P}_{\mathscr{W}} = (\Os, \sigma(\Cyls), \mu_{\Cyls})$,
where $\sigma(\Cyls)\subseteq \mathcal{P}(\Os)$
is the Borel $\sigma$-algebra generated by cylinders
$\cyl^\bool_\str = \{\eta \ | \ \eta(\str)=\bool\}$,
for $\bool \in \Bool$, and such that
$\mu_{\Cyls}(\cyl^{\bool}_{\str})=\frac{1}{2}$.
To formally define cylinders over $\Bool^\Ss$,
we slightly modify Billingsley's notion of \emph{cylinder
of rank n}.

\begin{defn}[Cylinder over $S$]
For any countable set $S$,
finite $K\subset S$ and $H\subseteq \Bool^K$,
$$
\cyl(H) = \{\eta \in \Bool^S \ | \ \eta |_K \in H\},
$$
is a \emph{cylinder} over $S$.
\end{defn}
\noindent
Then, $\Cyls$ and $\sigma(\Cyls)$
are defined in the standard way and
the probability measure over it is defined as follows:\footnote{For further
details, see~\cite{Davoli}.}

\begin{defn}[Cylinder Measure]
For any countable set $S$, 
$K\subseteq S$ and $H\subseteq \Bool^K$ such
that $\cyl(H)=\{\eta \in \Bool^S \ | \ \eta|_K \in H\}$,
$$
\mu_{\Cyls}(\cyl(H))= \frac{|H|}{2^{|K|}}.
$$
\end{defn}
\noindent
This is a measure over $\sigma(\mathscr{C})$.
%

%
Then, formulas of $\rl$ are interpreted as sets \emph{measurable} sets.

\begin{defn}[Semantics for Formulas in $\rl$]\label{df:RLformulas}
Given a term $t$, a formula $F$,
and an environment $\xi: \mathcal{G} \to \Ss$,
where $\mathcal{G}$ is the set of term variables,
the \emph{interpretation of $F$ under $\xi$}
is the measurable set of sequences $\model{F}_\xi
\in \sigma(\Cyls)$ inductively defined as follows:

\begin{minipage}{\linewidth}
\begin{minipage}[t]{0.35\linewidth}
\begin{align*}
\\
\model{\Flip(t)}_\xi &\df
\{\eta \ | \ \eta(\model{t}_\xi) = \one\} \\
\model{t=s}_\xi &\df
\begin{cases}
\Os \ &\text{if } \model{t}_\xi = \model{s}_\xi   \\
\emptyset \ &\text{otherwise}
\end{cases} \\
\model{t\subseteq s}_\xi &\df
\begin{cases}
\Os \ &\text{if } \model{t}_\xi \ssubseteq \model{s}_\xi \\
\emptyset \ &\text{otherwise}
\end{cases} 
\end{align*}
\end{minipage}
\hfill
\begin{minipage}[t]{0.55\linewidth}
\begin{align*}
\model{\neg G}_\xi &\df \Os - \model{G}_\xi  \\
\model{G \vee H}_\xi &\df  \model{G}_\xi \cup \model{H}_\xi \\
\model{G \wedge H}_\xi &\df \model{G}_\xi \cap \model{H}_\xi \\
\model{G\rightarrow H}_\xi &\df (\Os - \model{G}_\xi) \cup
\model{H}_\xi \\
\model{(\exists x)G}_\xi &\df 
\bigcup_{i\in\Ss} \model{G}_{\xi\{x\leftarrow i\}} \\
\model{(\forall x)G}_\xi &\df 
\bigcap_{i\in \Ss} \model{G}_{\xi \{x\leftarrow i\}}. 
\end{align*}
\end{minipage}
\end{minipage}
\end{defn}
\noindent
As anticipated, this semantics is well-defined.
Indeed, the sets $\model{\Flip(t)}_\xi$,
$\model{t=s}_\xi$ and $\model{t\subseteq s}_\xi$
are measurable and measurability is preserved
by all the logical operators.

A interpretation of the language
$\rl$ in the usual sense is given due to an environment $\xi$
plus the choice of an interpretation $\eta$ for $\Flip(x)$.

\begin{defn}[Standard Semantics for Formulas in $\rl$]
Given a $\rl$-formula $\fone$, and an interpretation
$\rho=(\xi,\eta^{\mathtt{FLIP}})$,
where $\xi:\mathcal{G}\to \Ss$ and 
$\eta^{\mathtt{FLIP}} \subseteq \Os$,
the \emph{interpretation of $\fone$ in $\rho$} $\model{\fone}_\rho$,
is inductively defined as follows:

\begin{minipage}{\linewidth}
\begin{minipage}[t]{0.35\linewidth}
\begin{align*}
\model{\Flip(t)}_\rho &\df 
\begin{cases}
1 \ &\text{if } \eta^{\mathtt{FLIP}}(\model{t}_\rho) = 
\one \\
0 \ &\text{otherwise}
\end{cases} \\
\model{t=s}_\rho &\df
\begin{cases}
1 \ &\text{if } \model{t}_\rho = \model{s}_\rho \\
0 \ &\text{otherwise.}
\end{cases} \\
\model{t \subseteq s}_\rho &\df
\begin{cases}
1 \ &\text{if } \model{t}_\rho \ssubseteq \model{s}_\rho \\
0 \ &\text{otherwise}
\end{cases}
\end{align*}
\end{minipage}
\hfill
\begin{minipage}[t]{0.55\linewidth}
\begin{align*}
\model{\neg G}_\rho &\df 1 - \model{G}_\rho \\
\model{G \wedge H}_\rho &\df min\{\model{G}_\rho,
\model{H}_\rho\} \\
\model{G \vee H}_\rho &\df max\{\model{G}_\rho, 
\model{H}_\rho\} \\
\model{G \rightarrow H}_\rho &\df
max\{(1-\model{G}_\rho), \model{H}_\rho\} \\
\model{(\forall x)G}_\rho &\df
min\{\model{G}_{\rho\{x \leftarrow \str\}} \ | \ \str \in \Ss\} \\
\model{(\exists x)G}_\rho &\df max\{\model{G}_{\rho\{x\leftarrow \str\}}
\ | \ \str \in \Ss\}.
\end{align*}
\end{minipage}
\end{minipage}
\end{defn}

\begin{notation}
For readability's sake,
we abbreviate $\model{\cdot}_{\rho}$ simply as
 $\model{\cdot}_\eta$,
and $\model{\cdot}_\xi$ as  $\model{\cdot}$.
\end{notation}
\noindent
Observe that \emph{quantitative}
and \emph{qualitative}
semantics for $\rl$ are mutually related,
as can be proved by induction on the structure of formulas~\cite{Davoli}.

\begin{prop}
For any formula $F$ in $\rl$, environment $\xi$,
function $\eta \in \Os$ and $\rho=(\eta,\xi)$,
$$
\model{F}_{\xi,\eta} = 1 \ \ \text{iff} \ \ 
\eta \in \model{F}_\rho.
$$
\end{prop}

\section{$\RS$ characterizes $\POR$}\label{sec:PORandRS}

As said, our proof follows a so-to-say standard path~\cite{Buss86,Ferreira88}.
The first step consists in showing that functions in $\POR$
are precisely those which are $\Sigma^b_1$-representable in
$\RS$.
To do so, we extend Buss' representability conditions by
adding a constraint to link the quantitative semantics
of formulas in $\RS$ with the additional functional
parameter $\eta$ of oracle recursive functions.

\begin{defn}[$\Sigma^b_1$-Representability]\label{df:representability}
A function $f:\Ss^k \times \Os \to \Ss$ is 
\emph{$\Sigma^b_1$-representable
in $\RS$} if there is a 
$\Sigma^b_1$-formula $\fone(\vec{x},y)$ of $\rl$ such that:
\begin{enumerate}
\itemsep0em
\item $\RS\vdash (\forall \vec{x})(\exists y)\fone(\vec{x},y)$
\item $\RS \vdash (\forall \vec{x})(\forall y)(\forall z)
\big(\fone(\vec{x},y) \wedge \fone(\vec{x},z) \rightarrow y=z\big)$
\item for all $\str_1,\dots, \str_j,\strT \in \Ss$ and $\eta\in \Os$,
$$
f(\str_1,\dots, \str_j,\eta) = \strT \ \ \ \text{iff} \ \ \ 
\eta \in \model{\fone
(\overline{\str_1},\dots, \overline{\str_j},\overline{\strT})}.
$$
\end{enumerate}
\end{defn}
\noindent
We recall that the language $\rl$ allows us to associate
the formula $\fone$ with both a \emph{qualitative} 
– namely, when dealing with 1. and 2. –
and a \emph{quantitative} interpretation – namely, in 3.
Then, in  Section~\ref{sec:PORtoRS},
we prove the following theorem.

\begin{theorem}[$\POR$ and $\RS$]\label{thm:PORandRS}
For any function $f:\Ss^k\times \Os \to \Ss$,
$f$ is $\Sigma^b_1$-representable in $\RS$
when $f\in\POR$.
\end{theorem}
\noindent
In particular, that any function in $\POR$
is $\Sigma^b_1$-representable in
$\RS$ is proved in Section~\ref{sec:PORtoRS}
by a straightforward induction on the structure
of probabilistic oracle functions.
The other direction is established in 
Section~\ref{sec:RStoPOR} by a realizability argument
very close to the one offered in~\cite{CookUrquhart}.

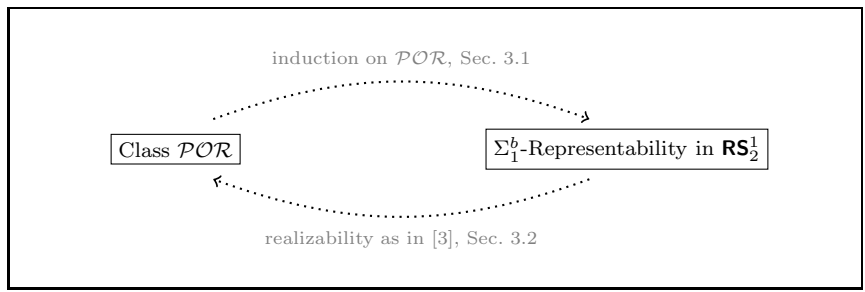
\begin{figure}[h!]
\begin{center}
\framebox{
\parbox[t][3.5cm]{11cm}{
\footnotesize{

 \begin{center}
 \begin{tikzpicture}[node distance=2cm]
\node[draw] at (-3,0) (a) {Class $\POR$};
\node[draw] at (3,0) (b) {$\Sigma^b_1$-Representability in $\RS$};
\node at (0,1.2) {\textcolor{gray}{\scriptsize{induction on $\POR$, Sec.~\ref{sec:PORtoRS}}}};
  \node at (0,-1.2) {\textcolor{gray}{\scriptsize{realizability as in~\cite{CookUrquhart}, Sec.~\ref{sec:RStoPOR}}}};

    \draw[->,dotted,thick] (-2.5,0.4) to [bend left=20] (2.5,0.4);
    \draw[->,dotted,thick] (2.5,-0.4) to [bend left=20] (-2.5,-0.4);    
   %
\end{tikzpicture} 
 \end{center}
 }}}
\caption{Relating $\POR$ and $\RS$}
\end{center}
\end{figure}

 \normalsize

\subsection{Functions in $\POR$ are $\Sigma^b_1$-Representable in $\RS$}\label{sec:PORtoRS}

We prove that any function in $\POR$ is 
$\Sigma^b_1$-representable in $\RS$ by
constructing the desired formula by induction on the structure
of oracle functions.
Preliminarily notice that, for example,  
the formula $(\forall \vec{x})(\exists y)G(\vec{x},y)$
occurring in condition 1. is \emph{not} in
$\Sigma^b_1$, since its existential
quantifier is not bounded.
Hence, in order to prove the inductive steps of Theorem~\ref{thm:PORtoRS}
– namely, composition
and bounded recursion on notation –
we need to adapt Parikh's theorem~\cite{Parikh}
to $\RS$.\footnote{The 
theorem is usually presented in the 
context of Buss' bounded theories,
as stating that given a bounded formula
$\fone$ in $\mathcal{L}_\Nat$ such that $\BussS \vdash
(\forall \vec{x})(\exists y)\fone$,
then there is a term $t(\vec{x})$
such that also $\BussS 
\vdash (\forall \vec{x})(\exists y\leq t(\vec{x}))
\fone(\vec{x},y)$~\cite{Buss86,Buss98}.
Furthermore, due to~\cite{FerreiraOitavem},
Buss' syntactic proof can be adapted
to $\Sigma^b_1$-NIA in a natural way.
The same result holds for $\RS$,
which does not contain any specific rule
concerning $\Flip(\cdot)$.}

\begin{prop}[``Parikh''~\cite{Parikh}]\label{prop:Parikh}
Let $\fone(\vec{x},y)$ be a bounded formula in $\rl$
such that
$\RS \vdash (\forall \vec{x})(\exists y) \fone(\vec{x},y)$.
Then, there is a term $t$ such that,
$$
\RS \vdash (\forall \vec{x})(\exists y \preceq t(\vec{x}))\fone(\vec{x},y).
$$
\end{prop}
%
%

\begin{theorem}\label{thm:PORtoRS}
Every $f\in \POR$ is $\Sigma^b_1$-representable
in $\RS$.
\end{theorem}

\begin{proof}[Proof Sketch]
The proof is by induction on the structure
of functions in $\POR$.\footnote{For further details, see Appendix~\cite{appendix:TaskA}.}

\emph{Base Case.} Each basic function is
$\Sigma^b_1$-representable in $\RS$.
There are five possible sub-cases:

\begin{itemize}
\itemsep0em

\item \emph{Empty (String) Function.} 
$f=\emptyS$ is $\Sigma^b_1$-represented in $\RS$ by
the formula:
$$
F_{\emptyS}(x,y) : x=x \wedge y=\epsilon.
$$
\begin{enumerate}
\itemsep0em

\item 
Existence is proved considering
$y=\epsilon$. 
For the reflexivity of identity both
$\RS\vdash x=x$ and $\RS \vdash \epsilon=\epsilon$
hold.
So, by rules for conjunction,
we obtain $\RS\vdash x=x \wedge \epsilon=\epsilon$,
and conclude:
$$
\RS\vdash (\forall x)(\exists y)(x=x \wedge y=\epsilon).
$$

\item Uniqueness is proved assuming $\RS \vdash
x= x \wedge z=\epsilon$.
By rules for conjunction,
in particular $\RS \vdash z=\epsilon$, and since
$\RS \vdash y=\epsilon$,
by the transitivity of identity, we conclude
$$
\RS\vdash y=z.
$$

\item Assume $\emptyS(\str, \eta^*) = \strT$.
If $\strT=\eepsilon$, then:
\begin{align*}
\model{\overline{\str} = \overline{\str} 
\wedge \overline{\strT} = \epsilon}
&=
\model{\overline{\str} = \overline{\str}} 
\cap
\model{\overline{\strT} = \epsilon} \\
&= \Os \cap \Os \\
&= \Os.
\end{align*}
So, in this case, for any $\eta^*$,
$\eta^* \in \model{\overline{\str}=\overline{\str}
\wedge \overline{\strT} = \epsilon}$,
as clearly $\eta^* \in\Os$.

If $\strT\neq \eepsilon$, then
\begin{align*}
\model{\overline{\str} = \overline{\str} \wedge
\overline{\strT} = \epsilon}
&=
\model{\overline{\str} = \overline{\str}}
\cap \model{\overline{\strT} =\epsilon} \\
&= \Os \cap \emptyset \\
&= \emptyset. 
\end{align*}
So, for any $\eta^*$,
$\eta^* \not\in \model{\overline{\str}= \overline{\str}
\vee \overline{\strT} = \epsilon}$,
as clearly $\eta^* \not\in \emptyset$.
\end{enumerate}

\item \emph{Projection (String) Function.}
$f=\projS^n_i$, for $1\leq i\leq n$, 
is $\Sigma^b_1$-represented in $\RS$ by
the formula:
$$
F_{\projS^{n}_{i}}(x,y) : \bigwedge_{j\in J}(x_j=x_j)
\wedge y= x_i,
$$
where $J=\{1,\dots, n\}\setminus \{i\}$.

\item \emph{Word-Successor Function.} 
$f=\successorS_b$ is $\Sigma^b_1$-represented
in $\RS$ by the formula:
$$
F_{S_b}(x,y) : y= x\bbool
$$
where $\bbool=\zzero$ if $b=0$
and $\bbool=\oone$ if $b=1$.

\item \emph{Conditional (String) Function.}
$f=\condS$ is $\Sigma^b_1$-represented
in $\RS$ by the formula:
\begin{align*}
\fone_{\condS} (x,v,z_0,z_1,y) :
(x=\epsilon \wedge y=v) 
&\vee (\exists x' \preceq x)
(x=x'\zzero \wedge y=z_0) \\
&\vee 
(\exists x'\preceq x)
(x=x'\oone \wedge y=z_1).
\end{align*}

\item \emph{Query (String) Function.}
$f=\queryS$ is $\Sigma^b_1$-represented
in $\RS$ by the formula:
$$
F_{\queryS}(x,y) :
(\Flip(x) \wedge y=\oone) \vee
(\neg \Flip(x) \wedge y=\zzero).
$$
Notice that, in this case, the proof crucially
relies on the fact that oracle functions 
invoke \emph{exactly one} oracle.

\begin{enumerate}
\itemsep0em

\item Existence is proved by cases.
If $\RS \vdash \Flip(x)$,
let $y=\oone$.
By the reflexivity of identity,
$\RS \vdash \oone = \oone$ holds,
so also $\RS \vdash \Flip(x) \wedge \oone=\oone$.
By rules for disjunction, we conclude
$\RS\vdash (\Flip(x) \wedge \oone = \oone)
\vee (\neg \Flip(x) \wedge \oone= \zzero)$
and so,
$$
\RS \vdash (\exists y)\big((\Flip(x) \wedge y=\oone)
\vee (\neg \Flip(x) \wedge y=\zzero)\big).
$$
If $\RS \vdash \neg \Flip(x)$,
let $y=\zzero$.
By the reflexivity of identity 
$\RS \vdash \zzero=\zzero$ holds.
Thus, by the rules for conjunction,
$\RS \vdash \neg \Flip(x) \wedge \zzero=\zzero$
and for disjunction,
we conclude $\RS \vdash(\Flip(x) \wedge \zzero=\oone)
\vee (\neg \Flip(x) \wedge \zzero=\zzero)$ and so,
$$
\RS \vdash (\exists y) \big((\Flip(x) \wedge
y= \oone) \vee
(\neg \Flip(x) \wedge y=\zzero)\big).
$$

\item Uniqueness is established 
relying on the transitivity of identity.

\item Finally, it is shown that for every
$\str,\strT\in \Ss$ and 
$\eta^* \in \Os$,
$\queryS(\str,\eta^*)=\strT$ when
$\eta^*\in \model{F_{\queryS} 
(\overline{\str},\overline{\strT})}$.
Assume $\queryS(\str,\eta^*) = \one$,
which is $\eta^*(\str)=\one$,
\begin{align*}
\model{\fone_Q(\overline{\sigma},\overline{\tau})}
&=
\model{\Flip(\overline{\str}) \wedge 
\overline{\strT}=\oone}
\cup
\model{\neg \Flip(\overline{\str})
\wedge \overline{\strT}=\zzero} \\
&= (\model{\Flip(\overline{\str})} \cap
\model{\oone=\oone}) \cup
(\model{\neg \Flip(\overline{\str})}
\cap \model{\oone=\zzero}) \\
&= (\model{\Flip(\overline{\str})} \cap \Os)
\cup 
(\model{\neg \Flip(\overline{\str})} \cap
\emptyset) \\
&= \model{\Flip(\overline{\str})} \\
&= \{\eta \ | \ \eta(\str)=\one\}.
\end{align*}

\normalsize
Clearly, $\eta^* \in  \model{(\Flip(\overline{\str}) \wedge \overline{\strT} = \oone) \vee (\neg \Flip(\overline{\str}) \wedge \overline{\strT}=\zzero)}$.

The case $\queryS(\sigma,\eta^*)=\zero$ 
and the opposite direction are proved in a similar way. 
\end{enumerate} 
\end{itemize}

 \emph{Inductive Case.}
If $f$ is defined by composition or bounded recursion
from $\Sigma^b_1$-representable functions,
then $f$ is $\Sigma^b_1$-representable in $\RS$:

\begin{itemize}
\itemsep0em

\item \emph{Composition.}
Assume that $f$ is defined by composition from functions
$g,h_1,\dots, h_k$ so that
$$
f(\vec{x},\eta) = g(h_1(\vec{x},\eta), \dots, h_k(\vec{x},\eta),
\eta)
$$
and that $g,h_1,\dots, h_k$ are represented in $\RS$
by  the $\Sigma^b_1$-formulas
$F_g,F_{h_1},$ $\dots, F_{h_k}$, respectively.
By Proposition~\ref{prop:Parikh},
there exist suitable terms $t_{g}, t_{h_1},\dots, t_{h_k}$
such that condition 1. of Definition~\ref{df:representability}
can be strengthened to $\RS\vdash(\forall \vec{x})
(\exists y \preceq t_i)F_i(\vec{x},y)$ for each $i\in\{g,h_1,\dots, h_k\}$.
We conclude that $f(\vec{x},\eta)$ is $\Sigma^b_1$-represented
in $\RS$ by the following formula:
\begin{align*}
F_f(x,y) : (\exists z_1\preceq t_{h_1}(\vec{x}))
\dots 
(\exists z_k \preceq t_{h_k}(\vec{x}))
\big(F_{h_1}(\vec{x},z_1) &\wedge \dots
F_{h_k}(\vec{x},z_k) \\ 
&\wedge
F_g(z_1,..., z_k,y)\big).
\end{align*}
Indeed, by IH,
$F_g,F_{h_1},\dots, F_{h_k}$ are $\Sigma^b_1$-formulas.
Then, also $F_f$ is in $\Sigma^b_1$.
Conditions 1.-3. are proved to hold by slightly modifying standard
proofs.

\item \emph{Bounded Recursion.} Assume that $f$
is defined by bounded recursion from
$g,h_0,$ and $h_1$ so that:
\begin{align*}
f(\vec{x},\eepsilon,\eta) &= g(\vec{x},\eta) \\
f(\vec{x}, y\bool, \eta) &= h_i(\vec{x},y,f(\vec{x},y,\eta),\eta)|_{t(\vec{x},y)},
\end{align*}
where $i \in\{0,1\}$
and $\bool=\zero$ when $i=0$ while
and $\bool = \one$ when $i=1$.
Let $g,h_0, h_1$ be represented in $\RS$
by, respectively, the $\Sigma^b_1$-formulas
$F_g,F_{h_0},$ and $F_{h_1}$.
Moreover, by Proposition~\ref{prop:Parikh},
there exist suitable terms $t_g,t_{h_0},$ and $t_{h_1}$
such that condition 1. of Definition~\ref{df:representability}
can be strengthened to its ``bounded version''.
Then, it can be proved that $f(\vec{x},y)$
is $\Sigma^b_1$-represented in $\RS$ by the formula below:
\small
\begin{align*}
F_f(x,y) : \ &(\exists v \preceq t_g(\vec{x}) 
t_f(\vec{x}) (y \times t(\vec{x},y) t(\vec{x},y) \oone\oone))
\big(F_{lh}(v, \oone \times y \oone)  \\
& \wedge (\exists z\preceq t_g(\vec{x}))
(F_{eval}(v,\epsilon,z) \wedge
F_g(\vec{x},z)) \\
& \wedge (\forall u \subset y)(\exists z)
\big(\tilde{z} \preceq t(\vec{x},y))
\big(F_{eval}(v,\oone \times u, z) \wedge
F_{eval}(v,\oone \times u\oone,\tilde{z}) \\
& \wedge (u\zzero \subseteq y \rightarrow
(\exists z_0 \preceq t_{h_0}(\vec{x},u,z))
(F_{h_0}(\vec{x},u,z,z_0)
\wedge z_0|_{t(\vec{x},u)} = \tilde{z})) \\
& \wedge (u\oone \subseteq y \rightarrow
(\exists z_1 \preceq t_{h_1}(\vec{x},u,z))
(F_{h_1}(\vec{x},u,z,z_1) \wedge
z_1|_{t(\vec{x},u)} = \tilde{z}))\big)\big),
\end{align*}
\normalsize
where $\fone_{lh}$ and $F_{eval}$ are $\Sigma^b_1$-formulas
defined as in~\cite{Ferreira90}.
Intuitively, $F_{lh}(x,y)$ states that the number of $\oone$s
in the encoding of $x$ is $yy$,
while $F_{eval}(x,y,z)$ is a ``decoding'' formula
(strongly resembling G\"odel's $\beta$-formula),
expressing that the ``bit'' encoded in $x$ as its $y$-th bit
is $z$.
Moreover $x\subset y$ is an abbreviation for $x\subseteq y
\wedge x \neq y$.
Then, this formula $F_f$ satisfies all the requirements to 
$\Sigma^b_1$-represent in $\RS$ the function $f$,
obtained by bounded recursion from $g,h_0,$ and $h_1$.
In particular, conditions 1. and 2. concerning existence and 
uniqueness,
have already been proved to hold by Ferreira~\cite{Ferreira90}.
Furthermore, $F_f$ expresses that,
given the desired encoding sequence $v$:
(i.) the $\eepsilon$-th bit of $v$ is (the encoding
of) $z'$ such that $F_g(\vec{x},z')$ holds,
where (for IH) $F_g$ is the $\Sigma^b_1$-formula
representing the function $g$,
and (ii.) given that for each $u\subset y$,
$z$ denotes the `` bit'' encoded in $v$ at position 
$1\times u$ and, similarly, $\tilde{z}$ is the next ``bit'',
encoded in $v$ at position $\oone \times u\oone$,
then if $u\bbool \subseteq y$
(that is, if we are considering the initial substring of $y$
the last bit of which corresponds to $\bbool$),
then there is a $z_b$ such that
$F_{h_b}(\vec{x},y,z,z_b)$, where $F_{h_b}$ $\Sigma^b_1$-represents
the function $f_{h_b}$ and the truncation of $z_b$
at $t(\vec{x},u)$ is precisely
$\tilde{z}$, with $b=0$ when $\bbool=\zzero$ and
$b=1$ when $\bbool=\oone$.\footnote{Otherwise said, 
if $u\zzero\subseteq y$, there is a $z_0$ such that
the $\Sigma^b_1$-formula $F_{h_0}(\vec{x},
u,z,z_0)$ represents  the function $h_0$
and, in this case, $\tilde{z}$ corresponds
to the truncation of $z_0$ at $t(\vec{x},u)$,
that is the ``bit'' encoded by $v$ at the position
$\oone \times u\oone$ (i.e.~corresponding to
$u\zzero \subseteq y$) is precisely such $\tilde{z}$.
Equally, if $u\oone\subseteq y$, there is a 
$z_0$ such that the $\Sigma^b_1$-formula
$F_{h_1}(\vec{x},u,z,z_1)$ represents now the function
$h_1$ and $\tilde{z}$ corresponds to the truncation
of $z_1$ at $t(\vec{x},u)$,
that is the ``bit'' encoded by $v$ at position
$\oone\times u\oone$ (i.e.~corresponding to $u\oone
\subseteq y)$ is precisely such $\vec{z}$. For further details, see Appendix~\cite{appendix:TaskA}.}
\end{itemize}

\end{proof}

\subsection{The functions which are $\Sigma^b_1$-Representable
in $\RS$ are in $\POR$}\label{sec:RStoPOR}

Here, we consider the opposite direction.
Our proof is obtained by adapting that by Cook and Urquhart 
for $\IPV$~\cite{CookUrquhart}.
It passes through a \emph{realizability interpretation}
of the intuitionistic version of $\RS$,
called $\IRS$ and is structured as follows:
\begin{enumerate}
\itemsep0em
\item First, 
we define $\PORl$ a basic equational theory 
for a simply typed $\lambda$-calculus 
endowed with primitives corresponding to functions
of $\POR$.

\item Second, 
we introduce a first-order \emph{intuitionistic}
theory $\IPOR$, which extends $\PORl$
with the usual predicate calculus as well as an
$\NP$-induction schema.
It is shown that $\IPOR$ is strong enough to prove
all theorems of $\IRS$, the intuitionistic
version of $\RS$.

\item Then,
we develop a realizability interpretation of
$\IPOR$ (inside itself), showing that from
any derivation of $(\forall x)(\exists y)\fone(x,y)$
(where $\fone$ is a $\Sigma^b_0$-formula)
one can extract a $\lambda$-term $\term{t}$
of $\POR^\lambda$, such that $(\forall x)\fone(x,\term{t} x)$ is provable in $\IPOR$.
From this we deduce that every function which
is $\Sigma^b_1$-representable in $\IRS$ is in $\POR$.

\item Finally, 
we extend this result to classical $\RS$ 
showing that any $\Sigma^b_1$-formula
provable in $\IPOR$ + Excluded Middle ($\EM$, for short)
is already provable in $\IPOR$.
\end{enumerate}

\subsubsection{The System $\PORl$}
We define an equational theory for a simply
typed $\lambda$-calculus augmented with
primitives for functions of $\POR$.
Actually, these do not
exactly correspond to the ones of $\POR$, 
although the resulting function algebra is proved 
equivalent.\footnote{Our choice follows the principle
that the defining equations for the functions
different from the recursion operator should not 
depend on it.
}

\paragraph{The Syntax of $\PORl$.}
We start by considering the syntax of $\PORl$.
\begin{defn}[Types of $\PORl$]
\emph{Types of $\PORl$} are defined by the
grammar below:
$$
\typeOne := s \midd \typeOne \arrowT
\typeOne.
$$
\end{defn}

\begin{defn}[Terms of $\PORl$]\label{df:termsPORl}
\emph{Terms of $\PORl$} are standard,
simply typed $\lambda$-terms plus the
constants below:
\begin{align*}
\term{0}, \term{1}, \term{\epsilon} &:  s \\
\circ &: s \arrowT s \arrowT s \\
\term{Tail} &: s \arrowT s \\
\term{Trunc} &: s \arrowT s \arrowT s \\
\term{Cond} &: s \arrowT s \arrowT s 
\arrowT s \arrowT s \\
\term{Flipcoin} &: 
s \arrowT s \\
\term{Red} &: 
s \arrowT (s \arrowT s \arrowT s)
\arrowT
(s \arrowT s \arrowT s)
\arrowT 
(s \arrowT s) \arrowT s 
\arrowT s.
\end{align*}
\end{defn}
\noindent
Intuitively, $\term{Tail}(x)$ computes the
string obtained by deleting the first digit
of $x$;
$\term{Trunc}(x,y)$ computes the string 
obtained by truncating $x$ at the length
of $y$;
$\term{Cond}(x,y,z,w)$ computes
the function that yields $y$ when $x=\eepsilon$,
$z$ when $x=x'\zero$, 
and $w$ when $x=x'\one$;
$\term{Flipcoin}(x)$ indicates a random 
$\zero/\one$ generator;
$\term{Rec}$ is the operator for bounded recursion
on notation.

\begin{notation}
We usually abbreviate $x\circ y$ as $xy$.
Moreover, for readability's sake,
being
$\term{T}$ any constant $\term{Tail},
\term{Trunc}, \term{Cond}, \term{Flipcoin},
\term{Rec}$ of arity $n$,
we indicate $\term{T}\term{u}_1,\dots, \term{u}_n$
as $\term{T}(\term{u}_1,\dots, \term{u}_n)$.
\end{notation}

$\PORl$ is reminiscent of $\PV$ by Cook and Urquhart~\cite{CookUrquhart}
without the induction rule (R5) that we do not need.
The main difference being the constant $\term{Flipcoin}$,
which, as said, intuitively denotes a function
which randomly generates either $\zero$ or
$\one$ when reads a string.\footnote{These interpretations will be made clear by
Definition~\ref{df:provRepr} below.}

\begin{remark}
In the following, we often define terms
implicitly using bounded recursion on notation.
Otherwise said, we define new terms, say
$\term{F} : s \arrowT \dots \arrowT s$, by equations
of the form:
\begin{align*}
\term{F} \vec{x} \term{\epsilon} &\df
\term{G} \vec{x} \\
\term{F} \vec{x}(y\term{0}) &\df
\term{H}_0 \vec{x} y(\term{F}\vec{x}y) \\
\term{F}\vec{x}(y\term{1}) &\df
\term{H}_1 \vec{x} y(\term{F}\vec{x}y),
\end{align*}
where $\term{G}, \term{H}_0, \term{H}_1$
are already-defined terms,
and the second and third equations
satisfy a length bound given by some term $\term{K}$
(which is usually $\lambda \vec{x} \lambda y.\term{0}$).
The term $\term{F}$ can be explicitly defined as
follows:
$$
\term{F} \df \lambda \vec{x}. \lambda y.
\term{Rec}(\term{G}\vec{x}, \lambda yy'.
\term{H}_0 \vec{x}yy', \lambda y y'. \term{H}_1
\vec{x}yy', \term{K}\vec{x},y).
$$
\end{remark}
\noindent
We also introduce the following abbreviations for composed
functions:
\begin{itemize}
\itemsep0em

\item $\term{B}(x) \df \term{Cond}(x,\epsilon,\term{0},
\term{1})$ denotes the function that computes the last
digit of $x$,
i.e.~coerces $x$ to a Boolean value
\item $\term{BNeg}(x) \df \term{Cond}(x,\epsilon,\term{1},\term{0})$ denotes the function that
computes the Boolean negation of $\term{B}(x)$
\item $\term{BOr}(x,y) \df \term{Cond}(\term{B}(x),
\term{B}(y), \term{B}(y), \term{1})$ denotes
the function that coerces $x$ and $y$ to Booleans
and then performs the OR operation
\item $\term{BAnd}(x,y) \df \term{Cond}(\term{B}(x),
\epsilon, \term{0}, \term{B}(y))$ denotes
the function that coerces $x$ and $y$ to Booleans
and then performs the AND operation
\item $\term{Eps}(x) \df \term{Cond}(x,\term{1},\term{0}, \term{0})$ denotes the characteristic
function of the predicate ``$x=\epsilon$''
\item $\term{Bool}(x) \df \term{BAnd}(\term{\Eps}(\term{Tail}(x)), \term{BNeg}(\term{Eps}(x)))$
denotes the characteristic function of the predicate
``$x=\zzero \vee x=\oone$''

\item $\term{Zero}(x) \df \term{Cond}(\term{Bool}(x),
\term{0}, \term{Cond}(x,\term{0},\term{0},\term{1}),
\term{0})$ denotes the characteristic function
of the predicate ``$x=\zzero$''

\item $\term{Conc}(x,y)$ denotes the concatenation
function defined by the equations below:
\begin{align*}
\term{Conc}(x,\epsilon) \df x \ \ \ \ \ \ \ \
\term{Conc}(x,y\term{b}) \df \term{Conc}(x,y)\term{b}, \\
\end{align*}
with $\term{b} \in\term{\{0,1\}}$.
\item $\term{Eq}(x,y)$ denotes the characteristic
function of the predicate ``$x=y$'' and is defined
by double recursion by the equation below:
\begin{align*}
\term{Eq}(\epsilon,\epsilon) &\df \term{1} \ \ \ \ \ \ \ \ 
\term{Eq}(\epsilon, y\term{b}) \df \term{0} \\
\term{Eq}(x\term{b}, \epsilon) 
= \term{Eq}(x\term{0},y\term{1})
= \term{Eq}(x\term{1},y\term{0}) &\df
\term{0} 
\ \ \ \ \ \ 
\term{Eq}(x\term{b}, y\term{b}) 
\df \term{Eq}(x,y),
%
\end{align*}
with $\term{b} \in \{\term{0}, \term{1}\}$
\item $\term{Times}(x,y)$ denotes the function
for self-concatenation, $x,y\mapsto x\ttimes y$,
and is defined by the equations below:
$$
\term{Times}(x,\epsilon) \df \epsilon  \ \ \ \ \ \ \ \
\term{Times}(x,y\term{b}) \df 
\term{Conc}(\term{Times}(x,y), x),
$$
with $\term{b}\in\{\term{0},\term{1}\}$.
\item $\term{Sub}(x,y)$ denotes the initial-substring
functions, $x,y \mapsto \Sf(x,y)$,
and is defined by bounded recursion as follows:
$$
\term{Sub}(x,\epsilon) \df \term{Eps}(x) \ \ \ \
\ \ \ \
\term{Sub}(x,y\term{b}) \df \term{BOr}(
\term{Sub}(x,y), \term{Eq}(x,y\term{b})), 
$$
%
with $\term{b}\in \{\term{0},\term{1}\}$.
\end{itemize}

\begin{defn}[Formulas of $\PORl$]
\emph{Formulas of $\PORl$} are equations
$\termOne = \termTwo$, where $\termOne$
and $\termTwo$ are terms of type $s$.
\end{defn}

\paragraph{The Theory $\PORl$.}
We now introduce the theory $\PORl$.
\begin{defn}[Theory $\PORl$]
Axioms of $\PORl$ are the following ones:

\begin{itemize}

\item Defining equations for the constants of $\PORl$:
\small
\begin{align*}
\epsilon x= x\epsilon = x  \ \ \ \ \ \ & \ \ \ \ \ \ 
x(y\term{b}) = (xy)\term{b} \\
\\
\term{Tail}(\epsilon) = \epsilon \ \ \ \ \ \ & \ \ \ \ \ \ 
\term{Tail}(x\term{b}) = x \\
\\
\term{Trunc}(x,\epsilon) = 
\term{Trunc}(\epsilon, x) &= \epsilon \\
\term{Trunc}(x\term{b}, y\term{0}) =
\term{Trunc}(x\term{b}, y\term{1})
&= 
\term{Trunc}(x,y)\term{b} \\
\\
\term{Cond}(\epsilon, y,z,w) = y \ \ \ \ \ \ \
\term{Cond}(x\term{0}, y,z,w) &= z \ \ \ \ \ \ \ 
\term{Cond}(x\term{1},y,z,w) = w \\
\\
\term{Bool}(\term{Flipcoin}(x)) &= \term{1} \\
\\
\term{Rec}(x,h_0,h_1,k,\epsilon) &= x \\
\term{Rec}(x,h_0,h_1,k,y\term{b}) &= 
\term{Trunc}(h_by(\term{Rec}(x,h_0,h_1,k,y)),ky), \\
%
\end{align*}
\normalsize
where $\term{b}\in\{\term{0},\term{1}\}$ and 
$b\in\{0,1\}$.\footnote{When if $\term{b}=\term{0}$,
then $b=0$ and $\term{b}=\term{1}$, then $b=1$.}

\item The $(\beta)$- and $(\nu)$-axioms:
\small
\begin{align*}
\Chole \big[(\lambda x.\termOne)\termTwo\big]
&=
\Chole \big[\termOne\{\termTwo/x\}\big] 
\tag{$\beta$} \\
\Chole\big[\lambda x.\termOne x\big] &=
\Chole\big[\termOne\big] \tag{$\nu$}.
\end{align*}
\normalsize
where $\Chole\big[\cdot\big]$ indicates a context
with a unique occurrence of the hole $\big[ \ \big]$,
so that $\Chole\big[\termOne\big]$
denotes the variable capturing replacement of
$\big[ \ \big]$ by $\termOne$ in $\Chole\big[ \ \big]$.
\end{itemize}
The inference rules of $\PORl$ are the following
ones:
\small
\begin{align*}
\termOne = \termTwo &\vdash \termTwo = \termOne
\tag{R1} \\
\termOne = \termTwo, \termTwo = \termThree &\vdash
\termOne = \termThree \tag{R2} \\
\termOne = \termTwo &\vdash \termThree\{\termOne/x\}
= \termThree\{\termTwo/x\} \tag{R3} \\
\termOne = \termTwo &\vdash
\termOne\{\termThree/x\} =
\termTwo\{\termThree/x\}. \tag{R4}
\end{align*}
\normalsize
\end{defn}
\noindent
As predictable, $\vdash_{\PORl} \termOne=\termTwo$
expresses that the equation $\termOne=\termTwo$
is deducible using instances of the axioms above plus
 inference rules (R1)-(R4).
Similarly, given any set $\TSet$ of equations,
$\TSet \vdash_{\PORl} \termOne=\termTwo$
expresses that the equation $\termOne=\termTwo$
is deducible using instances of the quoted axioms and
rules together with equations from $\TSet$.

\paragraph{Relating $\POR$ and $\PORl$.}
For any string $\str \in\Ss $, let $\ooverline{\str}:s$
denote the term of $\PORl$ corresponding
to it, that is:
\begin{center}
$\ooverline{\eepsilon}=\epsilon \ \ \ \ \ \ \ \ \ \ \ \ \ \ \ \ \ \ \ \
\ \ \ \
\ooverline{\str\zero} = \ooverline{\str}\term{0} \ \ \ \ \ \ \ \ \ \ \ \ \ \ \ \ 
\ \ \ \ \ \ \ \ 
\ooverline{\str \one} = \ooverline{\str} \term{1}.$
\end{center}
For any $\eta\in \Os$, let $\TSet_\eta$
be the set of all equations of the form 
$\term{Flipcoin}(\ooverline{\str})= \ooverline{\eta(\str)}$.

\begin{defn}[Provable Representability]\label{df:provRepr}
Let $f:\Os\times \Ss^j \to \Ss$.
A term $\termOne: s \arrowT \dots \arrowT s$
of $\PORl$ \emph{provably represents f}
when for all strings $\str_1,\dots, \str_j, \str \in\Ss$,
and $\eta\in\Os$,
$$
f(\str_1,\dots, \str_j,\eta) = \str 
\ \ \ \text{iff} \ \ \ 
\TSet_\eta \vdash_{\PORl} \termOne
\ooverline{\str_1} \dots \ooverline{\str_j} = \ooverline{\str}.
$$
\end{defn}

\begin{ex}\label{ex:Flipcoin}
The term $\term{Flipcoin} : s \arrowT s$ provably
represents the query function $\query(x,\eta)=\eta(x)$ of $\POR$,
since for any $\str\in \Ss$ and $\eta \in \Os$,
$$
\term{Flipcoin}(\ooverline{\str}) =
\ooverline{\eta(\str)} \vdash_{\PORl}
\term{Flipcoin}(\ooverline{\str}) = \ooverline{\query(\str,\eta)}.
$$
\end{ex}
\noindent
We now consider some of the terms described above
and show them to
provably represent the intended functions.
Let $Tail(\str,\eta)$ indicate the string obtained by
chopping the first digit of $\str$, 
and
$Trunc(\str_1,\str_2,\eta) = \str_1 |_{\str_2}$.

\begin{lemma}
Terms $\term{Tail}, \term{Trunc}$ and $\term{Cond}$
provably represent the functions $Tail,$ $Trunc,$
and $\Cf$, respectively.
\end{lemma}

\begin{proof}[Proof Sketch]
For $\term{Tail}$ and $\term{Cond}$, the claim follows immediately from the defining axioms
of the corresponding constants.
For example, if $\str_1 = \str_2\zero$,
then
$Tail(\str_1,\eta)=\str_2$ (and
$\ooverline{\str_1}=\ooverline{\str_2\zero}
= \ooverline{\str_2}\term{0}$).
Using the defining axioms of $\term{Tail}$,
$$
\vdash_{\PORl} \term{Tail}(\ooverline{\str_1}) =
\term{Tail}(\ooverline{\str_2 \zero}) = \ooverline{\str_2}.
$$
For $Trunc$ by double induction on
$\str_1,\str_2\in \Ss$ we conclude:
$$
\vdash_{\PORl} \term{Trunc}(\ooverline{\str_1},\ooverline{\str_2}) =
\ooverline{\str_1}|_{\ooverline{\str_2}}.
$$
\end{proof}
\noindent
We generalize this result by Theorem~\ref{thm:PRepr} below.

\begin{theorem}\label{thm:PRepr}
\begin{enumerate}
\itemsep0em

\item Any function $f\in\POR$ is provably represented
by a term $\termOne \in \PORl$.

\item For any term $\termOne\in\PORl$,
there is a function $f\in\POR$
such that $f$ is provably represented by $\termOne$.
\end{enumerate}
\end{theorem}

\begin{proof}[Proof Sketch]
1.
The proof is by induction on the structure of
$f\in\POR$:

\emph{Base Case.}
Each base function is provably representable.
Let us consider two examples only:

\begin{itemize}
\itemsep0em

\item \emph{Empty Function.}
$f=\Ef$ is provably represented by
the term $\lambda x.\epsilon$.
For any string $\str\in \Ss$,
$\ooverline{\Ef(\str,\eta)}=\ooverline{\eepsilon}
= \epsilon$ holds.
Moreover,
$\vdash_{\PORl} (\lambda x.\epsilon)\ooverline{\str}
= \epsilon$
is an instance of the $(\beta)$-axiom.
So, we conclude:
$$
\vdash_{\PORl} (\lambda x.\epsilon) \ooverline{\str}
= \ooverline{\Ef(\str,\eta)}.
$$

\item \emph{Query Function.} $f=\query$ is provably represented by
the term $\term{Flipcoin}$, as observed in Example~\ref{ex:Flipcoin}
above.
\end{itemize}

\emph{Inductive case.}
Each function defined by composition
or bounded recursion from provably represented 
functions 
is provably represented as well.
We consider the case of bounded recursion.
Let $f$ be defined as:
\begin{align*}
f(\str_1,\dots, \str_n,\eepsilon,\eta) &=
g(\str_1,\dots, \str_n,\eta) \\
f(\str_1,\dots, \str_n, \str\zero,\eta)
&= h_0(\str_1,\dots, \str_n,\str, f(\str_1,\dots,
\str_n,\str,\eta),\eta) |_{k(\str_1,\dots, \str_n,\str)} \\
f(\str_1, \dots, \str_n, \str\one, \eta)
&=
h_1(\str_1,\dots, \str_n,\str, f(\str_1,\dots,
\str_n,\str,\eta),\eta) |_{k(\str_1,\dots, \str_n,\str)}.
\end{align*}
By IH, $g,h_0,h_1$, and $k$ are provably represented
by the corresponding terms $\term{t_{g}}, \term{t_{{h}_1}},
\term{t_{{h}_{2}}},\term{t_{k}}$, respectively.
So, for any $\str_1,\dots, \str_{n+2},\str\in\Ss$
and $\eta\in \Os$,
we derive:
\begin{align*}
\TSet_\eta &\vdash_{\PORl}
\term{t_g} \ooverline{\str_1} \dots \ooverline{\str_n}
=
\ooverline{g(\str_1,\dots, \str_n,\eta)} 
\tag{$\term{t_g}$} \\
\TSet_\eta &\vdash_{\PORl}
\term{t_{h_0}} \ooverline{\str_1} \dots
\ooverline{\str_{n+2}} =
\ooverline{h_0(\str_1,\dots, \str_{n+2},\eta)}
\tag{$\term{t_{h_0}}$} \\
\TSet_\eta &\vdash_{\PORl}
\term{t_{h_1}}\ooverline{\str_1} \dots
\ooverline{\str_{n+2}} = \ooverline{h_1(\str_1,\dots,
\str_{n+2},\eta)} 
\tag{$\term{t_{h_1}}$} \\
\TSet_{\eta} &\vdash_{\PORl}
\term{t_k}\ooverline{\str_1} \dots \ooverline{\str_n}
= \ooverline{k(\str_1,\dots, \str_n,\eta)}.
\tag{$\term{t_k}$}
\end{align*}
We can prove by induction on $\str$ that,
$$
\TSet_\eta \vdash_{\PORl} \term{t_f}
\ooverline{\str_1} \dots \ooverline{\str_n}
\ooverline{\str} =
\ooverline{f(\str_1,\dots, \str_n,\eta)},
$$
where 
$$
\term{t_f} : \lambda x_1\dots \lambda x_n.
\lambda x.\term{Rec}(\term{t_g} x_1 \dots
x_n, \term{t_{h_0}} x_1 \dots x_n,
\term{t_{h_1}} x_1\dots x_n,
\term{t_k}x_1\dots x_n,x).
$$
Then,
\begin{itemize}
\itemsep0em

\item if $\str=\eepsilon$, then
$f(\str_1,\dots, \str_n,\str,\eta)=
g(\str_1,\dots, \str_n,\eta)$.
Using the $(\beta)$-axiom, we deduce,
$$
\vdash_{\PORl} \term{t_f} \ooverline{\str_1}
...
\ooverline{\str_n} \ooverline{\str}
= \term{Rec}(\term{t_g} \ooverline{\str_1}
... \ooverline{\str_n},
\term{t_{h_0}}\ooverline{\str_1} ...
\ooverline{\str_n},
\term{t_{h_1}} \ooverline{\str_1} ...
\ooverline{\str_n},
\term{t_k} \ooverline{\str_1} ...
\ooverline{\str_n}, \ooverline{\str})
$$
and using the axiom 
$\term{Rec}(\term{t_g}x_1\dots x_n,
\term{t_{h_0}} x_1\dots x_n,
\term{t_{h_1}} x_1\dots x_n,
\term{t_k} x_1\dots$ 
$x_n,\epsilon)
= \term{t_g}x_1\dots x_n$,
we obtain,
$$
\vdash_{\PORl} \term{t_f}\ooverline{\str_1}
\dots \ooverline{\str_n} \ooverline{\str}
= \term{t_g}\ooverline{\str_1}\dots
\ooverline{\str_n},
$$
by (R2) and (R3).
We conclude using ($\term{t_g}$) together
with (R2).

\item if $\str=\str_m\zero$,
then $f(\str_1,\dots, \str_n,\str,\eta)=
h_0(\str_1,\dots, \str_n,\str_m, f(\str_1,
\dots,$ $\str_n,\str,\eta),$ $\eta)|_{k(\str_1,\dots, \str_n,\str_m)}$.
By IH,
we suppose,
$$
\TSet_\eta \vdash_{\PORl}
\term{t_f}\ooverline{\str_1}\dots
\ooverline{\str_n}\ooverline{\str_m}
= \ooverline{f(\str_1,\dots, \str_n,\str',\eta)}.
$$
Then, using the $(\beta)$-axiom 
$\term{t_f} \ooverline{\str_1}\dots \ooverline{\str_n}
\ooverline{\str} = \term{Rec}(\term{t_g}$
$\ooverline{\str_1} \dots \ooverline{\str_n},
\term{t_{h_0}}\ooverline{\str_1} \dots \ooverline{\str_n},$
$\term{t_{h_1}} \ooverline{\str_1} \dots
\ooverline{\str_n}, \term{t_k}\ooverline{\str_1}
\dots \ooverline{\str_n}, \ooverline{\str})$
the axiom $\term{Rec}(g,h_0,h_1,k,x\term{0})$
= $\term{Trunc}(h_0x (\term{Rec}$ $(g,h_0,h_1,k,\term{0})), kx)$ and IH we deduce,
$$
\vdash_{\PORl} \term{t_f}\ooverline{\str_1}
... \ooverline{\str_n} \ooverline{\str}
= \term{Trunc}(\term{t_{h_0}}\ooverline{\str_1}
... \ooverline{\str_n} \ooverline{\str_m}
\ooverline{f(\str_1, ..., \str_n,\str_m,\eta)},
\term{t_k} \ooverline{\str_1} ...
\ooverline{\str_n})
$$
by (R2) and (R3).
Using ($\term{t_{h_0}}$) 
and ($\term{t_k}$)
we conclude using (R3) and (R2):
$$
\vdash_{\PORl} \term{t_f}
\ooverline{\str_1} ... \ooverline{\str_n}
\ooverline{\str} =
\ooverline{h_0(\str_1,..., \str_n,\str_m,
f(\str_1,..., \str_n,\str_m,\eta))|_{k(\str_1,...,
\str_n,\str_m)}}.
$$

\item the case $\str=\str_m\one$ is proved in
a similar way.
\end{itemize}

2. It is a consequence of the normalization
property for the simply typed $\lambda$-calculus:
a $\beta$-normal term $\term{t} : s \arrowT \dots
\arrowT s$ cannot contain variables of higher
types.
By exhaustively inspecting possible normal forms one can check
that these all represent functions in $\POR$.
\end{proof}

\begin{cor}\label{cor:provablyRepresented}
For any function $f:\Ss^j\times \Os \to \Ss$,
$f\in\POR$ when $f$ is provably represented
by some term $\term{t} : s \arrowT \dots \arrowT s
\in \PORl$.
\end{cor}

\subsubsection{The Theory $\IPOR$}

We introduce a first-order \emph{intuitionistic}
theory, called $\IPOR$, which extends
$\PORl$ with basic predicate calculus
and a restricted induction principle.
We also define $\IRS$ as a variant of $\RS$
having the intuitionistic – rather than classical –predicate
calculus as its logical basis.
All theorems of $\PORl$ and $\IRS$
are provable in $\IPOR$.
In fact, $\IPOR$ can be seen as an extension
of $\PORl$, and
 provides a language
to associate derivations in $\IRS$ with 
poly-time computable functions
(corresponding to terms of $\IPOR$). 

\paragraph{The Syntax of $\IPOR$.}
The equational theory $\PORl$ is rather weak.
In particular, even simple equations,
such as $x=\term{Tail}(x)\term{B}(x)$,
cannot be proved in it.
Indeed, induction is needed:
\begin{align*}
&\vdash_{\PORl} \epsilon = \epsilon\epsilon 
= \term{Tail}(\epsilon) \term{B}(\epsilon) \\
x= \term{Tail}(x) \term{B}(x) &\vdash_{\PORl}
x\term{0} = \term{Tail}(x\term{0})\term{B}(x\term{0}) \\
x=\term{Tail}(x) \term{B}(x) &\vdash_{\PORl}
x\term{1} = \term{Tail}(x\term{1})\term{B}(x\term{1}).
\end{align*}
From this we would like to deduce, by induction,
that $x=\term{Tail}(x)\term{B}(x)$.
Thus, we introduce $\IPOR$, the language of which extends
that of $\PORl$ with
(a translation for) all expressions of $\RS$.
In particular, the grammar for terms of $\IPOR$
is precisely the same as that of Definition~\ref{df:termsPORl},
while that for formulas is defined below.

\begin{defn}[Formulas of $\IPOR$]
\emph{Formulas of $\IPOR$} are defined
as follows:
(i.) all equations of $\PORl$ $\termOne=\termTwo$,
are formulas of $\IPOR$;
(ii.) for any (possibly open) term of $\PORl$,
say
$\termOne$ and $\termTwo$, $\termOne \subseteq \termTwo$
and $\Flip(\termOne)$
are formulas of $\IPOR$;
(iii.) formulas of $\IPOR$ are closed under
$\wedge,\vee,\rightarrow, \forall, \exists$.
\end{defn}

\begin{notation}
We adopt the standard conventions:
$\bot\df \term{0} = \term{1}$ 
and 
$\neg \fone \df \fone \rightarrow \bot$.
\end{notation}
\noindent
The notions of $\Sigma^b_0$- and $\Sigma^b_1$-formula
of $\IPOR$ are precisely those for $\RS$,
as introduced in Definition~\ref{df:Sigmab1}.

\begin{remark}
Any formula of $\RS$ can be seen as a formula
of $\IPOR$, where each occurrence
of the symbol $\zzero$ is replaced by $\term{0}$,
of $\oone$ by $\term{1}$,
of $\conc$ by $\circ$ (usually omitted),
of $\times$ by $\term{Times}$.
In the following, we assume that any formula
of $\RS$ is a formula of $\IPOR$, modulo the
substitutions defined above.
\end{remark}

\begin{defn}[The Theory $\IPOR$]
The axioms of $\IPOR$
include standard rules of the intuitioinstic first-order
predicate calculus,
usual rules for the equality symbol,
and axioms below:
\begin{enumerate}

\itemsep0em

\item All axioms of $\PORl$

\item $x\subseteq y \leftrightarrow \term{Sub}(x,y)=\term{1}$

\item $x=\epsilon \vee x=\term{Tail}(x)\term{0}
\vee x=\term{Tail}(x)\term{1}$

\item $\term{0}=\term{1} \rightarrow x=\epsilon$

\item $\term{Cond}(x,y,z,w)=w' \leftrightarrow
(x=\epsilon \wedge w'= y) \vee (x=\term{Tail}(x)\term{0}
\wedge w' = z) \vee
(x=\term{Tail}(x) \term{1} \wedge w'=w)$

\item $\Flip(x) \leftrightarrow \term{Flipcoin}(x)=\term{1}$

\item Any formula of the form,
$$
\big(\fone(\epsilon) \wedge (\forall x)(\fone(x)
\rightarrow \fone(x\term{0})) \wedge
(\forall x)(\fone(x) \rightarrow \fone(x\term{1}))\big)
\rightarrow (\forall y)\fone(y),
$$
where $\fone$ is of the form $(\exists z\preceq \termOne)\termTwo=\termThree$,
with $\termOne$ containing only first-order
open variables.
\end{enumerate}
\end{defn}

\begin{notation}[$\NP$-Predicate]
We refer to a formula
of the form $(\exists z \preceq \termOne)\termTwo=\termThree$,
with $\termOne$ containing only first-order open variables,
as an \emph{$\NP$-predicate}.
\end{notation}
\noindent

\paragraph{Relating $\IPOR$ with $\PORl$ and $\IRS$.}
Now that $\IPOR$ has been introduced we show that
theorems  of both $\PORl$ and 
the intuitionistic version of $\RS$ are derived
in it.
First, Proposition~\ref{prop:PORltoIPOR}
is easily easily established by inspecting
all rules of $\PORl$.

\begin{prop}\label{prop:PORltoIPOR}
Any theorem of $\PORl$ is a theorem of $\IPOR$.
\end{prop}

Then,
we consider  $\IRS$
and establish that every theorem  in it
is derivable in $\IPOR$. 
To do so, we  prove a
few  properties concerning $\IPOR$.
In particular, its recursion schema 
differs from that of $\IRS$
as dealing with formulas of the form $(\exists y\preceq \termOne)\termTwo = \termThree$
and not with all the $\Sigma^b_1$-ones.
The two schemas are related by
Proposition~\ref{prop:Sigmab0}, proved by induction
on the structure of formulas.\footnote{For further details,
see Appendix~\cite{appendix:taskB}.}

\begin{prop}\label{prop:Sigmab0}
For any $\Sigma^b_0$-formula $\fone(x_1,\dots,x_n)$
in $\rl$, 
there exists a term $\term{t_\fone}(x_1,\dots, x_n)$
of $\PORl$ such that:
\begin{enumerate}
\itemsep0em
\item $\vdash_{\IPOR} \fone \leftrightarrow
\term{t_\fone} =\term{0}$
\item $\vdash_{\IPOR} \term{t_\fone} = \term{0}
\vee \term{t_\fone} = \term{1}$.
\end{enumerate}
\end{prop}
\noindent
This leads us to the following corollary
and allows us to prove Theorem~\ref{thm:IRStoIPOR}
relating $\IPOR$ and $\IRS$.

\begin{cor}
\begin{itemize}
\itemsep0em
\item[i.] For any $\Sigma^b_0$-formula
$\fone$, $\vdash_{\IPOR} \fone \vee \neg\fone$.

\item[ii.] For any closed $\Sigma^b_0$-formula
 of $\rl$ $\fone$ and $\eta \in \Os$,
 either $\TSet_\eta \vdash_{\IPOR} \fone$
 or $\TSet_\eta \vdash_{\IPOR} \neg\fone$.
\end{itemize}
\end{cor}

\begin{theorem}\label{thm:IRStoIPOR}
Any theorem of $\IRS$ is a theorem of $\IPOR$.
\end{theorem}

\begin{proof}
First, observe that, as a consequence of Proposition~\ref{prop:Sigmab0},
for any $\Sigma^b_1$-formula 
$\fone=(\exists x_1\preceq t_1)\dots (\exists x_n
\preceq t_n)\ftwo$ in $\rl$,
$$
\vdash_{\IPOR} \fone \leftrightarrow
(\exists x_1\preceq \term{t}_1) \dots
(\exists x_n \preceq \term{t}_n) \term{t_\ftwo}
= \term{0},
$$
\noindent
any instance of the $\Sigma^b_1$-recursion
schema of $\IRS$ is derivable in $\IPOR$
from the $\NP$-induction schema.
Then, 
it suffices to check that all basic axioms
of $\IRS$ are provable in $\IPOR$.
\end{proof}
\noindent
This result also leads to the following
straightforward consequences.

\begin{cor}\label{cor:Sigmab0}
For any closed $\Sigma^b_0$-formula $\fone$
and $\eta\in\Os$,
either $\TSet_\eta \vdash_{\IPOR} \fone$
or $\TSet_\eta \vdash_{\IPOR} \neg \fone$.
\end{cor}
\noindent
Due to Corollary~\ref{cor:Sigmab0},
we can even establish the following Lemma~\ref{lemma:Sigmab0IRS}.

\begin{lemma}\label{lemma:Sigmab0IRS}
Let $\fone$ be a closed $\Sigma^b_0$-formula
of $\rl$ and $\eta\in\Os$, then:
$$
\TSet_\eta \vdash_{\IPOR} \fone \ \ \
\text{iff} \ \ \ 
\eta \in \model{\fone}.
$$
\end{lemma}
\begin{proof}
$(\Rightarrow)$ This soundness result is established
by induction on the structure of rules for $\IPOR$.

$(\Leftarrow)$ For Corollary~\ref{cor:Sigmab0},
we know that either $\TSet_\eta \vdash_{\IPOR}
\fone$ or $\TSet_\eta \vdash_{\IPOR} \neg\fone$.
Hence, if $\eta\in\model{\fone}$, then it cannot
be $\TSet_\eta \vdash_{\IPOR} \neg \fone$
(by soundness).
So, we conclude $\TSet_\eta \vdash_{\IPOR}\fone$.
\end{proof}

\subsubsection{Realizability}

Here, we introduce realizability as internal
to $\IPOR$.
As a corollary, we obtain that from any derivation in
$\IRS$  – actually, in $\IPOR$ –
of a formula in the form $(\forall x)(\exists y)\fone(x,y)$,
one can extract a functional term of $\PORl$
$\term{f} : s\arrowT s$,
such that $\vdash_{\IPOR}(\forall x)\fone(x,\term{f}x)$.
This allows us to conclude that if a function
$f$ is $\Sigma^b_1$-representable in $\IRS$,
then $f\in \POR$.

\begin{notation}
Let $\varO, \varT$ denote finite sequences
of term variables, (resp.) $x_1,\dots, x_n$ and
$y_1,\dots, y_k$, and $\varO(\varT)$ be an
abbreviation for $y_1(\varO),\dots, y_k(\varO)$.
Let $\Lambda$ be a shorthand for the empty sequence
and $y(\Lambda) \df y$.
\end{notation}

\begin{defn}\label{df:realizability}
Formulas $x\realize\fone$
are defined by induction as follows:
\begin{align*}
\Lambda \realize \fone &\df \fone \\
\varO, \varT \realize (\ftwo \wedge
\fthree) &\df
(\varO \realize \ftwo)  \wedge
(\varT \realize \fthree) \\
z,\varO,\varT \realize (\ftwo \vee \fthree)
&\df
(z=\term{0} \wedge \varO \realize \ftwo)
\vee 
(z \neq \term{0} \wedge \varT \realize 
\fthree) \\
\varT \realize (\ftwo \rightarrow \fthree)
&\df
(\forall \varO)(\varO \realize \ftwo
\rightarrow \varT(\varO) \realize \fthree) 
\wedge (\ftwo \rightarrow \fthree) \\
z, \varO \realize (\exists y)\ftwo 
&\df
\varO \realize \ftwo\{z/y\} \\
\varO \realize (\forall y)\ftwo
&\df 
(\forall y)(\varO(y) \realize \ftwo),
\end{align*}
where no variable in $\varO$
occurs free in $\fone$.
Given terms $\termO = \term{t_1},
\dots, \term{t_n}$, we let:
$$
\termO \realize \fone \df 
(\varO \realize \fone) \{\termO/\varO\}.
$$
\end{defn}
\noindent
We relate the derivability of these
new formulas with that of formulas
of $\IPOR$.
Proofs below are by induction, respectively,
on the structure of $\IPOR$-formulas 
and on the height of derivations.

\begin{theorem}[Soundness]\label{thm:realSoundness}
If $\vdash_{\IPOR} \termO \realize \fone$,
then $\vdash_{\IPOR}\fone$.
\end{theorem}

\begin{notation}
Given $\Gamma=\fone_1,\dots, \fone_n$,
let $\varO \realize \Gamma$ be a
shorthand for $\varO_1\realize \fone_1,
\dots,$ $\varO_n \realize \fone_n$.
\end{notation}

\begin{theorem}[Completeness]\label{thm:realCompleteness}
If $\vdash_{\IPOR} \fone$,
then there exist terms $\termO$, such
that $\vdash_{\IPOR} \termO \realize
\fone$.
\end{theorem}
\begin{proof}[Proof Sketch]
We prove that if $\Gamma \vdash_{\IPOR}
\fone$, then
there exist terms $\termO$ such that 
$\varO \realize \Gamma \vdash_{\IPOR}
\termO \varO_1\dots \varO_n \realize
\fone$.
The proof is by induction on the derivation
of $\Gamma\vdash_{\IPOR} \fone$.
Let us  consider just the case of
rule $\vee R_1$ as an example:
\begin{prooftree}
\AxiomC{$\vdots$}
\noLine
\UnaryInfC{$\Gamma \vdash \ftwo$}
\RightLabel{$\vee R_1$}
\UnaryInfC{$\Gamma \vdash \ftwo\vee \fthree$}
\end{prooftree}
By IH, there exist terms $\termT$,
such that $\termO \realize \Gamma
\vdash_{\IPOR} \termO\termT \realize
\ftwo$.
Since $x,y\realize \ftwo \vee \fthree$
is defined as 
$(x=\term{0} \wedge y \realize \ftwo)
\vee (x\neq \term{0} \wedge y
\realize \fthree)$,
we can take $\termO=\term{0},\termT$.
\end{proof}

\begin{cor}\label{cor:IPORSigma}
Let $(\forall x)(\exists y)\fone(x,y)$
be a closed theorem of $\IPOR$,
where $\fone$ is a $\Sigma^b_1$-formula.
Then, there exists a closed term
$\term{t}: s\arrowT s$ of $\PORl$
such that:
$$
\vdash_{\IPOR} (\forall x)\fone(x,\term{t}x).
$$
\end{cor}
\begin{proof}
By Theorem~\ref{thm:realCompleteness},
there exist $\termO=\term{t},w$ such that
$\vdash_{\IPOR} \termO \realize (\forall x)(\exists y)\fone(x,y)$.
So, by Definition~\ref{df:realizability},
\begin{align*}
\termO \realize (\forall x)(\exists y)\fone(x,y)
&\equiv
(\forall x)(\termO (x) \realize
(\exists y)\fone (x,y)) \\
&\equiv 
(\forall x)(w(x) \realize \fone(x,\term{t}x)).
\end{align*}
From this, by Theorem~\ref{thm:realSoundness}, 
we deduce,
$$
\vdash_{\IPOR} (\forall x) \fone (x,\term{t}x).
$$
\end{proof}

\paragraph{Functions which are $\Sigma^b_1$-Representable in $\IRS$ are in $\POR$.}
Now, we have all the ingredients to prove
that if a function is $\Sigma^b_1$-representable
in $\IRS$, in the sense of Definition~\ref{df:representability},
then it is in $\POR$.

\begin{cor}\label{cor:Sigmab1toPOR}
For any function $f:\Os \times \Ss \to \Ss$,
if there is a closed $\Sigma^b_1$-formula in $\rl$
$\fone(x,y)$, such that:
\begin{enumerate}
\itemsep0em
\item $\IRS \vdash (\forall x)(\exists !y) \fone(x,y)$
\item $\model{\fone(\ooverline{\str_1}, \ooverline{\str_2})} =
\{\eta \ | \ f(\str_1,\eta) = \str_2\}$,
\end{enumerate}
then $f\in\POR$.
\end{cor}

\begin{proof}
Since $\vdash_{\IRS} (\forall x)(\exists !y)\fone(x,y)$,
by Theorem~\ref{thm:IRStoIPOR}
$\vdash_{\IPOR} (\forall x)(\exists !y) \fone(x,y)$.
Then, from $\vdash_{\IPOR} (\forall x)(\exists y)\fone(x,y)$
we deduce $\vdash_{\IPOR}(\forall x)\fone(x,\term{g}x)$
for some closed term $\term{g}\in\PORl$, 
by Corollary~\ref{cor:IPORSigma}.
Furthermore, by Theorem~\ref{thm:PRepr}.2,
there is a $g\in \POR$ such that for any $\str_1,\str_2\in\Ss$
and $\eta\in \Os$,
$g(\str_1,\eta)=\str_2$ when
$T_\eta \vdash_{\IPOR} \term{g} \ooverline{\str_1} 
= \ooverline{\str_2}$.
So, by
Proposition~\ref{prop:PORltoIPOR},
for any $\str_1,\str_2\in\Ss$ and $\eta\in\Os$
if $g(\str_1,\eta)=\str_2$, then 
$T_\eta \vdash_{\IPOR} \term{g}\ooverline{\str_1}=
\ooverline{\str_2}$
and so 
$\TSet_\eta \vdash_{\IPOR} 
\fone(\ooverline{\str_1},\ooverline{\str_2})$.
By Lemma~\ref{lemma:Sigmab0IRS},
$\TSet_\eta \vdash_{\IPOR} \fone(\ooverline{\str_1},
\ooverline{\str_2})$
when $\eta\in\model{\fone(\ooverline{\str_1},\ooverline{\str_2})}$,
that is $f(\str_1,\eta)=\str_2$.
But then $f=g$, so since $g\in\POR$ also $f \in \POR$.
\end{proof}

 \footnotesize
 \begin{center}
 \begin{figure}[h!]
 \begin{center}
\framebox{
\parbox[t][7.2cm]{11cm}{
\footnotesize{

 \begin{center}
 \begin{tikzpicture}[node distance=2cm]
\node at (-3,3) (a) {$\vdash_{\IRS} (\forall x)(\exists y)\fone(x,y)$}; 
\node at (-3,1.5) (b) {$\vdash_{\IPOR} (\forall x)(\exists y) \fone(x,y)$}; 
\node at (-3.8,2.3) {\textcolor{gray}{T.~\ref{thm:IRStoIPOR}}};
\node at (-3.8,0.9) {\textcolor{gray}{C.~\ref{cor:IPORSigma}}};
\node at (-3,0.3) (c) {there is $\term{g}\in\PORl$ (closed)};
\node at (-3,0) (c1) {$\vdash_{\IPOR}(\forall x)\fone(x,\term{g}x)$};
\node at (-3.8,-0.8) {\textcolor{gray}{T.~\ref{thm:PRepr}}};
\node at (-3,-1.7) (d) {there is a $g\in\POR$};
\node at (-3,-2) (d1) {$g(\str_1,\eta)
=\str_2 \Leftrightarrow T_\eta 
\vdash_{\PORl}
\term{g} \ooverline{\str_1} = \ooverline{\str_2}$};
\node at (-1.8,-3) (d2) {$T_\eta \vdash_{\IPOR} \term{g}
\ooverline{\str_1} =\ooverline{\str_2}$};
\node at (-2.5,-2.5) (d3) {\textcolor{gray}{P.~\ref{prop:PORltoIPOR}}};

\node at (3,0) (j) {$\vdash_{\IPOR} (\forall x)\fone(x,\ooverline{\str_2})$};
\node at (3,-1) (j1) {$f(\str_1,\eta)=\str_2$};

\node at (3,-2) (h) {$f (=g)\in\POR$};
\node at (3.6,-0.45) (h1) {\textcolor{gray}{L~\ref{lemma:Sigmab0IRS}}};

%
%
%

    \draw[->,thick] (a) to (b);
    \draw[->,thick] (b) to (c);
    \draw[->,thick] (c1) to (d);
    \draw[->,thick] (-1.9,-2.3) to (-1.9,-2.8);
    \draw[->,dotted] (j1) to (h);
    \draw[->,dotted] (-1,0) to (1.2,0);
    \draw[->,dotted] (d2) to (j);
    \draw[->,dotted] (d) to [bend left=20] (h);
    \draw[->,thick] (j) to (j1);

   %
\end{tikzpicture} 
\end{center}
  }}}
\caption{Proof Schema of Corollary~\ref{cor:Sigmab1toPOR}}
\end{center}
\end{figure}
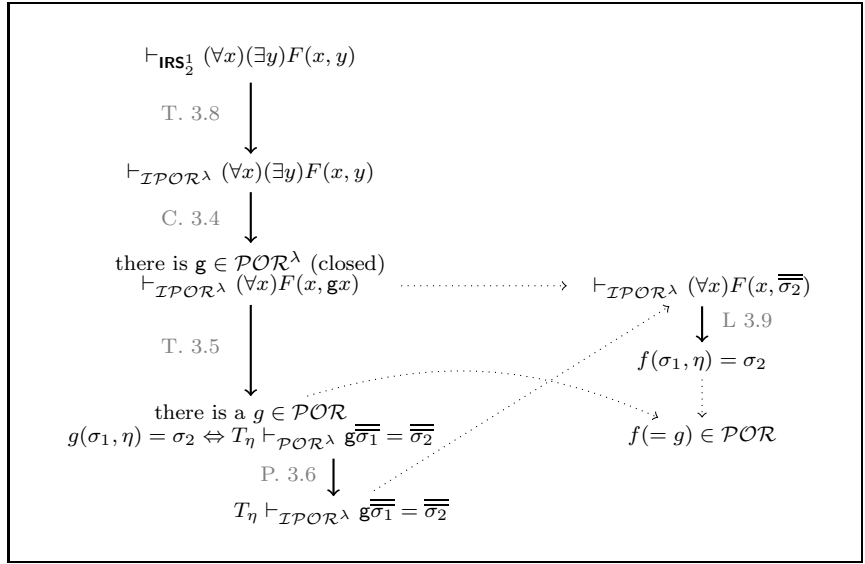  
\end{center}
\normalsize

\subsubsection{$\forall \NP$-Conservativity of $\IPOR$ + $\EM$ over $\IPOR$}

Corollary~\ref{cor:Sigmab1toPOR} is already
very close to the result we are looking for.
The remaining step to conclude our proof is its extension
from intuitioninstic $\IRS$ to classical $\RS$,
showing that any function which is $\Sigma^b_1$-representable 
in $\RS$ is also in $\POR$.
The proof is obtained by adapting the method
from~\cite{CookUrquhart}.
We start by considering an extension
of $\IPOR$ via $\EM$ and
show that the realizability interpretation
extends to it 
so that for any of its closed theorems
$(\forall x)(\exists y\preceq \term{t})\fone(x,y)$,
being $\fone$ a $\Sigma^b_1$-formula,
there is a closed term $\term{t}:s\arrowT s$
of $\PORl$ such that $\vdash_{\IPOR} (\forall x)\fone(x,\term{t}x)$.


\paragraph{From $\IPOR$ to $\IPOR + \Markov$}
Let $\EM$ be the excluded-middle schema,
$\fone \vee \neg \fone$, and 
\emph{Markov's principle} be defined as follows,
\begin{align*}
\neg \neg (\exists x)\fone \rightarrow (\exists x)\fone,
\tag{$\mathbf{Markov}$}
\end{align*}
where $\fone$ is a $\Sigma^b_1$-formula.

\begin{prop}\label{prop:EMtoMarkov}
For any $\Sigma^b_1$-formula $\fone$,
if $\vdash_{\IPOR+\EM} \fone$, then
$\vdash_{\IPOR+\Markov}\fone$.
\end{prop}

\begin{proof}[Proof Sketch]
The claim is proved by applying the double negation
translation,
with the following two remarks: (1)
for any $\Sigma^b_0$-formula $\fone$,
$\vdash_{\IPOR} \neg\neg\fone \rightarrow \fone$;
(2) using $\Markov$, the
double negation of an instance of the $\NP$-induction
can be shown equivalent to an instance
of the $\NP$-induction schema.
\end{proof}
\noindent
We conclude by showing that the realizability
interpretation defined above
extends to $\IPOR+\Markov$, that is for any
closed theorem
$(\forall x)(\exists y\preceq \term{t})\fone(x,y)$
with $\fone$ $\Sigma^b_1$-formula, of $\IPOR+\Markov$, 
there
is a closed term of $\PORl$ $\term{t}:s \arrowT s$,
such that $\vdash_{\IPOR} (\forall x)\fone(x,\term{t}x)$.

\paragraph{From $\IPOR$ to $(\IPOR)^*$.}
Let us assume given a subjective encoding
$\sharp:(s \arrowT s) \arrowT s$ in
$\IPOR$ of first-order unary functions as strings,
together with a ``decoding'' function
$\term{app}: s \arrowT s \arrowT s$ satisfying:
$$
\vdash_{\IPOR} \term{app}(\sharp \term{f},x)
= \term{f}x.
$$
Moreover, let 
$$
x*y \df \sharp(\lambda z.\term{BAnd}(\term{app}(x,z),
\term{app}(y,z)))
$$
and
$$
T(x) \df (\exists y)(\term{B}(\term{app}(x,y))=\term{0}).
$$
There is a \emph{meet semi-lattice}
structure on the set of terms of type $s$
defined by $\term{t}\sqsubseteq \term{u}$
when $\vdash_{\IPOR} T(\term{u})
\rightarrow T(\term{t})$ with top element 
$\TTop \df \sharp(\lambda x.\term{1})$
and meet given by $x*y$.
Indeed, from $T(x*\term{1}) \leftrightarrow
T(x)$, $x\sqsubseteq \TTop$ follows.
Moreover, from $\term{B}(\term{app}(x,\term{u}))=\term{0}$,
we obtain $\term{B}(\term{app}(x*y,\term{u}))=
\term{BAnd}(\term{app}(x,\term{u}), \term{app}(y,\term{u})) = \term{0}$,
whence $T(x) \rightarrow T(x*y)$,
i.e.~$x*y\sqsubseteq x$.
One can similarly prove $x*y\sqsubseteq y$.
Finally, from $T(x) \rightarrow T(v)$
and $T(y)\rightarrow T(v)$,
we deduce $T(x*y)\rightarrow T(v)$,
by observing that $\vdash_{\IPOR}T(x*y)
\rightarrow T(y)$.
Notice that the formula $T(x)$ is \emph{not}
a $\Sigma^b_1$-one, as its existential quantifier
is not bounded.

\begin{defn}
For any formula of $\IPOR$ $\fone$,
and fresh variable $x$, we define
formulas $x\Vdash \fone$ inductively:
\begin{align*}
x \Vdash \fone &\df \fone \vee T(x) \ \ \ (\fone \ atomic) \\
x \Vdash \ftwo \wedge \fthree
&\df x \Vdash \ftwo \wedge x
\Vdash \fthree \\
x \Vdash \ftwo \vee \fthree
&\df
x \Vdash \ftwo \vee
x \Vdash \fthree \\
x \Vdash \ftwo \rightarrow \fthree
&\df
(\forall y)(y\Vdash \ftwo \rightarrow
x*y \Vdash \fthree) \\
x \Vdash (\exists y)\ftwo &\df
(\exists y)x \Vdash \ftwo \\
x \Vdash (\forall y)\ftwo &\df
(\forall y)x \Vdash \ftwo.
\end{align*}
\end{defn}
\noindent
The following Lemma~\ref{lemma:NPind}
is established by induction on the structure
of formulas in $\IPOR$, as in~\cite{CoquandHofmann}.

\begin{lemma}\label{lemma:NPind}
If $\fone$ is provable in $\IPOR$ without using
$\NP$-induction,
then $x\Vdash \fone$ is provable in $\IPOR$.
\end{lemma}

\begin{lemma}\label{lemma:Sigmab0IPOR}
Let $\fone=(\exists x\preceq \term{t})\ftwo$,
where $\ftwo$ is a $\Sigma^b_0$-formula.
Then,
there exists a term $\term{u_\fone}:s$ with
$\FV(\term{u_\fone})=\FV(\ftwo)$ such that:
$$
\vdash_{\IPOR} \fone \leftrightarrow T(\term{u_\fone}).
$$
\end{lemma}

\begin{proof}
Since $\ftwo(x)$ is a $\Sigma^b_0$-formula,
for all terms $\term{v}:s$,
$\vdash_{\IPOR} \ftwo(x) \leftrightarrow \term{t_{x\preceq t\wedge\ftwo}} (x) =\term{0}$,
where $\term{t_{x\preceq t\wedge \ftwo}}$ has
the free variables of $t$ and $\ftwo$.
Let $\fthree(x)$ be a $\Sigma^b_0$-formula,
it is show by induction on its structure that
for any term $\term{v}:s$, 
$\term{t_{\fthree(v)}} = \term{t_\fthree}(\term{v})$.
Then,
$$
\vdash_{\IPOR} \fone \leftrightarrow 
(\exists x)\term{t_{x\preceq t\wedge \ftwo}}(x) =
\term{0} \leftrightarrow (\exists x)T(\sharp(\lambda
x.\term{t_{x\preceq t\wedge \ftwo}}(x))).
$$
So, we let $\term{u_\fone}=\sharp(\lambda x.\term{t_{x\preceq t\wedge \ftwo}}(x))$.
\end{proof}
\noindent 
From which we also deduce the following three properties:
\begin{itemize}
\itemsep0em
\item[i.] $\vdash_{\IPOR} (x\Vdash \fone)
\leftrightarrow (\fone \vee T(x))$

\item[ii.] $\vdash_{\IPOR} (x\Vdash \neg\fone) 
\leftrightarrow (\fone \rightarrow T(x))$

\item[iii.] $\vdash_{\IPOR} (x\Vdash \neg\neg \fone)
\leftrightarrow (\fone \vee T(x))$,
\end{itemize}
where $\fone$ is a $\Sigma^b_1$-formula.

\begin{cor}[Markov's Principle]
If $\fone$ is a $\Sigma^b_1$-formula, then
$$
\vdash_{\IPOR} x \Vdash \neg \neg \fone
\rightarrow \fone.
$$
\end{cor}

To define the extension $(\IPOR)^*$ of $\IPOR$, 
we introduce $\PIND$ in a formal way.

\begin{defn}[PIND]\label{df:PIND}
Let $\PIND(\fone)$ indicate the formula:
$$
\big(\fone(\epsilon) \wedge
\big((\forall x)(\fone(x) \rightarrow \fone(x\term{0}))
\wedge (\forall x)(\fone(x) \rightarrow \fone(x\term{1}))\big)
\rightarrow (\forall x)\fone(x).
$$
\end{defn}
\noindent
Observe that if $\fone(x)$ is a formula
of the form $(\exists y\preceq \termOne)\termTwo=\termThree$, 
then $z\Vdash \PIND(\fone)$
is of the form $\PIND(\fone(x)\vee T(z))$,
which is \emph{not} an instance of the 
$\NP$-induction schema (as the formula
$T(z)=
(\exists x)\term{B}(\term{app}(z,x))=\term{0}$
is not bounded).

\begin{defn}[The Theory $(\IPOR)^*$]
Let $(\IPOR)^*$ indicate the theory
extending $\IPOR$ with all instances
of the induction schema $\PIND(\fone(x)
\vee \ftwo)$, where $\fone(x)$
is of the form $(\exists y\preceq \termOne)
\termTwo=\termThree$,
and $\ftwo$ is an arbitrary formula with
$x\not\in\FV(\ftwo)$.
\end{defn}
\noindent
We then deduce the following Proposition relating 
derivability in $\IPOR$ and in $(\IPOR)^*$.

\begin{prop}
For any $\Sigma^b_1$-formula
$\fone$,
if $\vdash_{\IPOR} \fone$,
then $\vdash_{(\IPOR)^*} x\Vdash \fone$.
\end{prop}
Finally, we  extend realizability 
to $(\IPOR)^*$ by constructing a realizer for 
$\PIND(\fone(x) \vee \ftwo)$.

\begin{lemma}\label{lemma:real}
Let $\fone(x) : (\exists y\preceq \termOne)
\termTwo=\term{0}$ and $\ftwo$
be any formula not containing free occurrences
of $x$.
Then, there exist terms $\termO$
such that:
$$
\vdash_{\IPOR} \termO \realize \PIND(\fone(x)
\vee \ftwo).
$$
So, by Theorem~\ref{thm:realSoundness},
we obtain that for any 
$\Sigma^b_1$-formula $\fone$
and formula $\ftwo$, with $x\not\in \FV(\fone)$,
$$
\vdash_{\IPOR} \PIND(\fone(x) \vee \ftwo).
$$
\end{lemma}

\begin{prop}
For any $\Sigma^b_1$-formula $\fone$ and $\ftwo$
with $x\not\in\FV(\fone)$, $\vdash_{\IPOR} \PIND(\fone(x)\vee \ftwo)$.
\end{prop}

\begin{cor}[$\forall \NP$-Conservativity of 
$\IPOR+\EM$ over $\IPOR$]
Let $\fone$ be a $\Sigma^b_1$-formula,
if $\vdash_{\IPOR+\EM}(\forall x)(\exists y\preceq \termOne) \fone(x,y)$,
then $\vdash_{\IPOR}(\forall x)
(\exists y\preceq \termOne)\fone(x,y)$.
\end{cor}

\paragraph{Concluding the Proof.}
We conclude our proof establishing Proposition~\ref{prop:IPORMarkov}.

\begin{prop}\label{prop:IPORMarkov}
Let $(\forall x)(\exists y\preceq \termOne) \fone(x,y)$
be a closed term of $\IPOR+\Markov$,
where $\fone$ is a $\Sigma^b_1$-formula.
Then, there exists a closed term of $\PORl$ 
$\termOne:s \arrowT s$, such that:
$$
\vdash_{\IPOR} (\forall x)\fone (x,\termOne x).
$$
\end{prop}

\begin{proof}
If $\vdash_{\IPOR+\Markov} (\forall x)(\exists y)\fone(x,y)$, then by Parikh's Proposition~\ref{prop:Parikh}, 
also 
$\vdash_{\IPOR+\Markov} (\exists y\preceq \termOne)\fone(x,y)$.
Moreover, 
$\vdash_{(\IPOR)^*}z\Vdash (\exists y \preceq \termOne) \fone(x,y)$.
Then, let us consider $\ftwo = (\exists y\preceq \termOne)
\fone(x,y)$.
By taking $\termThree=\termTwo_\ftwo$,
using Lemma~\ref{lemma:Sigmab0IPOR},
we deduce $\vdash_{(\IPOR)^*} \ftwo$
and, thus, by Lemma~\ref{lemma:NPind} and~\ref{lemma:real},
we conclude that there exist $\termO,\termT$
such that $\vdash_{\IPOR} \termO,\termT\realize \ftwo$,
which implies $\vdash_{\IPOR} \fone(x,\termO x)$, and so
$\vdash_{\IPOR}(\forall x)(\fone(x),\termO x)$.
\end{proof}
\noindent
So, by Proposition~\ref{prop:EMtoMarkov},
if $\vdash_{\IPOR+\EM} (\forall x)(\exists y\preceq \term{t}) \fone(x,y)$,
being $\fone$ a closed $\Sigma^b_1$-formula,
then there is a closed term of $\POR$ 
$\termOne : s \arrowT s$,
such that $\vdash_{\IPOR} (\forall x)\fone(x,\termOne x)$.
Finally, we conclude the desired
Corollary~\ref{cor:RStoPOR} for classical $\RS$
arguing as above.

\begin{cor}\label{cor:RStoPOR}
Let $\RS \vdash (\forall x)(\exists y\preceq t) \fone (x,y)$, where $\fone$
is a $\Sigma^b_1$-formula with only
$x$ and $y$ free.
For any function $f:\Ss \times \Os \to \Ss$,
if $(\forall x)(\exists y\preceq t)\fone (x,y)$
represents $f$ so that:
\begin{enumerate}
\itemsep0em
\item $\RS \vdash (\forall x)(\exists !y)\fone(x,y)$
\item $\model{\fone(\ooverline{\str_1},\ooverline{\str_2})} = \{\eta \ | \ f(\str_1,\eta)=
\str_2\}$,
\end{enumerate}
then $f\in\POR$.
\end{cor}
\noindent
Now, putting Theorem~\ref{thm:PORtoRS} and Corollary~\ref{cor:RStoPOR} together,
we conclude that Theorem~\ref{thm:PORandRS} holds 
and that $\RS$ provides an \emph{arithmetical} characterization of functions 
in $\POR$.

\section{Relating $\POR$ and Poly-Time PTMs}\label{sec:PORandPTM}

Theorem~\ref{thm:PORandRS} is still not enough
to characterize probabilistic classes, 
which are defined in terms of functions computed by PTMs, 
and, as  observed, there is a crucial difference
between the ways in which these machines and oracle functions
access randomness.
So, our next goal consists in filling this gap,
by relating these two classes in a precise way.\footnote{Actually, 
this proof is particularly convoluted so, for simplicity's sake,
we present here just its skeleton.
The interested reader can however found further details in~\cite{appendix:TaskC},
and the full proof was presented in~\cite{Davoli}.}

\subsection{Preliminaries}

We start by defining (or re-defining)
the classes of functions (over strings) computed by poly-time
PTMs and
of functions computed by poly-time stream machines,
that is TMs with an extra oracle tape.

\subsubsection{The Class $\RFP$}
We start by (re-)defining the class of functions computed
by poly-time PTMs.\footnote{Clearly, there is
a strong affinity with the standard definition of random functions~\cite{Santos69} 
presented for example in Section~\ref{sec:arithStruc}.
Here, pseudo-distributions and functions are over strings rather than
numbers. 
Furthermore, we are now considering machines
explicitly associated with a (polynomial-)time resource bound.}

\begin{defn}[Class $\RFP$]
Let $\dist(\Ss)$ denote the set of 
functions $f:\Ss \to [0,1]$
such that $\sum_{\str\in \Ss}f (\str)=1$.
The \emph{class $\RFP$} is made of all functions 
$f:\Ss^k \to \dist(\Ss)$ such that,
for some PTM $\PTM$ running in polynomial time,
and every $\str_1,\dots, \str_k,\strT\in \Ss$,
$f(\str_1,\dots, \str_k)(\strT)$ coincides with the
probability that $\PTM(\str_1\sharp \dots \sharp \str_k)\Downarrow
\strT$.
\end{defn}
\noindent
So – similarly to $\langle \PTM\rangle$ by~\cite{Santos69,Gill77} –
the function computed by this machine
associates each possible output with a probability
corresponding to the actual probability that a run of the
machine actually produces that output, and
we need to adapt the notion of 
$\Sigma^b_1$-representability accordingly.

\begin{defn}
A function $f:\Ss^k \to \dist(\Ss)$ is
\emph{$\Sigma^b_1$-representable in $\RS$}
if there is a $\Sigma^b_1$-formula of 
$\rl$  $\fone(x_1,\dots, x_k,y)$,
such that:
\begin{enumerate}
\itemsep0em
\item $\RS \vdash (\forall \vec{x})(\exists !y)\fone(\vec{x},y)$,
\item for all $\str_1,\dots, \str_k,\strT\in \Ss$,
$f(\str_1,\dots, \str_k,\strT)=\mu(\model{\fone(\overline{\str_1},
\dots,\overline{\str_k},\strT)})$.
\end{enumerate}
\end{defn}
\noindent
The central result of this chapter can be re-stated as follows: 

\begin{restatable}{theorem}{thm:RSandRFP}\label{thm:RSandRFP}
For any function $f:\Ss^k \to \dist(\Ss)$, $f$
is $\Sigma^b_1$-representable in $\RS$ when $f\in\RFP$.
\end{restatable}
\noindent
The proof of Theorem~\ref{thm:RSandRFP}
relies on Theorem~\ref{thm:PORandRS}, 
once we relate the function
algebra $\POR$ with the class $\RFP$
by the Lemma~\ref{lemma:PORandRFP} below.

\begin{restatable}{lemma}{lemma:PORandRFP}\label{lemma:PORandRFP}
For any functions $f:\Ss^k\times \Os \to \Ss$ in $\POR$,
there exists $g:\Ss \to \dist(\Ss)$ in $\RFP$ such that
for all $\str_1,\dots, \str_k,\strT \in \Ss$,
$$
\mu(\{\omega \ | \ f(\str_1,\dots, \str_k,\eta)=\strT \})=
g(\str_1,\dots,\str_k,\strT)
$$
and viceversa.
\end{restatable}
\noindent
However, the proof of Lemma~\ref{lemma:PORandRFP}
is convoluted,
as it is based on a chain of language simulations.

\subsubsection{Introducing the Class $\SFP$}
The core idea to relate $\POR$ and $\RFP$ is to introduce
an intermediate class, called $\SFP$.
This is the class of functions computed 
by a poly-time \emph{stream Turing machine} (STM, for short),
where
an STM is a deterministic TM with one extra
(read-only) tape intuitively accounting for
probabilistic choices:
at the beginning 
the extra-tape
is sampled from $\Bool^\Nat$;
then, at each computation step, 
the machine reads one new bit from this tape,
always moving to the right.

\begin{defn}[Class $\SFP$]\label{df:SFP}
The \emph{class $\SFP$} is made of functions
$f: \Ss^k \times \Bool^\Nat \to\Ss$, 
such that there is an STM $\STM$ running in polynomial time
such that 
for any $\str_1,\dots, \str_k \in \Ss$ and
$\omega \in \Bool^\Nat$,\footnote{By a slight abuse of notation,
we use $\omega, \omega', \omega'',\dots$ as meta-varibles for sequences
in $\Bool^\Nat$. Indeed, as said in Section~\ref{sec:CPLcpreliminaries}, the sets $\{0,1\}$ and $\{\zero,\one\}$
are basically equivalent.} 
$f(\str_1,\dots,\str_k,\omega)=\tau$
when
for inputs $\str_1\sharp \dots \sharp \str_k$ and tape
$\omega$, the machine $\STM$  outputs $\strT$.
\end{defn}

\subsection{Relating $\RFP$ and $\SFP$}
The global behavior of STMs and PTMs is similar,
but the former access to randomness in an explicit way:
instead of flipping a coin at each step,
the machine samples a stream of bits once, and
then reads one new bit at each step.
So, to prove the equivalence of the two models,
we pass through the following Proposition~\ref{prop:PTMandSTM}.

\begin{prop}[Equivalence of PTMs and STMs]\label{prop:PTMandSTM}
For any poly-time STM $\STM$, there is a poly-time PTM
$\STM^*$ such that for all strings $\str,\strT \in \Ss$,
$$
\mu(\{\omega \ | \ \STM(\str,\omega)=\strT\}) = \Pr[\STM^*(\str)=\strT],
$$
and viceversa.
\end{prop}

\begin{cor}[Equivalence of $\RFP$ and $\SFP$]\label{cor:RFPandSFP}
For any  $f:\Ss^k\to \dist(\Ss)$  in 
$\RFP$, there is a $g:\Ss^k\times \Bool^\Nat
\to \Ss$ in $\SFP$,
such that for all $\str_1,\dots, \str_k,\strT\in \Ss$,
$$
f(\str_1,\dots, \str_k,\strT) = \mu(\{\omega \ | \ 
g(\str_1,\dots,\str_k,\omega)=\strT\}),
$$
and viceversa.
\end{cor}

\subsection{Relating $\SFP$ and $\POR$}
Finally, we need to 
prove the equivalence between $\POR$ and $\SFP$:
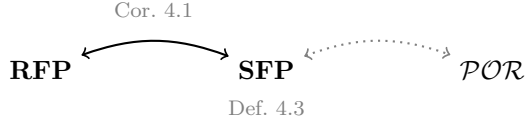
 \begin{center}
 \begin{tikzpicture}[node distance=2cm]
\node at (-3,0) (a) {$\RFP$}; 
\node at (-1.5,0.8) (b1) {\footnotesize{\textcolor{gray}{Cor.~\ref{cor:RFPandSFP}}}};
\node at (0,0) (b) {$\SFP$};
\node at (0,-0.5) (b2) {\footnotesize{\textcolor{gray}{Def.~\ref{df:SFP}}}};
\node at (3,0) (c) {$\POR$};

    \draw[<->,thick] (a) to [bend left=20] (b);
    \draw[<->,thick,dotted,gray] (b) to [bend left=20] (c);

   %
\end{tikzpicture} 
 \end{center}
Moving from PTMs to STMs,
we obtain a machine model 
which accesses randomness in a way 
which is similar to that of functions in $\POR$:
as seen, at the beginning of the computation an oracle
is sampled,
and computation proceeds querying it.
Yet, there are still relevant differences in the way
in which these families of machines treat randomness.
While functions of $\POR$ access an oracle
in the form of a function
$\eta \in \Bool^\Ss$,
the oracle for an STM is a stream of bits
$\omega\in \Bool^\Nat$.
Otherwise said, a function in $\POR$
is of the form $f_{\POR}:\Ss^k \times \Os \to
\Ss$, whereas one in $\SFP$ is 
$f_{\SFP}:\Ss^k\times \Bool^\Nat \to \Ss$.
Then, we cannot compare them \emph{directly}, and
provide an indirect comparison 
in two main steps.

\subsubsection{From $\SFP$ to $\POR$}
First, we show that any function computable by a poly-time
STM is in $\POR$.

\begin{prop}[From $\SFP$ to $\POR$]\label{prop:SFPtoPOR}
For any $f:\Ss^k\times \Bool^\Nat \to\Ss$
in $\SFP$, there is a function
$f^{\star} : \Ss^k \times \Os \to \Ss$ in $\POR$
such that for all $\str_1,\dots,\str_k,\strT \in\Ss$
and $\omega\in \Bool^\Nat$,
$$
\mu(\{\omega \in \Bool^\Nat \ | \ f(\str_1,\dots, \str_k,
\omega) = \strT\}) = 
\mu(\{\eta \in \Os \ | \ 
f^{\star} (\str_1,\dots, \str_k,\eta)=\strT\}).
$$
\end{prop}
\noindent
The fundamental observation is that,
given an input $\str \in\Ss$ and the extra tape
$\omega \in \Bool^\Nat$,
an STM running in polynomial time can access a \emph{finite}
portion of $\omega$ only,
the length of which can be bounded by
some polynomial $p(|\str|)$.
Using this fact, we construct $f^{\star}$
as follows:
\begin{enumerate}
\itemsep0em
\item We introduce the new class $\PTF$,
made of functions 
$f:\Ss^k\times \Ss \to \Ss$ computed by
a \emph{finite stream Turing machine}
(FSTM, for short), 
the extra tape of which is a finite string.
\item 
We define a function $h \in \PTF$
such that for any $f:\Ss\times \Bool^\Nat \to \Ss$
with polynomial bound $p(x)$,
$$
f(n,\omega)=h(x,\omega_{p(|x|)}).
$$

\item We define $h':\Ss \times \Ss \times \Os 
\to \Ss$ such that,
$$
h'(x,y,\eta)=
h(x,y).
$$
By an encoding of FSTMs we show that
$h'\in\POR$. 
Moreover $h'$ can be defined \emph{without}
using the query function,
since the computation of $h'$ never looks at $\eta$.

\item Finally, we define an \emph{extractor function}
$e:\Ss\times \Os \to \Ss \in \POR$,
which mimics the prefix extractor 
$\omega_{p(|x|)}$, having its outputs 
\emph{the same distributions} of all possible $\omega$'s
prefixes,
even though within a different space.\footnote{Recall that $\eta\in \Bool^\Nat$, 
while the second argument of $e$ is in $\Os$.}
This is obtained by exploiting a bijection
$dyad:\Ss\to\Nat$,
ensuring that for each $\omega \in \Bool^\Nat$,
there is an $\eta\in\Bool^\Ss$
such that any prefix of $\omega$ is an output of
$e(y,\eta)$, for some $y$.
Since $\POR$ is closed under composition,
we finally define 
$$
f^\star (x,\eta) \df h'(x,e(x,\eta),\eta).
$$
\end{enumerate}

\subsubsection{From $\POR$ to $\SFP$}
%
In order to simulate functions
of $\POR$ via STMs we observe  not only that
these two models
invoke oracles of different shape,
but also that the former can manipulate such oracles
in a more liberal way: 
\begin{itemize}
\item STMs  query the oracle before each step is 
produced.
By contrast functions of $\POR$ may invoke the query function
$\query(x,\eta)$ freely during computation.
We call this access policy \emph{on demand}.

\item STMs query a new bit of the oracle at
each step of computation, and
cannot access previously observed bits.
We call this access policy \emph{lienear}.
By contrast, functions of $\POR$ can query
the same bits as many times as needed.
\end{itemize}
Consequently, a direct simulation
of $\POR$ via STMs is challenging
even for a basic function like $\query(x,\eta)$.
So, again, we follow an indirect
path:
we pass through a chain of simulations,
dealing with each of these differences separately.
%

\begin{enumerate}
\item First, we translate $\POR$ into
an imperative language $\SIFPra$
inspired by Winskel's IMP~\cite{Winskel},
with the same access policy as $\POR$.
$\SIFPra$ is endowed with assignments,
a $\mathtt{while}$ construct,
and a command $\Flip(e)$, which first evaluates
$e$ to a string $\str$ and then stores the value
$\eta(\str)$ in a register.
The encoding of oracle functions
in $\SIFPra$ is easily obtained by induction
on the function algebra. 

\item Then, we translate $\SIFPra$ into another
imperative language, called $\SIFPla$,
associated with a \emph{linear} policy
of access.
$\SIFPla$ is defined like $\SIFPra$ except for 
$\Flip(e)$, which is replaced by 
the new command $\mathtt{RandBit}()$
generating a random bit and storing it in a register.
A weak simulation from $\SIFPra$
into $\SIFPla$ is defined by progressively constructing
an \emph{associative table}
containing pairs in the form (string, bit) of past
observations.
Each time $\Flip(e)$ is invoked, 
the simulation checks whether a pair $(e,b)$ had
already been observed.
Otherwise increments the table by producing
a new pair $(e,\mathtt{RandBit}())$.
This is by far the most complex step of the whole
simulation.

\item The language $\SIFPla$ can be translated
into STMs.
Observe that the access policy of $\SIFPla$ is
still on-demand:
$\mathtt{RandBit}()$ may be invoked
or not before executing the instruction.
So, we first consider a translation from $\SIFPla$
into a variant of STMs admitting an on-demand
access policy
– that is, a computation step may or may not access
a bit from the extra-tape.
Then, the resulting program is encoded into a regular
STM.
Observe that 
we cannot expect that the machine $\STM^\dagger$
simulating an on-demand machine $\STM$
will produce \emph{the same} output and
oracle. 
%
Rather, as in many other cases, we show that
$\STM^\dagger$ can be defined so that,
for any $\str,\strT\in\Ss$,
the sets 
$\{\omega \ | \ \STM^\dagger (\str,\omega)=\strT\}$
and $\{\omega \ | \ \STM(\str,\omega)=\strT\}$
have the same measure.
\end{enumerate}

 \begin{figure}[h!]\label{fig:PORandSFP}
 \begin{center}
\framebox{
\parbox[t][4cm]{10.6cm}{
\footnotesize{

 \begin{center}
 \begin{tikzpicture}[node distance=2cm]
\node(por) at (-8,0) {$\POR$};
\node(por1) at (-8,-0.5) {$\Ss\times \Os\longrightarrow \Ss$};

\node(sfp) at (0,0) {$\SFP$};
\node(sfp1) at (0,-0.5) {$\Ss\times \Bool^{\Nat}\longrightarrow \Ss$};


\node[black](fsfp) at (-4,0) {finite $\SFP$};
\node[black](fsfp1) at (-4,-0.5) {$\Ss\times \Ss\longrightarrow \Ss$};

\draw[->,thick] (sfp) to node[above]{\tiny$ (x, \omega_{p(|x|)})\mapsfrom (x,\omega)$} (fsfp);
\draw[->,thick] (fsfp) to node[above]{\tiny$(x,\eta)\to (x, e(x,\eta))$} (por);


\node(sifpra) at (-8,-2.5) {\begin{tabular}{c}$\SIFPra$\\ \small imperative \\ \small random access\end{tabular}};

\node(sifpla) at (-3,-2.5) {\begin{tabular}{c}$\SIFPla$\\ \small imperative \\ \small linear access\end{tabular}};

\node(sfpod) at (0,-2.5) {\begin{tabular}{c}$\SFP$\\ \small on-demand\end{tabular}};

\draw[->, thick] (por1) to node[right]{\tiny\begin{tabular}{c}\tiny inductive\\\tiny  encoding\end{tabular}} (sifpra);
\draw[->, thick] (sifpra) to node[above]{\tiny \begin{tabular}{c}\tiny associative\\ \tiny table\end{tabular}} (sifpla);
\draw[->, thick] (sifpla) to (sfpod);
\draw[->, thick] (sfpod) to (sfp1);


\end{tikzpicture} 

 \end{center}
  }}}
\caption{Equivalence between $\POR$ and $\SFP$~\ref{thm:RSandRFP}}
\end{center}
\end{figure}
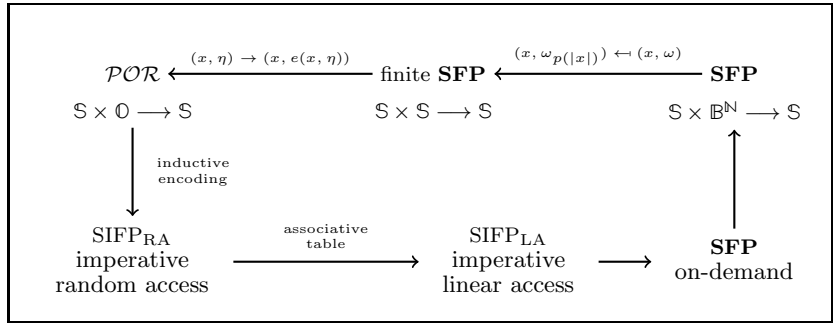

\subsubsection{Concluding the Proof.}
These ingredients are enough to conclude the proof
as outlined in Figure~\ref{fig:sketchRBA},
and to relate poly-time random functions 
and $\Sigma^b_1$-formulas of $\RS$.


 \small
 
 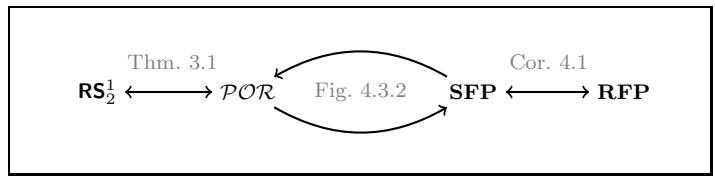
\begin{figure}[h!]\label{fig:sketchRBA}
 \begin{center}
\framebox{
\parbox[t][2cm]{9cm}{
\footnotesize{

 \begin{center}
 \begin{tikzpicture}[node distance=2cm]
 \node at (-3,0.4) (a1) {\textcolor{gray}{Thm.~\ref{thm:PORandRS}}};
\node at (-4,0) (a) {$\RS$}; 
\node at (-2,0) (b) {$\POR$}; 
\node at (2,0.4) (b1) {\textcolor{gray}{Cor.~\ref{cor:RFPandSFP}}};
\node at (1,0) (d) {$\SFP$};
\node at (3,0) (c) {$\RFP$};
\node at (-0.5,0) (g) {\textcolor{gray}{Fig.~\ref{fig:PORandSFP}}};

    \draw[<->,thick] (a) to(b);
    \draw[->,thick] (b) to [bend right=30] (d);
    \draw[->,thick] (d) to [bend right=30] (b);
    \draw[<->,thick] (c) to (d);

   %
\end{tikzpicture} 

 \end{center}
  }}}
\caption{Proof Sketch of Theorem~\ref{thm:RSandRFP}}
\end{center}
\end{figure}  

\normalsize

\bibliographystyle{plain}
\bibliography{bib}

\newpage
\appendix

\section{Proofs from Section~\ref{sec:PORandRS}}\label{appendix:TaskA}

\subsection{Further details on Theorem~\ref{thm:PORtoRS}}

\begin{proof}[Proof of Theorem~\ref{thm:PORtoRS}]
The proof is by induction on the structure 
of $f\in\POR$.

\emph{Base case.}
If $f$ is a basic function, then there are five possible
sub-cases:

\begin{itemize}

\item $f=E$ is $\Sigma^b_1$-represented in $\RS$
by the formula
$F_E(x,y) : x=x \wedge y=\epsilon$.
The existence condition is derived as follows:

\footnotesize
\begin{prooftree}
\AxiomC{}
\RightLabel{$Ax$}
\UnaryInfC{$z=z \vdash (\exists y)(z=z \wedge y=\epsilon),
z=z$}
\RightLabel{$Ref$}
\UnaryInfC{$\vdash(\exists y)(z=z \wedge y=\epsilon),
z=z$}
\AxiomC{}
\RightLabel{$Ax$}
\UnaryInfC{$\epsilon=\epsilon \vdash (\exists y)(z=z \wedge
y=\epsilon),\epsilon=\epsilon$}
\RightLabel{$Ref$}
\UnaryInfC{$\vdash(\exists y)(z=z \wedge y=\epsilon),\epsilon=\epsilon$}
\RightLabel{$\wedge R$}
\BinaryInfC{$\vdash (\exists y)(z=z \wedge y=\epsilon), z=z \wedge
\epsilon=\epsilon$}
\RightLabel{$\exists R$}
\UnaryInfC{$(\exists y) (z=z \wedge y=\epsilon)$}
\RightLabel{$\forall R$}
\UnaryInfC{$\vdash(\forall x)(\exists y)(x=x \wedge y=\epsilon)$}
\end{prooftree}

\normalsize
Uniqueness condition is derived below:

\footnotesize
\begin{prooftree}
\AxiomC{}
\RightLabel{$Ax$}
\UnaryInfC{$x_2=x_3,x_1=x_1,x_2=\epsilon, x_3=\epsilon \vdash
x_2=x_3$}
\RightLabel{$Repl$}
\UnaryInfC{$x_1=x_1,x_2=\epsilon, x_1=x_1,x_3=\epsilon
\vdash x_2 = x_3$}
\RightLabel{$\wedge L$s}
\UnaryInfC{$x_1=x_1\wedge x_2=\epsilon, x_1=x_1 \wedge
x_2=\epsilon \vdash x_2=x_3$}
\RightLabel{$\wedge L$}
\UnaryInfC{$(x_1=x_1 \wedge x_2=\epsilon)
\wedge (x_1=x_1\wedge x_2=\epsilon) \vdash
x_2=x_3$}
\RightLabel{$\rightarrow R$}
\UnaryInfC{$\vdash (x_1=x_1 \wedge x_2=\epsilon)
\wedge (x_1=x_1 \wedge x_2=\epsilon) \rightarrow
x_2=x_3$}
\RightLabel{$\forall Rs$}
\UnaryInfC{$\vdash (\forall x)(\forall y)(\forall z)((x=x
\wedge y=\epsilon) \wedge (x=x \wedge z=\epsilon) \rightarrow
y=z)$}
\end{prooftree}

\normalsize
The semantic condition is proved showing that
for all $n,m\in \Ss$ and $\eta^* \in\Os$,
$E(\str_1,\eta^*)=\str_2$ when $\eta \in \model{\overline{\str_1}=\overline{\str_1}
\wedge \overline{\str_2}=\epsilon}$.

\begin{itemize}
\itemsep0em
\item[$\Rightarrow$] Assume $\str_2=\eepsilon$.
So,
\begin{align*}
\model{\overline{\str_1}=\overline{\str_1} \wedge \overline{\str_2}
=\epsilon} &=
\model{\overline{\str_1}=\overline{\str_1} \wedge \overline{\eepsilon}
= \epsilon} \\
&= \model{\overline{\str_1}=\overline{\str_1} \wedge
\epsilon=\epsilon} \\
&= \model{\overline{\str_1}=\overline{\str_1}} \cap
\model{\epsilon=\epsilon} \\
&= \Os \cap \Os \\
&= \Os.
\end{align*}
Clearly, for any $\eta^*$, $\eta^*\in \Os$,
that is $\eta^* \in \model{\overline{\str_1}=\overline{\str_1}
\wedge \overline{m}=\epsilon}$.

\item[$\Leftarrow$] Assume (by contraposition) $m\neq \eepsilon$.
so,
\begin{align*}
\model{\overline{\str_1}=\overline{\str_1} \wedge
\overline{\str_2}=\epsilon} &=
\model{\overline{\str_1} = \overline{\str_1}} \cap
\model{\overline{\str_2} =\epsilon} \\
&= \Os \cap \emptyset \\
&= \emptyset. 
\end{align*} 
Clearly, for any $\eta^*$, $\eta^* \not \in \emptyset$,
so $\eta^* \not\in \model{\overline{\str_1}=\overline{\str_1}
\wedge \overline{\str_2}=\emptyset}$.

\end{itemize}

\item $f= P^n_i$, for $1\leq i\leq n \in \Nat$.
The function si $\Sigma^b_1$-represented in $\RS$
by the formula:
$\bigwedge_{j\in J}(x_j=x_j) \wedge y=x_i$,
where $J=\{1,\dots, n\}\setminus i$.
The proof is similar to the one above.

\item $f = S_b$.
The function is $\Sigma^b_1$-represented
in $\RS$ by the formula: $y=x\bbool$,
where $\bbool\in\{\zzero,\oone\}$
and $\bbool=\zzero$ when $b=0$
and $\bbool=\oone$ when $b=1$.
The existence condition is derived as follows:

\footnotesize
\begin{prooftree}
\AxiomC{}
\RightLabel{$Ax$}
\UnaryInfC{$z\bbool = z\bbool \vdash (\exists y)(y=z\bbool),
z\bbool = z\bbool$}
\RightLabel{$Ref$}
\UnaryInfC{$\vdash (\exists y)(y=z\bbool), z\bbool = z\bbool$}
\RightLabel{$\exists R$}
\UnaryInfC{$\vdash (\exists y)(y=z\bbool)$}
\RightLabel{$\forall R$}
\UnaryInfC{$\vdash (\forall x)(\exists y)(y=x\bbool)$}
\end{prooftree}

\normalsize
Uniqueness is derived below:

\footnotesize
\begin{prooftree}
\AxiomC{}
\RightLabel{$Ax$}
\UnaryInfC{$x_2=x_3,x_2=x_1\bbool, x_3=x_1\bbool \vdash
x_2=x_3$}
\RightLabel{$Repl$}
\UnaryInfC{$x_2=x_1\bbool, x_3 = x_1\bbool \vdash x_2=x_3$}
\RightLabel{$\wedge L$}
\UnaryInfC{$x_2=x_1\bbool \wedge x_3=x_1\bbool \vdash
x_2=x_3$}
\RightLabel{$\rightarrow R$}
\UnaryInfC{$\vdash (x_2=x_1\bbool) \wedge
(x_3=x_1\bbool) \rightarrow x_2=x_3$}
\RightLabel{$\forall R$s}
\UnaryInfC{$\vdash (\forall x)(\forall y)(\forall z)((y=x\bbool)
\wedge (z=x\bbool) \rightarrow y=z)$}
\end{prooftree}

\normalsize
The semantic condition is proved showing that for
all $\str_1,\str_2 \in \Ss$ and $\eta^* \in \Os$,
$S_b(\str_1,\eta^*) = \str_2$
when $\eta^* \in \model{\overline{\str_2} = \overline{\str_1}\bbool}$.

\begin{itemize}
\itemsep0em
\item[$\Rightarrow$]
Assume $b=0$.
So, $\str_2=\str_1\zero$ (and
$\bbool = \zzero$).
Then,
\begin{align*}
\model{\overline{\str_2}=\overline{\str_1}\zzero}
&=
\model{\overline{\str_1\zero} = \overline{n}\zzero} \\
&= \model{\overline{n}\zzero = \overline{n}\zzero} \\
&= \Os.
\end{align*}
For any $\eta^*$, $\eta^* \in \Os$,
that is $\eta^* \in \model{\overline{\str_2} = \overline{\str_1}\zzero}$.
Similarly, if $b=1$ and $\str_2=\str_1\one$
(and $\bbool = \oone$).
Then,
\begin{align*}
\model{\overline{\str_2} = \overline{\str_1}\oone}
&=
\model{\overline{n\one} = \overline{n}\oone} \\
&= \model{\overline{n}\oone = \overline{n}\oone} \\
&= \Os.
\end{align*}

\item[$\Leftarrow$] Assume (by contraposition)
$b=0$ and $\str_2\neq \str_1 \zero$
(and $\bbool = \oone$).
Then, clearly $\model{\overline{\str_2} = \overline{\str_1}\oone}
= \emptyset$.
So, for any $\eta^*\in \Os$, 
$\eta^* \not\in \emptyset$.
The case of $b=0$ and $m\neq n\oone$ (and $\bbool=\zzero$)
is equivalent.

\end{itemize}

\item $f=C$.
The function is $\Sigma^b_1$-representable in $\RS$
by the formula
$
F_C(x,v,z_0,z_1,y) : (x=\epsilon \wedge y=v)
\vee (\exists x' \preceq x)
(x=x'\zzero \wedge y=z_0)
\vee (\exists x'\preceq x)(x=x'\one \wedge y=z_1).
$

\item $f=Q$. The function is $\Sigma^b_1$-represented
in $\RS$ by $
F_Q : (\Flip(x) \wedge y=\oone) \vee
(\neg \Flip(x) \wedge y=\zzero).
$
The existence condition is proved as follows:

\end{itemize}

\footnotesize
\begin{prooftree}
\AxiomC{}
\RightLabel{$Ax$}
\UnaryInfC{$\Flip(x) \vdash
(\exists y)F_Q(z,y), \Flip(z)$}
\RightLabel{$\neg R$}
\UnaryInfC{$\vdash (\exists y)F_Q(z,y),
\Flip(z), \neg \Flip(z)$}

\AxiomC{$\mathcal{D}_{\exists Q}$}

\noLine
\UnaryInfC{$\vdash (\exists y)F_Q(z,y), \Flip(z),\oone=\zzero$}
\RightLabel{$\wedge R$}
\BinaryInfC{$\vdash (\exists y)F_Q(z,y),
\Flip(z), \neg\Flip(z) \wedge \oone=\zzero$}

\AxiomC{}
\RightLabel{$Ax$}
\UnaryInfC{$\oone=\oone \vdash ..., \oone=\oone$}
\RightLabel{$Ref$}
\UnaryInfC{$\vdash ..., \oone = \oone$}
\RightLabel{$\wedge R$}
\BinaryInfC{$\vdash (\exists y)F_Q(z,y), (\Flip(z) \wedge \oone=\oone), 
(\neg \Flip(z) \wedge \oone=\zzero)$}
\RightLabel{$\vee R$}
\UnaryInfC{$\vdash (\exists y)F_Q(z,y), (\Flip(z) \wedge \oone=\oone)
\vee (\neg \Flip(z) \wedge \oone=\zzero)$}
\RightLabel{$\exists R$}
\UnaryInfC{$\vdash (\exists y)((\Flip(z) \wedge y=\oone)
\vee (\neg \Flip(z) \wedge y=\zzero))
$}
\RightLabel{$\forall R$}
\UnaryInfC{$\vdash(\forall x)(\exists y)((\Flip(x) \wedge
y=\oone) \vee (\neg \Flip(x) \wedge
y=\zzero))$}
\end{prooftree}
\begin{itemize}
\item[]
where

$$
\mathcal{D}_{\exists Q}
$$

\begin{prooftree}
\AxiomC{}
\RightLabel{$Ax$}
\UnaryInfC{$\Flip(z)\vdash .... , \Flip(z)$}
\RightLabel{$\neg R$}
\UnaryInfC{$\vdash ...\Flip(z), \neg \Flip(z)$}

\AxiomC{}
\RightLabel{$Ax$}
\UnaryInfC{$\zzero=\zzero \vdash ..., \Flip(z),
\zzero=\zzero$}
\RightLabel{$Ref$}
\UnaryInfC{$\vdash ..., \Flip(z), \zzero=\zzero$}

\RightLabel{$\wedge R$}
\BinaryInfC{$\vdash ..., \Flip(z), \Flip(z)\wedge \zzero=\oone,
\neg \Flip(z) \wedge \zzero=\zzero$}
\RightLabel{$\vee R$}
\UnaryInfC{$\vdash(\exists y)F_Q(z,y), \oone = \zzero,
\Flip(z), F_Q(z,\zzero)$}
\RightLabel{$\exists R$}
\UnaryInfC{$\vdash (\exists y)F_Q(z,y),
\Flip(z),\oone=\zzero$}

\end{prooftree}

Uniqueness is shown as below:

\begin{prooftree}
\AxiomC{$\mathcal{D}_{\forall Q1}$}
\noLine
\UnaryInfC{$\Flip(x_1) \wedge x_2 = \oone,
F_Q(x_1,x_2) \vdash x_2 = x_3$}

\AxiomC{$\mathcal{D}_{\forall Q2}$}
\noLine
\UnaryInfC{$\neg \Flip(x_1) \wedge = \zzero,
F_Q(x_1,x_3) \vdash x_2=x_3$}

\RightLabel{$\vee L$}
\BinaryInfC{$F_Q(x_1,x_2), F_Q(x_1,x_3)
\vdash x_2=x_3$}
\RightLabel{$\wedge L$}
\UnaryInfC{$F_Q(x_1,x_2)\wedge F_Q(x_1,x_3) \vdash
x_2=x_3$}
\RightLabel{$\rightarrow R$}
\UnaryInfC{$\vdash F_Q(x_1,x_2)
\wedge F_Q(x_1,x_3) \rightarrow x_2 = x_3$}
\RightLabel{$\forall R$s}
\UnaryInfC{$\vdash (\forall x)(\forall y)(\forall z)(F_Q(x,y)
\wedge F_Q(x,z) \rightarrow y=z)$}
\end{prooftree}
\normalsize
where
\end{itemize}

$$
\mathcal{D}_{\forall Q1}
$$

\tiny
\begin{prooftree}
\AxiomC{}
\RightLabel{$Ax$}
\UnaryInfC{$x_2=x_3, \Flip(x_1), x_2=\oone,
x_3=\oone \vdash x_2=x_3$}
\RightLabel{$Repl$}
\UnaryInfC{$\Flip(x_1), x_2=\oone, \Flip(x_1),
x_3=\oone \vdash x_2=x_3$}
\RightLabel{$\wedge L$}
\UnaryInfC{$\Flip(x_1),x_2=\oone,
\Flip(x_1) \wedge x_3=\oone \vdash
x_2 = x_3$}

\AxiomC{}
\RightLabel{$Ax$}
\UnaryInfC{$\Flip(x_1), x_2=\oone, x_3=\zzero
\vdash x_2=x_3,\Flip(x_1)$}
\RightLabel{$\neg L$}
\UnaryInfC{$\Flip(x_1), x_2=\oone,
\neg \Flip(x_1), x_3=\zzero \vdash x_2=x_3$}
\RightLabel{$\wedge L$}
\UnaryInfC{$\Flip(x_1), x_2=\oone,
\neg \Flip(x_1) \wedge x_3=\zzero \vdash
x_2=x_3$}

\RightLabel{$\vee L$}
\BinaryInfC{$\Flip(x_1),x_2=\oone, F_Q(x_1,x-3)
\vdash x_2=x_3$}
\RightLabel{$\wedge L$}
\UnaryInfC{$\Flip(x_1) \wedge x_2=\oone,
F_{Q}(x_1,x_3) \vdash x_2=x_3$}
\end{prooftree}

\normalsize

$$
\mathcal{D}_{\forall Q2}
$$

\tiny
\begin{prooftree}
\AxiomC{}
\RightLabel{$Ax$}
\UnaryInfC{$\Flip(x_1), x_3=\oone, y=\zzero \vdash
x_2=x_3, \Flip(x_1)$}
\RightLabel{$\neg L$}
\UnaryInfC{$\neg \Flip(x_1), y=\zzero, \Flip(x_1),
x_3=\oone \vdash x_2=x_3$}
\RightLabel{$\wedge L$}
\UnaryInfC{$\neg \Flip(x_1), y=\zzero, 
\Flip(x_1) \wedge x_3=\oone \vdash x_2=x_3$}

\AxiomC{}
\RightLabel{$Ax$}
\UnaryInfC{$x_2=x_3, \neg \Flip(x_1), x_2=\zzero,
x_3=\zzero \vdash x_2=x_3$}
\RightLabel{$Repl$}
\UnaryInfC{$\neg \Flip(x_1), x_2=\zzero, \neg \Flip(x_1),
x_3=\zzero \vdash x_2=x_3$}
\RightLabel{$\wedge L$}
\UnaryInfC{$\neg \Flip(x_1), x_2=\zzero,
\neg \Flip(x_1) \wedge x_3=\zzero \vdash x_2=x_3$}

\RightLabel{$\vee L$}
\BinaryInfC{$\neg \Flip(x_1), y=\zzero, F_Q(x_1,x_3)
\vdash x_2=x_3$}
\RightLabel{$\wedge L$}
\UnaryInfC{$\neg \Flip(x_1) \wedge y=\zzero, F_Q(x_1,x_3)
\vdash x_2=x_3$}
\end{prooftree}

\begin{itemize}
\normalsize
\item[]

The semantic condition is proved considering that
for all
$\str_1,\str_2$ and $\eta^* \in \Os$,
$Q(\str_1,\eta^*) = \str_2$
when $\eta^* \in \model{(\Flip(\overline{\str_1})
\wedge \overline{\str_2}=\oone)
\vee (\neg \Flip(\overline{\str_1} \wedge
\overline{\str_2}=\zzero)}$.

\begin{itemize}
\itemsep0em
\item[$\Rightarrow$]
Assume $Q(\str_1,\eta^*) =\str_2$
and $\str_2=\one$,
that is $\eta^* (\str_1)=\one$.
So,

\begin{align*}
\model{F_Q(\overline{\str_1},\overline{\str_2}) }
&= 
\model{(\Flip(\overline{\str_1}) \wedge \overline{\one}
=\oone) \vee (\neg \Flip(\overline{\str_1}) \wedge
\overline{\one} = \zero)} \\
&= \model{(\Flip(\overline{\str_1}) \wedge
\oone = \oone) \vee (\neg \Flip(\overline{\str_1}) \wedge
\oone =\zzero)} \\
&= \model{\Flip(\overline{\str_1}) \wedge \oone = \oone}
\cup \model{\neg \Flip(\overline{\str_1}) \wedge \oone = \zzero} \\
&= \big(\model{\Flip(\overline{\str_1})} \cap \model{\oone =\oone}\big)
\cup \big(\model{\neg \Flip(\overline{\str_1})} \cap
\model{\oone=\zzero}\big) \\
&= \big(\model{\Flip(\overline{\str_1})} \cap \Os\big) 
\cup \big(\model{\neg \Flip(\overline{\str_1})} \cap 
\emptyset \big) \\
&= \model{\Flip(\overline{\str_1})} \cup \emptyset \\
&= \model{\Flip(\overline{\str_1})} \\
&= \{ \eta \ | \ \eta (\str_1) = \one\}.
\end{align*}
As seen $\eta^* (\str_1)=\one$,
so $\eta^* \in \{\eta \ | \ \eta(n) = \one\}
= \model{(\Flip(\overline{\str_1}) \wedge
\overline{\one} =\oone) \vee 
(\neg \Flip(\overline{\str_1}) \wedge \overline{\one}=\zzero)}$.

Assume $Q(\str_1,\eta^*)=\str_2$ and $\str_2=\zero$,
that is $\eta^* (\str_1)=\zero$.
So,

\begin{align*}
\model{F_Q(\overline{\str_1},\overline{\str_2})} &=
\model{(\Flip(\overline{\str_1}) \wedge \overline{\one}
= \oone) \vee (\neg \Flip(\overline{\str_1}) \wedge
\overline{\one} = \zzero)} \\
&= \model{(\Flip(\overline{\str_1}) \wedge \zzero=\oone)
\vee (\neg \Flip(\overline{\str_1}) \wedge \zzero=\zzero)} \\
&= \model{\Flip(\overline{\str_1}) \wedge \zzero=\oone}
\cup \model{\neg \Flip(\overline{\str_1}) \wedge \zzero=\zzero} \\
&=
\big(\model{\Flip(\overline{\str_1})} \cap \model{\zzero=\oone}\big)
\cup
\big(\model{\neg \Flip(\overline{\str_1})} \cap \model{\zzero=\zzero} \big) \\
&= (\model{\Flip(\overline{\str_1})} \cap \emptyset)
\cup (\model{\neg \Flip(\overline{\str_1})} \cap \Os) \\
&= \emptyset \cup \model{\neg \Flip(\overline{\str_1})} \\
&= \model{\neg \Flip(\overline{\str_1})} \\
&= \Os - \model{\Flip(\overline{\str_1})} \\\
&= \{ \eta \ | \ \eta (n) = \zero\}.
\end{align*}
As seen, $\eta^*(n)=\\zero$, so clearly 
$\eta^* \in \{\eta \ | \ \eta(n) = \zero\} =
\model{(\Flip(\overline{\str_1}) \wedge \overline{\zero}
= \oone ) \vee (\neg \Flip(\overline{\str_1})
\wedge \overline{\zero} = \zzero)}$.

\item[$\Leftarrow$]
Assume (by contraposition) $\eta^*(\str_1)=\one$,
so $Q(\str_1,\eta^*) =\one$
and $m=\zero$.
As shown above,
in this case $\model{(\Flip(\overline{\str_1}) \wedge
\overline{\zero}=\oone) \vee (\neg \Flip(\overline{\str_1})
\wedge \overline{\zero}=\zzero)}
= \{\eta \ | \ \eta(\str_1)=\zero\}$.
For assumption, $\eta^*$ is such that $\eta^*(\str_1)=\one$,
so $\eta^* \not \in \{\eta \ | \ \eta(\str_1)=\zero\}$,
that is $\eta^* \not\in \model{(\Flip(\overline{\str_1} \wedge
\overline{\zero}=\oone) \vee
(\neg \Flip(\overline{\str_1}) \wedge
\overline{\zero}=\zzero)}$.

Assume (by contraposition) $\eta^*(\str_1)=\zero$,
so $Q(\str_1,\eta^*)=\zero$ and $m=\one$.
As shown above, in the case $\model{(\Flip(\overline{\str_1})
\wedge \overline{\one}=\oone) \vee (\neg \Flip(\overline{\str_1})
\wedge \overline{\one} = \zzero)}
= \{\eta \ | \ \eta(\str_1)= \one\}$.
For assumption, $\eta^*$ is such that $\eta^*(\str_1)=\zero$,
so $\eta^* \not\in \{\eta \ | \ \eta(\str_1) =\one\}$,
that is $\eta^* \not \in \model{(\Flip(\overline{\str_1}) \wedge
\overline{\one}=\oone) \vee
(\neg \Flip(\overline{\str_1}) \wedge \overline{\one}=\zzero)}$.

\end{itemize}
\end{itemize}
\normalsize

\emph{Inductive Case.} If $f\in\POR$ is obtained from
$\Sigma^b_1$-representable functions by 
either composition or bounded recursion, then $f$
is $\Sigma^b_1$-representable as well.

\end{proof}

\section{Proofs from Section~\ref{sec:PORandPTM}}\label{appendix:TaskC}

\subsection{Relating $\RFP$ and $\SFP$}

\subsubsection{On Stream Turing Machine}

We define stream Turing machines in a formal
way.

\begin{defn}[Stream Turing Machine]
A \emph{stream Turing machine} 
is a quadruple $\STM= \langle \Qstates, q_0,\Sigma,
\delta\rangle$,
where:
\begin{itemize}
\itemsep0em
\item $\Qstates$ is a finite set of states
ranged over by $q_1,q_2,\dots$ 
\item $q_0\in\Qstates$ is an initial state
\item $\Sigma$ is a finite set of characters ranged over
by $c_1,c_2,\dots$
\item $\delta: \hat{\Sigma} \times \Qstates \times
\hat{\Sigma} \times \Bool \to \hat{\Sigma}
\times \{\mathtt{L},\mathtt{R}\}$ is a transition function describing
the new configuration reached by the machine.
\end{itemize}
$\mathtt{L}$ and $\mathtt{R}$ are two fixed and distinct symbols
(e.g. $\one$ and $\zero$),
$\hat{\Sigma}=\Sigma \cup\{\blank\}$
and $\blank\not\in\Sigma$ is the \emph{blank
character}.
\end{defn}

\begin{defn}[Configuration of STM]
The \emph{configuration of an STM}
is a quadruple $\langle \str, q, \strT,\omega\rangle$,
where:
\begin{itemize}
\itemsep0em
\item $\str\in\{\zero,\one,\blank\}^*$ is the portion
of the work tape
on the left of the head
\item $q\in\Qstates$ is the current state of the machine
\item $\strT \in\{\zero,\one,\blank\}^*$ is the portion
of the work tape on the right of the head
\item $\omega\in \Bool^\Nat$ is the portion of the oracle
tape that has not been read yet.
\end{itemize}
\end{defn}
\noindent
We can now  give the definition of family of reachability
relations for an STM.

\begin{defn}[Reachability function for STM]
Given an STM $\STM$ with transition function $\delta$,
we use $\vdash_\delta$ to denote its standard step function
and $\{\triangleright^n_{\STM}\}_n$
the smallest family of relations such that:

\begin{align*}
\langle \str,q,\strT,\omega \rangle &\triangleright^0_{\STM}
\langle \str,q,\strT,\omega \rangle \\
\Big(\langle \str,q,\strT,\omega \rangle
\triangleright^n_{\STM} \langle \str',q,\strT',\omega'\rangle\Big)
\ \wedge  \ \ & \\
 \Big(\langle \str',q',\strT',\omega'\rangle
\vdash_\delta 
\langle \str'',q',\strT'', \omega''\rangle\Big)
\rightarrow \ &  \Big(
\langle \str, q,\strT,\omega\rangle 
\triangleright^{n+1}_{\STM} 
\langle \str'',q',\strT'', \omega''\rangle \Big)
\end{align*}
\end{defn}

\begin{defn}[STM Computation]
Given an STM $\STM=\langle \Qstates,
q_0,\Sigma,\delta\rangle$, oracle tape
$\omega : \Nat \to \Bool$ and a function
$g:\Nat \to \Bool$,
we say that \emph{$\STM$ computes g},
$f_{\STM}=g$, when for every string $\str\in\Ss$,
and oracle tape $\omega \in\Bool^\Nat$,
there are $n\in\Nat, \strT\in\Ss,q'\in\Qstates$,
and a function $\psi: \Nat \to \Bool$ such that 
$$
\langle \epsilon, q_0, \str, \omega\rangle
\triangleright^n_{\STM}
(\strTT, q', \strT, \psi\rangle,
$$
and $\langle \strTT, q', \strT, \psi\rangle$
is a final configuration for $\STM$
with $f_{\STM}(\str,\omega)$ being the longest
suffix of $\strTT$ not including $\blank$.
\end{defn}
%
%
%

\subsubsection{Proof of Proposition~\ref{prop:PTMandSTM}}
We need to extend this definition to
probability distributions over $\Ss$,
so to deal with PTMs.

\begin{defn}
Given a PTM $\PTM$,
a configuration $\langle \str, q,\strT\rangle$,
we define the following sequence of random 
variables:
\begin{align*}
X^{\langle \str,q,\strT\rangle}_{\PTM,0} &\df
\omega \to \langle \str,q,\strT\rangle \\
X^{\langle \str,q,\strT\rangle}_{\PTM,n+1} 
&\df \omega \to 
\begin{cases}
\delta_{\bool}\big(X^{\langle \str,q,\strT\rangle}(\omega)\big)
&\text{if } \omega(n) = \bool \text{ and for some }
\langle \str',q',\strT'\rangle, \\
& \ \ \ \ \ \ \ \ 
\ \ \ \ \ \ 
\ \ \ \ \ \delta_\bool(X^{\langle \str,q,\strT\rangle}_{\PTM,n}
(\omega))= \langle \str',q',\strT'\rangle \\
X^{\langle \str,q,\strT\rangle}_{\PTM,n}(\omega)
&\text{if } \omega(n) = \bool \text{ and for no }
\langle \str',q',\strT'\rangle , \\
& \ \ \ \ \ \ \ \ 
\ \ \ \ \ \ 
\ \ \ \ \ \delta_\bool(X^{\langle \str,q,\strT\rangle}_{\PTM,n}
(\omega)) = \langle \str',q',\str'\rangle.
\end{cases}
\end{align*}
\normalsize
for any $\omega \in \Bool^\Nat$.
\end{defn}
\noindent
Intuitively, the variable $X^{\langle \str,q,\strT\rangle}_{\PTM,n}$
describes the configuration reached by the machine
after exactly $n$ transitions.
We say that a PTM $\PTM$ computes $Y_{\PTM,\str}$
when there is an $m\in \Nat$ such that
for any $\str\in\Ss$, $X^{\langle \str,q_0,\strT\rangle}_{\PTM,m}$
is final.
In this case, $Y_{\PTM,\str}$ is the longest suffix of 
$\pi_1\Big(X^{\langle \str, q_0,\epsilon\rangle}_{\PTM,m}\Big)$,
which does not contain $\blank$.

We now prove Proposition~\ref{prop:PTMandSTM},
which establishes the equivalence between STMs and PTMs.
\longv{
\begin{prop}\label{prop:PTMandSTM}[Equivalence between PTMs
and STMs]
For any poly-time STM $\STM$, there is a poly-time PTM
$\PTM^*$ such that for any strings $\str,\strT\in \Ss$,
$$
\mu(\{\eta \ | \ \STM(\str,\omega)=\strT\}) = \Pr[\STM^*(\str)=
\strT],
$$
and viceversa.
\end{prop}
}

\begin{proof}[Proof Sketch of Proposition~\ref{prop:PTMandSTM}]
We show that for any $\str, \strT \in \Ss$,
\begin{align*}
\mu
\big(\{\omega \in \Bool^\Nat \ | \ \STM(\str,\omega)=\strT\}\big) 
&=
\mu 
\big(\PTM(\str)^{-1} (\strT)\big) \\
\mu
\big(\{\omega \in \Bool^\Nat \ | \ \STM(\str,\omega)=\strT\}\big) &=
\mu
\big(\{\omega \in \Bool^\Nat \ | \ Y_{\STM,\str}(\omega)=\strT\}\big).
\end{align*}
We actually show a stronger result,
namely that there is a bijection $I$ : STMs $\to$
PTM such that for any $n\in\Nat$
\begin{align}
\{\omega \in \Bool^\Nat \ | \ \langle \str, q_0,\strT,\omega\rangle
\triangleright^n_\delta \langle
\strT, q,\psi,n\rangle\} 
=
\{\omega \in \Bool^\Nat \ | \ X^{\langle \epsilon,q_0,\str\rangle}_{I(\STM),n}
(\omega) = \langle \strT,q,\psi\rangle\}.
\tag{I.}
\end{align}
This entails,
\begin{align}
\{\omega \in \Bool^\Nat \ | \ \STM(\str,\omega)=\strT\} 
= 
\{ \omega \in \Bool^\Nat
\ | \ Y_{I(\STM),\str}(\omega)=\strT\}.
\tag{II.}
\end{align}
The bijection $I$ splits the function $\delta$ of $\STM$
so that if the corresponding character on the oracle tape is $\zero$,
then the transition is defined by $\delta_0$;
if the character is $\one$, it is defined by $\delta_1$.
%
We prove (I.) by induction on the number of steps required by 
$\STM$ to compute its input value.
\end{proof}

\subsection{From $\SFP$ to $\POR$}
To prove Proposition~\ref{prop:SFPtoPOR}
we show the correspondence between functions
computable by an STM and those computable by
a \emph{finite-stream} STM.

\begin{lemma}\label{lemma:STACS22}
For all $f\in \SFP$ with polynomial time-bound $p$
there is an $h\in\PTF$ such that,
for any $\omega \in \Bool^\Nat$ and $\str\in\Ss$,
$$
f(\str,\eta)= h(\str,\omega_{p(|\str|)}).
$$
\end{lemma}

\begin{proof}
Since $f\in\SFP$,
there is poly-time STM $\STM=\langle \Qstates,q_0,\Sigma,\delta\rangle$
such that $f=f_{\STM}$.
Let us consider the FSTM $\STM'$,
which is defined as $\STM$ but is finite.
Then for any $k\in\Nat$
and $\str', \strT',\strT'',\strT''' \in \Ss$,
$$
\langle \epsilon,q_0',\str,\str'\rangle
\triangleright^k_{\STM} \langle \strT,q,\strT',\strT''\rangle \ \ 
\text{ iff } \ \ 
 \langle\epsilon,q_0',\str,\str'\eta\rangle
\triangleright^k_{\STM'} \langle \strT,q,\strT'',\strT'''\eta\rangle.
$$
Furthermore, $\STM'$ requires the same number
of steps as those required by $\STM$, so
 is in $\PTF$ too.
We conclude the proof defining $h=f_{\PTM'}$.
\end{proof}
\noindent
The next step consists in showing that each function
$f\in\PTF$ corresponds to a function
which can be defined without recurring to $\query$.

\begin{lemma}\label{lemma:STACS23}
For any $f\in\PTF$ and $\str\in\Ss$,
there is a $g\in\POR$ such that for any
$x,y\in \Ss$ and $\eta\in\Os$,
$f(x,y)=g(x,y,\eta)$.
Furthermore, if $f$ is defined without recurring to
$\query$,
$g$ do not include $\query$ as well.
\end{lemma}

\begin{proof}[Proof Sketch]
Let us outline the main steps of the proof:
\begin{enumerate}
\itemsep0em
\item We define encodings using strings for configuration and
transition functions – called $e_c$ and $e_t$, respectively –
for FSTM.
Moreover, there is a function $step\in\POR$,
which satisfies the simulation schema in Figure~\ref{fig:step}.
Observe that $e_c,e_t$ and $step$ are correct
with respect to the given simulation.

\begin{center}
 \begin{figure}[h!]
 \begin{center}
\framebox{
\parbox[t][3.2cm]{9.5cm}{
\footnotesize{

  \centering
  \begin{tikzpicture}[node distance = 8 cm]
    \node at (-3.2,2) (c) {$c = \langle \sigma, q, \tau, y\rangle$};
    \node at (-3.2,0) (sc) {$\str_c\in \Ss$};
    \node at (3.2,2) (d) {$d = \langle \sigma, q, \tau, y\rangle$};
    \node at (3.2,0) (sd) {$\str_d\in \Ss$};
    \node at (0,2.3) {\textcolor{gray}{$\vdash_{\delta}$}};
    \node at (-3.6,1) {\textcolor{gray}{$e_c$}};
     \node at (3.6,1) {\textcolor{gray}{$e_c$}};
     \node at (0,-0.3) {\footnotesize{\textcolor{gray}{for any $\eta\in\Os$, $step(\str_c, e_t(\delta), \eta)=\str_d$}}};
     
    \draw[->] (c) edge  (d);
    \draw[->] (sc) edge  (sd);
    \draw[->] (c) edge (sc);
    \draw[->] (d) edge (sd);
  \end{tikzpicture}

  }}}
\caption{Behavior of $step$.}\label{fig:step}
\end{center}
\end{figure}
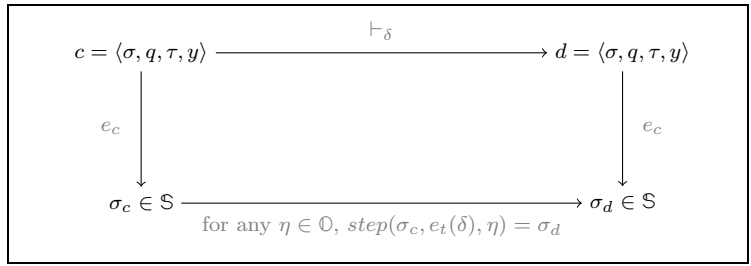  
\end{center}

\item For any $f\in\POR$ and $\str,\strT\in\Ss$,
if there is an $\rl$-term $t(x)$, which bounds the size
of $f(\str,\eta)$ for any possible input,
then 
$
f_m=\lambda z,x,\eta^{|z|}(x,\eta)\in\POR.\footnote{Observe
that if
$f$ is defined without recurring to $\query$,
also $f_m$ can be defined without using it.}
$

\item Given a machine $\STM$, if $\str\in\Ss$ is a correct
encoding for a configuration of $\STM$, 
then for any $\eta\in\Os$ $|step(\str,\eta)|$
is $\mathcal{O}(|\str|)$.

\item For any $\eta\in \Os$,
if $c=e_c(\str,q,\strT,y,\eta)$,
then
there is a function $dectape$ such that
 $dectape(x,\eta)$
is the longest suffix without occurrences of 
$\blank$.
\end{enumerate}
\end{proof}

Then, as a consequence of Lemma~\ref{lemma:STACS23},
each function $f\in\SFP$ can be simulated
by a function $g\in\POR$ using 
a polynomial prefix of the oracle for $f$
as an additional input.

\begin{cor}\label{cor:STACS8}
For any $f\in\SFP$ and polynomial $p$,
there is a function $f\in\POR$, such that
for any $\omega\in \Bool^\Nat,$
$\eta\in\Os$ and $x\in\Ss$,
$$
f(x,\omega)=g(x,\omega_{p(|x|),\eta}).
$$
\end{cor}

\begin{proof}
Assume $f\in\SFP$ and $y=\eta_{p(|x|)}$.
By Lemma~\ref{lemma:STACS22},
there is a function $h\in\PTF$ such that
for any $\eta:\Nat\to \Bool$ and $x\in\Ss$,
$$
f(x,\eta)=h(x,\eta_{p(|x|)}).
$$
Moreover, due to Lemma~\ref{lemma:STACS23},
there is a $g\in\POR$ such that for any $x,y\in\Ss,\omega
\in \Os$,
$$
g(x,y,\omega)=h(x,y).
$$
Them $g$ is the desired function.
\end{proof}

We now establish that there is a function $e\in\POR$,
which produces strings with the same distribution
of the prefixes of the functions in $\Ss^\Nat$.
Intuitively, this function extracts $|x|+1$ bits
from $\eta$ and concatenates them in the output.
The definition of $e$ passes through a dyadic representation of a number,
$dyad : \Nat \to \Ss$.
Therefore, the function $e$ creates
strings $\one^0,\one^1,\dots, \one^k$,
and samples the function $\eta$ on the given
coordinates –
namely
$dy(\one^0),dy(\one^1),\dots, dy(\one^k)$ –
concatenating the result in a string.

\begin{defn}[Function $dyad$]
The function $dyad : \Nat \to \Ss$ associates
each $n\in \Nat$ to the string obtained
by stripping the left-most bit from the binary representation
of $n+1$.
\end{defn}
\noindent
In order to define $dy(\cdot,\cdot)$,
some auxiliary functions are introduced,
namely $binsucc:\Ss\times \Os\to \Ss$,
\begin{align*}
binsucc(\epsilon,\eta) &\df \one \\
binsucc(x\zero,\eta) &\df x\one|_{x\zero\zero} \\
binsucc(x\one,\eta) &\df 
binsucc(x,\eta)\zero|_{x\zero\zero}.
\end{align*}
and $bin:\Ss\times \Os \to \Ss$,
\begin{align*}
bin(\epsilon,\eta) &\df \zero \\
bin(x\bool,\eta) &\df binsucc(bin(x,\eta),\eta)|_{x\bool}.
\end{align*}

\begin{defn}[Function $dy$]
The function $dy:\Ss\times \Os \to\Ss$ is defined as follows:
$$
dy(x,\eta) \df  lrs(bin(x,\eta),\eta),
$$
where $lrs$ is a string manipulator, which removes
the left-most bit from a string if it exists; otherwise returns 
$\epsilon$.\footnote{Full
details can be found in~\cite{Davoli}.}
\end{defn}
\noindent
The function $dyad(n)$ is easily shown bijective
and the following proposition is proved by induction.

\longv{
\begin{lemma}
The function $dyad(n)$ is bijective.
\end{lemma}

\begin{proof}
\begin{enumerate}
\itemsep0em
\item
The function is an injection. different numbers
have different binary encodings and, for any
$n\in\Nat^+$ and $\omega\in\Os$,
$bin(n,\omega)$ has $\one$ as its left-most bit.\footnote{This is 
proved by induction on $n$, leveraging the definition of $bin$.}
Given two distinct binary encodings of $n\neq m\in\Nat$,
say $\one\str$ and $\one\strT$, 
then $\str\neq \strT$.\footnote{Otherwise $n=m$.
Indeed, the function associating numbers to binary representations
is itself a binjection.}
\item
The function is surjective. 
It is computed by removing a bit, which is
always $\one$.
So, each string $\str\in\Ss$ is the image of
a number $n\in\Nat$ such that the binary encoding
of $n+1$ is $\one\str$.
This number always exists.
\end{enumerate}
\end{proof}
}

\begin{prop}\label{prop:STACS26}
For any $n\in\Nat,\str\in\Ss$ and $\eta\in\Os$,
if
$|\str|=n+1$, then 
$$
dy(\str,\eta)=dyad(n).
$$
\end{prop}

We also introduce the function $e$ and prove its correctness
(by induction on the structure of strings).

\begin{defn}[Function $e$]
Let $e:\Ss \times \Os \to \Ss$ be defined
as follows:
\begin{align*}
e(\epsilon,\eta) &\df \epsilon \\
e(x\bool,\eta) &\df e(x,\omega)\query(dy(x,\eta),
\eta)|_{x\bool}.
\end{align*}
\end{defn}

\begin{lemma}[Correctness of $e$]\label{lemma:eCor}
For any $\str\in \Ss$ and $i\in\Nat$,
if $|\str|=i+1$,
for each $j\leq i\in \Nat$ and $\eta\in\Os$:
\begin{itemize}
\itemsep0em
\item[i.] 
$e(\str,\eta)(i)=\eta(dy(\one^j,\eta))$,
\item[ii.] 
the length of $e(\str,\eta)$ is exactly
$i+1$.
\end{itemize}
\end{lemma}
We also introduce a relation $\sim_{dy}$,
which is again a bijection.
\begin{defn}
We define $\sim_{dy}$ as the smallest relation
in $\Os \times \Bool^\Nat$ such that:
$$
\eta \sim_{dy} \eta \ \ \text{iff} \ \ 
\text{ for any } n\in\Nat, \omega(n)=\eta(dy(\one^{n+1},
\omega)).
$$
\end{defn}

\longv{
\begin{lemma}
The relation $\sim_{dy}$ is a bijection.
\end{lemma}
\begin{proof}
By Proposition~\ref{prop:STACS26}
and Lemma~\ref{lemma:eCorr},
there is a $\omega\in\Os$ 
\textcolor{red}{in relation with $\eta$}.
Assume there is an $\omega'\neq \omega\in\Os$,
again in relation with $\eta$.
Then there is a $\str\in\Ss$, such that
$\omega(\str)\neq \omega'(\str)$.
By Proposition~\ref{prop:STACS26},
the value of $\omega$ does not affect the output.
Furthermore $dy$ is a bijection,
so there is an $n\in\Nat$ such that $dy(\one^{n+1},
\omega)=\str$ and
$\eta(n)=\omega_1(\str) \neq \omega_2(\str)=
\eta(n)$ is a contradiction.
The proof is the same when dealing with $\eta \in \Bool^\Nat$.
\end{proof}
}

\begin{lemma}\label{lemma:STACS29}
If $\omega \sim_{dy} \eta$
then, there is an $n\in\Nat$,
$$
e(\underline{n}_\Nat,\eta).
$$
\end{lemma}

\begin{proof}
The proof is by contraposition.
Suppose $\omega_n\neq e(\underline{n}_\Nat,\eta)$.
As a consequence of Lemma~\ref{lemma:eCor},
there is an $m\in\Nat$ such that 
$\omega(m)\neq \eta(dy(\underline{m}_\Nat,\eta))$,
which is a contradiction.
\end{proof}

We can finally prove Proposition~\ref{prop:SFPtoPOR}.

\begin{proof}[Proof of Proposition~\ref{prop:SFPtoPOR}]
By Corollary~\ref{cor:STACS8},
there is a function $f'\in\POR$ and a polynomial
$p$ 
such that for any
$\str,\strT\in\Ss$, $\omega \in \Bool^\Nat$
and $\eta \in \Os$,
\begin{align}
\str= \omega_{p(x)} \rightarrow 
f(\str,\omega)=f'(\str,\strT,\eta).
\tag{$*$}
\end{align}
Let us fix $\overline{\omega}\in\{\omega\in\Bool^\Nat \ | \ 
f(\str,\omega)=y\}$,
its image with respect to $sim_{dy}$ is in
$
\{\eta \in \Os \ | \ f'(x,e(p'(s(x,\eta),\eta),
\eta),\eta)= \strT\},
$
where $s$ is the function of $\POR$ computing
$\one^{|x|+1}$.
By Lemma~\ref{lemma:STACS29},
$$
\overline{\omega}_{p(x)} = e(p(size(x,\eta),\eta),
$$
where $p'\in \POR$
computes the polynomial $p$,
defined without recurring to $\query$.
%
%
Furthermore, given a fixed 
$\overline{\eta} \in\{\eta \in \Os \ | \
f'(\str,e(p'(size(\str,\eta),\eta), \eta),\eta)=
\strT\}$,
its pre-image with respect to $\sim_{dy}
\in \{\omega\in\Bool^\Nat \ | \ f(\str,\omega)=\strT\}$.
The proof is analogous to the one above.
Since $\sim_{dy}$ is a bijection between
$$
\mu(\{\omega\in\Bool^\Nat \ | \ f(\str,\omega)=\strT\})
= 
\mu(\{\eta \in \Os | f'(x,e(p(size(x,\eta),
\eta),\eta),\eta) = \strT\}),
$$
which concludes the proof.
\end{proof}

\subsection{From $\POR$ to $\SFP$}

First, we define the imperative language
$\SIFPra$ together with its
big-step semantics, 
and prove its poly-time programs equivalent
to $\POR$.
%

\begin{defn}[Programs of $\SIFPra$]
The language of programs of $\SIFPra$
$\LSra$ – i.e.~the set of strings
produced by non-terminal symbol
$\prog{Stm}_{\text{RA}}$ – is defined
as follows:
\small
\begin{align*}
\prog{Id} &\df X_i \midd Y_i \midd S_i \midd
R \midd Q \midd Z \midd T \\
\prog{Exp} &\df \epsilon \midd
\prog{Exp.0} \midd \prog{Exp.1} \midd
\prog{Id} \midd \prog{Exp} \sqsubseteq \prog{Id} \midd
\prog{Exp} \wedge \prog{Id} \midd
\neg \prog{Exp} \\
\prog{Stm}_{\text{RA}} &\df
\prog{Id} \leftarrow \prog{Exp} \midd
\prog{Stm}_{\text{RA}}; \prog{Stm}_{\text{RA}} \midd
\mathtt{while} (\prog{Exp}) \{\prog{Stm}\}_{\text{RA}} \midd
\mathtt{Flip}(\prog{Exp}),
\end{align*}
\normalsize
with $i\in\Nat$.
\end{defn}
\noindent
The big-step semantics associated with the language of $\SIFPra$
programs relies on the notion of \emph{store}.

\begin{defn}[Store]
A \emph{store} is a function 
$\Sigma:\prog{Id} \rightharpoonup  \Ss$.
An \emph{empty} store is a store which is total
and constant on $\progE$.
We represent such object as $[ \ ]$.
We define the updating of a store $\Sigma$
with a mapping from $y\in\prog{Id}$ to $\str\in \Ss$
as:
$$
\Sigma[y\leftarrow \str](x) \df \begin{cases}
\str \ &\text{if } x=y \\
\Sigma(x) \ &\text{otherwise.}
\end{cases}
$$
\end{defn}

\begin{defn}[Semantics of Expressions in $\SIFPra$]
The semantics of an expression $E\in\mathcal{L}(\prog{Exp})$
is the smallest relation
$\rightharpoonup : \mathcal{L}(\prog{Exp})
\times (\prog{Id} \to \Ss) \times \Os \times
\Ss$ closed under the following rules:

\small
\begin{minipage}{\linewidth}
\begin{minipage}[t]{0.4\linewidth}
\bigskip
\begin{prooftree}
\AxiomC{}
\UnaryInfC{$\langle \progE,\Sigma\rangle
\rightharpoonup \eepsilon$}
\end{prooftree}
\end{minipage}
\hfill
\begin{minipage}[t]{0.6\linewidth}
\begin{prooftree}
\AxiomC{$\langle e,\Sigma\rangle \rightharpoonup
\str$}
\UnaryInfC{$\langle e.\prog{b},\Sigma\rangle
\rightharpoonup \str \conc \bool$}
\end{prooftree}
\end{minipage}
\end{minipage}

\begin{minipage}{\linewidth}
\begin{minipage}[t]{0.4\linewidth}
\begin{prooftree}
\AxiomC{$\langle e,\Sigma\rangle
\rightharpoonup \str$}
\AxiomC{$\Sigma(\prog{Id})=\strT$}
\AxiomC{$\str\subseteq \strT$}
\TrinaryInfC{$\langle e\sqsubseteq \prog{Id},
\Sigma\rangle \rightharpoonup \one$}
\end{prooftree}
\end{minipage}
\hfill
\begin{minipage}[t]{0.6\linewidth}
\begin{prooftree}
\AxiomC{$\langle e,\Sigma\rangle
\rightharpoonup \str$}
\AxiomC{$\Sigma(\prog{Id})=\strT$}
\AxiomC{$\str \not\subseteq \strT$}
\TrinaryInfC{$\langle e\sqsubseteq \prog{Id},
\Sigma\rangle \rightharpoonup \zero$}
\end{prooftree}
\end{minipage}
\end{minipage}

\begin{minipage}{\linewidth}
\begin{minipage}[t]{0.4\linewidth}
\begin{prooftree}
\AxiomC{$\Sigma(\prog{Id})=\str$}
\UnaryInfC{$\langle \prog{Id},\Sigma\rangle
\rightharpoonup \str$}
\end{prooftree}
\end{minipage}
\hfill
\begin{minipage}[t]{0.6\linewidth}
\begin{prooftree}
\AxiomC{$\prog{Id} \not \in dom(\Sigma)$}
\UnaryInfC{$\langle \prog{Id},\Sigma\rangle
\rightharpoonup \eepsilon$}
\end{prooftree}
\end{minipage}
\end{minipage}

\begin{minipage}{\linewidth}
\begin{minipage}[t]{0.4\linewidth}
\begin{prooftree}
\AxiomC{$\langle e,\Sigma\rangle \rightharpoonup
\progZ$}
\UnaryInfC{$\langle \neg e,\Sigma\rangle
\rightharpoonup \one$}
\end{prooftree}
\end{minipage}
\hfill
\begin{minipage}[t]{0.6\linewidth}
\begin{prooftree}
\AxiomC{$\langle e,\Sigma\rangle \rightharpoonup\str$}
\AxiomC{$\sigma \neq\zero$}
\BinaryInfC{$\langle \neg e,\Sigma\rangle 
\rightharpoonup \zero$}
\end{prooftree}
\end{minipage}
\end{minipage}

\begin{minipage}{\linewidth}
\begin{minipage}[t]{0.4\linewidth}
\begin{prooftree}
\AxiomC{$\langle e,\Sigma\rangle \rightharpoonup
\one$}
\AxiomC{$\Sigma(\prog{Id})=\one$}
\BinaryInfC{$\langle e\wedge \prog{Id},\Sigma\rangle
\rightharpoonup \one$}
\end{prooftree}
\end{minipage}
\hfill
\begin{minipage}[t]{0.6\linewidth}
\begin{prooftree}
\AxiomC{$\langle e,\Sigma \rangle \rightharpoonup
\str$}
\AxiomC{$\Sigma(\prog{Id})=\strT$}
\AxiomC{$\str\neq \one \wedge
\strT \neq \one$}
\TrinaryInfC{$\langle e \wedge \prog{Id},\Sigma\rangle
\rightharpoonup \zero$}
\end{prooftree}
\end{minipage}
\end{minipage}
$$
$$
\normalsize
where $\prog{b}\in \{\prog{0},\prog{1}\}$.\footnote{We assume that
if $\prog{b}=\prog{1}$, then $\bool =\one$,
and if $\prog{b}=\prog{0}$, then $\bool=\zero$.}
\end{defn}

\begin{defn}[Big-Step Operational Semantics]
The semantics of a program $\prog{P} \in\LSra$ is the smallest
relation $\triangleright \subseteq \LSra \times
(\prog{Id} \to \Ss)
\times \Os \times (\prog{Id} \to \Ss)$
closed under the following rules:

\begin{minipage}{\linewidth}
\begin{minipage}[t]{0.4\linewidth}
\begin{prooftree}
\AxiomC{$\langle e,\Sigma\rangle \rightharpoonup
\str$}
\UnaryInfC{$\langle \prog{Id} \rightharpoonup
e,\Sigma,\eta\rangle \triangleright 
\Sigma[\prog{Id} \leftarrow \str]$}
\end{prooftree}
\end{minipage}
\hfill
\begin{minipage}[t]{0.6\linewidth}
\begin{prooftree}
\AxiomC{$\langle s,\Sigma,\eta\rangle 
\triangleright \Sigma'$}
\AxiomC{$\langle t,\Sigma',\eta\rangle
\triangleright \Sigma''$}
\BinaryInfC{$\langle s;t,\Sigma,\eta\rangle
\triangleright \Sigma''$}
\end{prooftree}
\end{minipage}
\end{minipage}

\begin{prooftree}
\AxiomC{$\langle e,\Sigma \rangle \rightharpoonup
\one$}
\AxiomC{$\langle s,\Sigma,\eta\rangle 
\triangleright \Sigma'$}
\AxiomC{$\langle \mathtt{while}(e)\{s\},
\Sigma',\eta\rangle \triangleright \Sigma''$}
\TrinaryInfC{$\langle \mathtt{while}(e)\{s\},
\Sigma, \eta\rangle \triangleright \Sigma''$}
\end{prooftree}

\begin{minipage}{\linewidth}
\begin{minipage}[t]{0.4\linewidth}
\begin{prooftree}
\AxiomC{$\langle e,\Sigma \rangle \rightharpoonup
\str $}
\AxiomC{$\str \neq \one$}
\BinaryInfC{$\langle \mathtt{while}(e)\{s\},
\Sigma,\eta\rangle \triangleright \Sigma$}
\end{prooftree}
\end{minipage}
\hfill
\begin{minipage}[t]{0.6\linewidth}
\begin{prooftree}
\AxiomC{$\langle e,\Sigma\rangle \rightharpoonup
\str$}
\AxiomC{$\eta(\str)=\bool$}
\BinaryInfC{$\langle \mathtt{Flip}(e),\Sigma,\eta\rangle
\triangleright \Sigma[R\leftarrow \bool]$}
\end{prooftree}
\end{minipage}
\end{minipage}
\end{defn}
\noindent
The semantics allows us to associate each 
program of $\SIFPra$ to the function
it evaluates:

\begin{defn}
Function evaluation
by a correct $\SIFPra$ program $\prog{P}$
is $\model{\cdot} : \LSra \to (\Ss^n \times \Os 
\to \Ss)$ defined as below:\footnote{Instead
of the infixed notation for $\triangleright$,
we use its prefixed notation.
So, the notation express the store associated
to the $P$,
$\Sigma$ and $\omega$ by $\triangleright$.
Moreover, we employ the same function symbol
$\triangleright$ to denote two distinct functions:
the big-step operational semantics of $\SIFPra$
and the big-step operational semantics
of programs in $\SIFPla$.}
$$
\model{\prog{P}} \df \lambda x_1,\dots, x_n,\eta
\triangleright (\langle \prog{P}, [ \ ] [X_1 \leftarrow x_1],
\dots, [X_n \leftarrow x_n],\eta\rangle )
(R).
$$
\end{defn}
Observe that, among the different registers,
the register $R$ is meant to contain
the value computed by the program at the
end of its execution.
Similarly, the $\{X_i\}_{i\in\Nat}$ registers
are used to store the inputs of the function.
We now establish the correspondence beteween
$\POR$ and $\SIFPra$.

\begin{lemma}\label{lemma:PORtoSIFPra}
For any function $f\in\POR$,
there is a poly-time $\SIFPra$ program
$\prog{P}$ such that
for all $x_1,\dots, x_n$,
$$
\model{\prog{P}}(x_1,\dots, x_n,\omega)=f(x_1,\dots,
x_n,\omega).
$$
Moreover, if $f\in\POR$, then
$\prog{P}$ does not contain any $\mathtt{Flip}(e)$
statement.
\end{lemma}

\begin{proof}[Proof Sketch]
The proof is quite simple under a technical
viewpoint,
relying on the fact that it is possible
to associate to any function of $\POR$
an equivalent poly-time program,
and on the possibility to compose them
and implement bounded recursion on notation
in $\SIFPra$ with a polynomial overhead.
Concretely, for any function $f\in\POR$,
we define a program $\mathscr{L}_f$ 
such that $\model{\mathscr{L}_f}(x_1,\dots, x_n)
= f(x_1,\dots, x_n)$.
The correctness of $\mathscr{L}_f$ is given 
by the following invariant properties:
\begin{itemize}
\itemsep0em

\item the result of the computation is stored in 
$R$

\item inputs are stored in the registers of the group
$X$

\item the function $\mathscr{L}$ does not 
change
the values it accesses as input.
\end{itemize}
We define $\mathscr{L}_f$
as follows:
\begin{align*}
\mathscr{L}_E &= R \leftarrow \epsilon \\
\mathscr{L}_{S_0} &= R \leftarrow X_0.\prog{0} \\
\mathscr{L}_{S_1} &= R \leftarrow X_0.\prog{1} \\
\mathscr{L}_{P^n_i} &= R \leftarrow X_i \\
\mathscr{L}_C &= R \leftarrow X_1 \sqsubseteq X_2 \\
\mathscr{L}{\query} &= \mathtt{Flip}(X_1).
\end{align*}
The correctness of base cases is trivial to prove
and the only translation containing 
$\mathtt{Flip}(e)$ for some 
$e\in \mathcal{L}(\prog{Exp})$ is that of $\query$.
The encoding of composition and bounded
recursion are more convoluted.\footnote{The 
proof of their correctness requires a conspicuous amount of
low-level definitions and technical results. 
For an extensive presentation, see~\cite{Davoli}.}
\end{proof}

\subsubsection{From $\SIFPra$ to $\SIFPla$}
We now show that every program in $\SIFPra$
is equivalent to one in $\SIFPla$.
First, we need to define the language and semantics
for $\SIFPla$.

\begin{defn}[Language of $\SIFPla$]
The language of $\SIFPla$ $\LSla$,
that is the set of strings produced by non-terminal
symbols $\prog{Stm}_{\text{LA}}$, is defined as follows:
$$
\prog{Stm}_{\text{LA}} \df \prog{Id} \leftarrow \prog{Exp} \midd
\prog{Stm}_{\text{LA}} ; \prog{Stm}_{\text{LA}} \midd
\mathtt{while}(\prog{Exp})\{\prog{Stm}\}_{\text{LA}} \midd
\mathtt{RandBit}().
$$
\end{defn}

\begin{defn}[Big-Step Semantics of $\SIFPla$]
The semantics of a program $\prog{P}\in\LSla$ is the smallest relation
$\triangleright \subseteq (\LSla \times (\prog{Id} \to
 \Ss \times \Bool^\Nat) \times
 ((\prog{Id} \to \Ss \times \Bool^\Nat)$
 closed under the following rules:

\begin{minipage}{\linewidth}
\begin{minipage}[t]{0.4\linewidth}
 \begin{prooftree}
 \AxiomC{$\langle e,\Sigma\rangle \rightharpoonup \str$}
 \UnaryInfC{$\langle \prog{Id} \leftarrow e,\Sigma,\omega\rangle
 \triangleright \langle \Sigma[\prog{Id}\leftarrow \str],\omega\rangle$}
 \end{prooftree}
 \end{minipage}
 \hfill
 \begin{minipage}[t]{0.6\linewidth}
\begin{prooftree}
\AxiomC{$\langle s,\Sigma,\omega\rangle \triangleright
\langle \Sigma',\omega'\rangle$}
\AxiomC{$\langle t,\Sigma',\omega\rangle \triangleright
\langle \Sigma'',\omega''\rangle$}
\BinaryInfC{$\langle s;t,\Sigma,\omega\rangle \triangleright
\langle \Sigma'',\omega''\rangle$}
\end{prooftree} 
\end{minipage}
\end{minipage}

\begin{prooftree}
\AxiomC{$\langle e,\Sigma\rangle \rightharpoonup
\one$}
\AxiomC{$\langle s,\Sigma,\omega\rangle \triangleright
\langle \Sigma',\omega'\rangle$}
\AxiomC{$\langle \mathtt{while}(e) \{s\}, \Sigma',\omega\rangle
\triangleright \langle \Sigma'',\omega''\rangle$}
\TrinaryInfC{$\langle \mathtt{while}(e)\{s\},
\Sigma,\omega\rangle \triangleright \langle \Sigma',\omega''\rangle$}
\end{prooftree}

\begin{minipage}{\linewidth}
\begin{minipage}[t]{0.4\linewidth}
\begin{prooftree}
\AxiomC{$\langle e,\Sigma\rangle \rightharpoonup \str$}
\AxiomC{$\str\neq \one$}
\BinaryInfC{$\langle \mathtt{while}(e)\{s\},\Sigma,\omega\rangle
\triangleright \langle \Sigma, \omega\rangle$}
\end{prooftree}
 \end{minipage}
 \hfill
 \begin{minipage}[t]{0.6\linewidth}
 \bigskip
\begin{prooftree}
\AxiomC{}
\UnaryInfC{$\langle \mathtt{RandBit}(), \Sigma,
\bool \omega\rangle \triangleright
\langle\Sigma[R\leftarrow \bool],\omega\rangle$}
\end{prooftree}
\end{minipage}
\end{minipage}
\end{defn}

\begin{lemma}\label{lemma:SIFPratoSIFPla}
For each total program $\prog{P}\in \SIFPra$,
there is a $\query\in \SIFPla$ such that,
for any $x,y\in \Nat$,
$$
\mu(\{\omega \in \Os \ | \ \model{\prog{P}}(x,\eta)=y\})
= \mu(\{\omega \in \Bool^\Nat \ | \ 
\model{\query}(x,\eta)=y\}).
$$
Moreover, if $\prog{P}$ is poly-time $\query$, too.
\end{lemma}

\begin{proof}[Proof Sketch]
We prove Lemma~\ref{lemma:SIFPratoSIFPla}
showing that $\SIFPra$ can be simulated in 
$\SIFPla$ given two novel \emph{small-step} semantic relations
($\toLA,\toRA$) obtained by splitting the corresponding
big-step semantics into smallest transitions,
one per each $\cdot \ ; \ \cdot$ instruction.
The intuitive idea behind this novel semantics
is to enrich big-step operational semantics with some piece
of information, which are needed to build 
an induction proof of the reduction from
$\SIFPra$ to $\SIFPla$.
In particular, we employ a list $\Psi$ containing
pairs $(x,\prog{b})$,
to keep track of the previous call to the primitive
$\Flip(x)$ forv $\SIFPra$,
and to the result of the $x$-th call of 
primitive $\mathtt{RandBit}()$ for $\SIFPla$.

The main issue 
consists in the simulation of the access to the
random tape.
Then, we define the translation so that
it stores in a specific and fresh register
an associative table recording all queries 
$\str$ within a $\Flip(\str)$ instruction and
the result \prog{b}
picked from $\eta$ and returned.
The addition of the map $\Psi$ allows 
to replicate the content of the associative table $\Psi$
explicitly in the semantics of program.
This simulation requires a translation of
$\Flip(e)$ into an equivalent procedure.
\begin{itemize}
\itemsep0em
\item at each simulated query $\Flip(e)$,
the destination program looks up the associative
table

\item if it finds the queried coordinate $e$ within
a pair $(e,\prog{b})$, it returns $\prog{b}$.
Otherwise, (i)
it reduced $\Flip(e)$ to a call of 
$\mathtt{RandBit}()$ which outputs either $\bool=\zero$
or $\bool=\one$, 
(ii) it records the couple $\langle e,\bool\rangle$
in the associative table and returns $\bool$.
\end{itemize}
The construction is convoluted.
This kind of simulation preserves the distributions
of strings computed by the original program:
\begin{enumerate}
\itemsep0em
\item we show that the big-step and small-step
semantics for $\SIFPra$ and $\SIFPla$ are
equally expressive.

\item we define a relation $\Theta$ between configurations of
the small-step semantics  of $\SIFPra$
and $\SIFPla$:
\begin{itemize}
\itemsep0em

\item[$*$] first, we define a function
$\alpha:\LSra \to \LSla$ mapping the program
$\prog{P}_{\text{RA}} \in \LSra$ into the corresponding
$\prog{P}_{\text{LA}}\in\LSla$, translating 
$\Flip(e)$ as described above.

\item[$*$] we define a relation $\beta$ between a triple
$\langle \prog{P}_{\text{RA}},\Sigma,\Psi\rangle$ and a single
store $\Gamma$,
which is meat to capture the configuration-to -store
dependencies between the configuration
of $\prog{P}_{\text{RA}}$ running with store $\Sigma$
and computed associative table $\Psi$
and the store $\Gamma$ of simulating $\prog{P}_{\text{LA}}$.
So, $\Gamma$ stores a representation of $\Psi$
into a dedicated register.

\item[$*$] a function $\gamma$ which transforms the constraints
on the oracle gathered by $\toRA$ to the information
collected by $\toLA$.
Then, $\Theta$ is so defined that $
\Theta(\langle \prog{P},\Sigma_1,\Psi\rangle, \langle \query,
\Sigma_2,\Theta\rangle)$ holds when:
(i.) $\alpha(\prog{P})=\query$, (ii.) $\beta(\langle 
\prog{P},\Sigma_1,\Psi\rangle,
\Sigma_2)$,
(iii.) $\gamma(\Psi)=\Phi$,
(iv.) $\mu(\Psi)=\mu(\Phi)$. 
\end{itemize}

\item we show that $\Theta$ associates to
each triple $\langle \prog{P}_{\text{RA}},\Sigma, \Psi\rangle$
other triples 
$\langle \prog{P}_{\text{LA}},\Gamma,\Phi\rangle$
which weakly simulate the relation $\toRA$ with respect
to $\toLA$.
This is displayed by Figure~\ref{fig:SIFPratoSIFPla}.
\end{enumerate}
  \begin{figure}[]
    \centering
    \begin{tikzpicture}[node distance=5cm]
        \node(PRA) at (-3,2) {$\langle P_{\text{RA}};Q_{\text{RA}}, \Sigma, \Psi\rangle$};
        \node(PLA) at (-3,-0) {$\langle P_{\text{LA}};Q_{\text{LA}}, \Gamma, \Phi\rangle$};
        \node(P1RA) at (3,2) {$\langle Q_{\text{RA}}, \Sigma', \Psi'\rangle$};
        \node(P1LA) at (3,0) {$^*\langle Q_{\text{LA}}, \Gamma', \Phi'\rangle$};
        \node at (-3.3,1) {$\Theta$};
        \node at (3.3,1) {$\Theta$};

        \draw[->,densely dashed,thick] (PRA) to (P1RA);
        \draw[->,densely dashed,thick] (PLA) to (P1LA);
        \draw[->] (PRA) to (PLA);
        \draw[->] (P1RA) to (P1LA);
    \end{tikzpicture}
    \caption{Commutation schema between $\SIFPra$ and $\SIFPla$}
    \label{fig:SIFPratoSIFPla}
  \end{figure}
\end{proof}

\subsubsection{From $\SIFPla$ to $\SFP_{\text{OD}}$}
 We introduce $\SFP_{\text{OD}}$,
  the class corresponding to $\SFP$ defined
  on a variation of STMs which can read characters from the oracle
  \emph{on-demand},
  and show that $\SIFPla$ can be reduced to it.
  For readability's sake, we avoid cumbersome details,
  focussing on an informal (but exhaustive) description
  of how to build the on-demand STM.
  %
  
  \begin{prop}\label{prop:SIFPlatoSFPod}
  For every $\prog{P}\in\LSla$, there is a $\STM^{\prog{P}} \in \SFP$
  such that for any $x\in\Ss$ and 
 $\omega \in\Bool^\Nat$,
 $$
 \prog{P}(x,\omega)=\STM^{\prog{P}}(x,\omega).
 $$
  Moreover, if $\prog{P}$ is poly-time, then $\STM^\prog{P}$ is also
  poly-time. 
  \end{prop}

  \begin{proof}[Proof Sketch]
 The construction relies on the implementation of a
 program in $\SIFPla$ by a multi-tape on-demand STM
 which uses a tape to store the values of each register,
 plus an additional tape containing the partial results
 obtained during the evaluation 
 of the expressions and another tape containing $\omega$.
 We denote the tape used for storing the
 result coming from the evaluation of the expression with $e$.
The following invariant properties holds:
\begin{itemize}
 \itemsep0em
 \item on each tape, values are stored to the immediate
 right of the head,
 
 \item the result of the last expression evaluated is stored
 on tape $e$ to the immediate right of the head.
\end{itemize} 
The value of expressions of $\SIFP$ can be computed
using the tape $e$.\footnote{In particular we prove this by induction
on the syntax of expressions. (1)
each access to the value stored in a register
consists in a copy of the content
of the corresponding tape to $e$, which is a simple operation
(due to the invariant properties above)
(2)
concatenations – namely $f.\prog{0}$ and
$f.\prog{1}$ – are implemented by the addition of a character
at the end of $e$, which contains the value of $f$.
(3) the binary expressions are non-trivial.
Since one of the two operands is a register identifier,
the machine can directly compare $e$ with the tape,
corresponding to the identifier, and to replace the content
of $e$ with the result of the comparison,
which in all cases is $\zero$ or $\one$.
Further details can be found in~\cite{Davoli}.}
All operations can be implemented without consuming any
character on the oracle tape
and with linear time with respect to the
size of the value of the expression. 
We assign a sequence of machine states,
$q^I_{s_1},q_{s_i}^1,\dots, q_{s_i}^F$, 
to each statement $s_i$.
\longv{
\begin{itemize}
\itemsep0em
\item assignments consist in copies of the value in $e$
to the tape corresponding to the destination register
and delations of the value on $e$, as replaced by $\blank$.
This is implemented without consuming any character on
the oracle tape.

\item The sequencing operation $s;t$ is implemented
inserting a composed transition from $q_s^F$ to $q_t^I$
in $\delta$.
This does not consume the oracle tape.

\item a $\mathtt{while}()$ statement $s=\mathtt{while}(f)\{t\}$
\end{itemize}}
In particular, the statement $\mathtt{RandBit}()$ is implemented consuming
a character on the tape and copying its value on the tape
which corresponds to the register $R$.

Furthermore, if we assume $\prog{P}$ to be poly-time,
after the simulation of each statement,
it holds that:
\begin{itemize}
\itemsep0em

\item the length of the non-blank portion of the first tapes
corresponding to the register is polynomially bounded
as their contents are precisely the the contents of
$\prog{P}$'s registers, which are polynomially bounded as a consequence
of the hypotheses of their poly-time complexity,

\item the head of all the tapes corresponding to the registers
point to the left-most symbol of the string thereby contained.
\end{itemize}
It is well-known that the reduction of the number of tapes
on a poly-time TM comes with a polynomial overhead in time.
For this reason, we conclude that the poly-time multi-tape
on-demand STM can be 
reduced to a poly-time canonical on-demand STM.
This concludes our proof.
  \end{proof}
 It remains to show that each on-demand STM can be reduced
 to an equivalent STM.

 \begin{lemma}\label{lemma:STACS32}
 For every $\STM =\langle \Qstates, \Sigma, \delta,q_0\rangle
 \in \SFP_{\text{OD}}$,
 the machine $\STM'=\langle \Qstates, \Sigma, H(\delta), q_0\rangle
 \in \SFP$ is such that for every $n\in\Nat$
 and configuration of $\STM$ $\langle \str,q,\strT,\eta\rangle$,
 and for every $\str',\strT'\in\Ss,q\in\Qstates$,
 $
 \mu(\{\omega \in \Bool^\Nat \ | \ (\exists \omega')\langle\str,q,\strT,\omega\rangle
 \triangleright^n_\delta \langle \str',q',\strT',\omega'\rangle\})$
= 
$\mu(\{\xi \in \Bool^\Nat \ | \ (\exists \xi')
\langle \str,q,\strT,\xi\rangle \triangleright^n_{H(\delta)}
\langle \str',q',\strT',\xi'\rangle \}).$
 \normalsize
 \end{lemma}

 \begin{proof}[Proof Sketch]
 Even in this case, the proof relies on a reduction.
 In particular, we show that given an on-demand STM,
 it is possible to build an STM, which is equivalent
 to the former.
 Intuitively, the encoding from an on-demand STM to an ordinary 
 STM takes the transition function $\delta$ of the STM
 and substitutes each transition not causing the oracle tape to shift
 – i.e. tagged with $\natural$ – with two distinct transitions,
 where $\natural$ is substituted by $\zero$ and $\one$, respectively.
 This causes the resulting machine to produce an identical transition but
 shifting the head on the oracle tape on the right.
We also define an encoding from an on-demand STM
to a canonical STM as follows:
$$
H \df \langle \Qstates, \Sigma, \delta, q_0\rangle \to 
\big\langle \Qstates, \Sigma, \bigcup \Delta_H(\delta), q_0\big\rangle.
$$
where $\Delta_H$ is,
\begin{align*}
\Delta_H\big(\langle p, c_r,\zero,q,c_w,d\rangle\big) &\df 
\{\langle p,c_r, \zero, q,c_w,d\rangle\} \\
\Delta_H\big(\langle p,c_r,\one,q,c_w,d\rangle\big) &\df
\{\langle p,c_r,\one,q,c_w,d\rangle\} \\
\Delta_H \big(\langle p,c_r,\natural, q,c_w,d\rangle\big) &\df
\{\langle p, c_r,\zero, q,c_w,d\rangle,
\langle p,c_r,\one,q,c_w,d\rangle\}.
\end{align*}
%
%

\end{proof}

\subsubsection{From $\POR$ to $\SFP$}
We then conclude the proof relating $\POR$ and $\SFP$.

\begin{prop}[From $\POR$ to $\SFP$]\label{prop:PORtoSFP}
For any $f:\Ss^k\times \Os \to \Ss \in\POR$,
there is a function $f^\star : \Ss^k\times \Bool^\Nat \to \Ss$
such that for all $\str_1,\dots, \str_k,\strT\in \Ss$,
$$
\mu\big(\{\omega \in \Bool^\Nat \ | \ f(\str_1,\dots, \str_k,\omega)=\strT\}\big) =
\mu\big(\{\eta \in \Os \ | \ f^\star(\str_1,\dots,\str_k,\eta)=\strT\}\big).
$$
\end{prop}

\begin{proof}
It is a straightforward consequence of Lemma~\ref{lemma:PORtoSIFPra}, Lemma~\ref{lemma:SIFPratoSIFPla}, Proposition~\ref{prop:SIFPlatoSFPod},
and Lemma~\ref{lemma:STACS32}.
\end{proof}
  
\end{document}